\documentclass[11pt]{article}
\pdfoutput=1
\usepackage{amsfonts,amsmath,amssymb,amsthm}
\usepackage{subcaption}
\usepackage[pdftex]{graphicx}
\usepackage[hmargin=1in,vmargin=1in]{geometry}
\usepackage[sort&compress]{natbib}				
\usepackage{booktabs}
\usepackage{appendix}
\usepackage{footnote}
\usepackage{multirow}
\usepackage[colorlinks,citecolor=blue]{hyperref}
\usepackage{enumerate}
\usepackage{caption}
\usepackage[finnish, british]{babel}
\usepackage{capt-of}
\usepackage{hyperref}
\usepackage{algorithm}
\usepackage{enumitem}
\usepackage{float}
\usepackage{mathabx}
\usepackage{algorithm,algpseudocode}
\usepackage[flushleft]{threeparttable}
\def\qed{\rule{2mm}{2mm}}
\def\indep{\perp \!\!\! \perp}

\let\footnote=\endnote

\usepackage{tikz,graphicx,pgfplots,pgfkeys,subcaption}   

\def\addlegendimage{\csname pgfplots@addlegendimage\endcsname}

\usepackage[pagewise,mathlines]{lineno}
\mathchardef\dash="2D
\tikzset{
  treenode/.style = {shape=rectangle, rounded corners,
                     draw, align=center},
  root/.style     = {circle, draw, top color = gray!30, bottom color=gray!30},
  leaf/.style = {treenode, font=\ttfamily\normalsize, top color = green!30, bottom color = green!30},
}
\newtheorem{theorem}{Theorem}[section]
\newtheorem{lemma}{Lemma}[section]
\newtheorem{definition}{Definition}[section]

\newtheorem{proposition}{Proposition}[section]
\theoremstyle{assumption}
\newtheorem{assumption}{Assumption}[section]
\theoremstyle{definition}

\newtheorem{example}{Example}[section]
\newtheorem{remark}{Remark}[section]
\usepackage{etoolbox} 
\AtEndEnvironment{remark}{~\qed}%

\begin{document}

\author{
Max Tabord-Meehan\\
Department of Economics\\
University of Chicago\\
\url{maxtm@uchicago.edu}
}

\title{Stratification Trees for Adaptive Randomization in Randomized Controlled Trials\\
}

\maketitle
\vspace{-10mm}
\begin{abstract}

This paper proposes an adaptive randomization procedure for two-stage randomized controlled trials. The method uses data from a first-wave experiment in order to determine how to stratify in a second wave of the experiment, where the objective is to minimize the variance of an estimator for the average treatment effect (ATE). We consider selection from a class of stratified randomization procedures which we call stratification trees: these are procedures whose strata can be represented as decision trees, with differing treatment assignment probabilities across strata. By using the first wave to estimate a stratification tree, we simultaneously select which covariates to use for stratification, how to stratify over these covariates, and the assignment probabilities within these strata. Our main result shows that using this randomization procedure with an appropriate estimator results in an asymptotic variance which is minimal in the class of stratification trees. Moreover, our results are able to accommodate a large class of assignment mechanisms within strata, including stratified block randomization. In a simulation study, we find that our method, paired with an appropriate cross-validation procedure, can improve on ad-hoc choices of stratification. We conclude by applying our method to the study in \cite{karlan2017}, where we estimate stratification trees using the first wave of their experiment. 

\end{abstract}

\clearpage

\section{Introduction} \label{sec:intro}
This paper proposes an adaptive randomization procedure for two-stage randomized controlled trials (RCTs). The method uses data from a first-wave experiment in order to determine how to stratify in a second wave of the experiment, where the objective is to minimize the variance of an estimator for the average treatment effect (ATE). We consider selection from a class of stratified randomization procedures which we call stratification trees: these are procedures whose strata can be represented as decision trees, with differing treatment assignment probabilities across strata.

Stratified randomization is ubiquitous in randomized experiments. In stratified randomization, the space of available covariates is partitioned into finitely many categories (i.e. strata), and randomization to treatment is performed independently across strata. Stratification has the ability to decrease the variance of estimators for the ATE through two parallel channels. The first channel is from ruling out treatment assignments which are potentially uninformative for estimating the ATE. For example, if we have information on the sex of individuals in our sample, and average outcomes vary with sex, then performing stratified randomization over this characteristic can reduce variance (we present an example of this for the standard difference-in-means estimator in Appendix \ref{sec:intro_example}). The second channel through which stratification can decrease variance is by allowing for differential treatment assignment probabilities across strata. For example, if we again consider the setting where we have information on sex, then it could be the case that for males the outcome under one treatment varies much more than under the other treatment. As we show in Section \ref{sec:discuss}, this can be exploited to reduce variance by assigning treatment according to the \emph{Neyman Allocation}, which in this example would assign more males to the more variable treatment. Our proposed method leverages insights from supervised machine-learning to exploit both of these channels, by simultaneously selecting \emph{which} covariates to use for stratification, \emph{how} to stratify over these covariates, as well as the optimal assignment probabilities within these strata, in order to minimize the variance of an estimator for the ATE.  

Our main result shows that using our procedure results in an ``optimal" (to be made precise later) stratification of the covariate space, where we restrict ourselves to stratification in a class of decision trees. A decision tree partitions the covariate space such that the resulting partition  can be interpreted through a series of yes or no questions (see Section \ref{sec:notation} for a formal definition and some examples). We focus on strata formed by decision trees for several reasons. First, since the resulting partition can be represented as a series of yes or no questions, it is easy to communicate and interpret, even with many covariates. This feature could be particularly important in many economic applications, because many RCTs in economics are undertaken in partnership with external organizations (for example, many RCTs described in \citealt{karlan2016} were undertaken in this way), and thus clear communication of the experimental design could be crucial.  Second, as we explain in Section \ref{sec:mainres}, using partitions based on decision trees gives us theoretical and computational tractability. Third, as we explain in Section \ref{sec:extensions}, using decision trees allows us to flexibly address the additional goal of minimizing the variance of estimators for subgroup-specific effects. Lastly, decision trees naturally encompass the type of stratifications usually implemented by practitioners. The use of decision trees in statistics and machine learning goes back at least to the work of Breiman \citep[see][for classic textbook treatments]{breiman1984, gyorfi1996}, and has seen a recent resurgence in econometrics \citep[examples include][]{athey2016, athey2017}.

An important feature of our theoretical results is that we allow for the possibility of so-called restricted randomization procedures \emph{within} strata. Restricted randomization procedures limit the set of potential treatment allocations, in order to force the true treatment assignment proportions to be close to the desired target proportions \citep[common examples used in a variety of fields include][]{efron1971, wei1978, antognini2004, kuznetsova2011, zelen1974}. Restricted randomization induces dependence in the assignments within strata, which complicates the analysis of our procedure. By extending techniques recently developed in \cite{bugni2017}, our results accommodate a large class of restricted randomization schemes, including stratified block randomization, which as we discuss in Example \ref{ex:sbr_treat} is a popular method of randomization.

Although our main focus is on increasing efficiency, stratified randomization has additional practical benefits beyond reducing the variance of ATE estimators. For example, when a researcher wants to analyze subgroup-specific effects, stratifying on these subgroups serves as a form of pre-analysis registration. To that end, we also present results on how to extend our procedure when targeting subgroup-specific effects.

The literature on design and inference in RCTs is vast (references in \citealt{athey2017survey}, \citealt{cox2000}, \citealt{glennerster2013}, \citealt{pukelsheim2006}, \citealt{rosenberger2015}, and from a Bayesian perspective, \citealt{ryan2016}, provide an overview). The classical literature on optimal randomization, going back to the work of \cite{smith1918} \citep[see][for a textbook treatment]{silvey2013}, maintains a parametric relationship for the outcomes with respect to the covariates, and targets efficient estimation of the model parameters. In contrast, our paper follows a recent literature which instead maintains a non-parametric model of potential outcomes, and targets efficient estimation of treatment effect parameters. This recent literature can be broadly divided into ``one-stage" procedures, which do not use prior data on all experimental outcomes to determine how to assign treatment \citep[examples include][]{kallus2013, kasy2016, barrios2014,aufenanger2017,quistorff2020}, and ``multi-stage" procedures, of which our method is an example.
Multi-stage procedures use prior data on the experimental outcomes to determine how to randomize. For example, they may use response information from previous experimental waves to determine how to randomize in subsequent waves of the experiment. We will call these procedures \emph{response-adaptive}. Although response adaptive methods typically require information from a prior experiment, such settings do arise in economic applications. First, many social experiments have a pilot phase or multi-stage structure. For example, \cite{simester2006}, \cite{karlan2008}, and \cite{karlan2017} all feature a multi-stage structure, and \cite{karlan2016} advocate the use of pilot experiments to help avoid potential implementation failures when scaling up to the main study. Second, many research areas have seen a profusion of related work which could be used as a first wave of data in a response-adaptive procedure \citep[see for example the discussion in the introduction of][]{hahn2011}. The study of response-adaptive methods to inform many aspects of experimental design, including how to randomize, has a long history in the  literature on clinical trials, both from a frequentist and Bayesian perspective \citep[see for example the references in][]{hu2006,sverdlov2015,cheng2003}, as well as in the literature on bandit problems \citep[see][]{bubeck2012}. 

Three papers which propose response-adaptive randomization methods in a framework similar to ours are \cite{hahn2011}, \cite{chambaz2014} and \cite{bai2019} \citep[see also][for related work in the presence of network interference]{viviano2020}. \cite{hahn2011} develop a procedure which uses the information from a first-wave experiment to compute the propensity score that minimizes the asymptotic variance of an ATE estimator, over a \emph{discrete} set of covariates (i.e. they stratify the covariate space ex-ante). They then use the resulting propensity score to assign treatment in a second-wave experiment. In contrast, our method computes the optimal assignment proportions over a data-driven discretization of the covariate space. \cite{chambaz2014} propose a multi-stage procedure which uses data from previous experimental waves to compute an optimal propensity score, where the propensity score is constrained through entropy restrictions. However, their method requires the selection of several tuning parameters, as well as additional regularity conditions, and their optimal target depends on these features in a way that may be hard to assess in practice. Their results are also derived in a framework where the number of experimental waves goes to infinity, which may not be an appropriate framework for many settings encountered in economics. Moreover, the results in both \cite{hahn2011} and \cite{chambaz2014} assume that assignment was performed completely independently across individuals in a given wave. In contrast, we reiterate that our results will accommodate a large class of stratified randomization schemes. \cite{bai2019} derives the MSE-optimal blocking of an experimental sample for the difference-in-means estimator, given a \emph{fixed} assignment proportion, and shows that this blocking takes the form of a ``matched-pairs" style design. He then proposes procedures which use information from a first-wave experiment to approximate the optimal blocking in a second-wave experiment. In contrast, in this paper we maintain an asymptotic framework where the number of strata is \emph{fixed}, which precludes the type of blocking designs considered in \cite{bai2019}. However, as explained in \cite{bai2019}, it is in fact possible to combine his procedure with the one proposed in this paper, by implementing his optimal blocking \emph{within} each stratum produced by our method. See Remark \ref{rem:depth} for further discussion.


The paper proceeds as follows: In Section \ref{sec:setup}, we provide a motivating discussion, set up the notation, and formally define the set of randomization procedures we consider. In Section \ref{sec:results}, we present the formal results underlying the method as well as several relevant extensions. In Section \ref{sec:simulations}, we perform a simulation study to assess the performance of our method in finite samples. In Section \ref{sec:application}, we consider an application to the study in \cite{karlan2017}, where we estimate stratification trees using the first wave of their experiment. Section \ref{sec:conclusion} concludes.

\section{Preliminaries} \label{sec:setup}
In this section we discuss some preliminary concepts and definitions. Section \ref{sec:discuss} presents a series of simplified examples which we use to motivate our procedure. Section \ref{sec:notation} establishes some notation and provides the definition of a \emph{stratification tree}, as well as our notion of a \emph{randomization procedure}. 

\subsection{Motivating Discussion}\label{sec:discuss}
We present a series of simplified examples which we use to motivate our proposed method. First we study the problem of optimal experimental assignment without covariates. We work in a standard potential outcomes framework: let $(Y(1),Y(0))$ be potential outcomes for a binary treatment $A \in \{0, 1\}$, and let the observed outcome $Y$ for an individual be defined as
\begin{equation}\label{eq:pot}
Y = Y(1)A + Y(0)(1-A)~.
\end{equation}
Let
$$\mu_a := E[Y(a)], \sigma^2_a := \text{Var}(Y(a))~,$$  
for $a \in \{0, 1\}$. Our quantity of interest is the average treatment effect
$$\theta := \mu_1 - \mu_0~.$$
Suppose we perform an experiment to obtain a size $n$ sample $\{(Y_i, A_i)\}_{i=1}^n$, where the sampling process is determined by $\{(Y_i(1),Y_i(0))\}_{i=1}^n$, which are i.i.d, and the treatment assignments $\{A_i\}_{i=1}^n$, where exactly $n_1 := \sum_{i=1}^nA_i = \lfloor{n\pi\rfloor}$ individuals are \emph{randomly} assigned to treatment $A = 1$, for some $\pi \in (0,1)$ (however, we emphasize that our results will accommodate other methods of randomization). Given this sample, consider estimation of $\theta$ through the standard difference-in-means estimator:
$$\hat{\theta}^{(1)} :=  \frac{1}{n_1}\sum_{i=1}^{n}Y_iA_i - \frac{1}{n - n_1}\sum_{i=1}^nY_i(1-A_i)~.$$
It can then be shown that
$$\sqrt{n}(\hat{\theta}^{(1)} - \theta)  \xrightarrow{d} N(\theta, V^{(1)})~,$$
where 
$$V^{(1)} := \frac{\sigma_1^2}{\pi} + \frac{\sigma_0^2}{1 - \pi}~.$$
Our goal is to choose $\pi$ to minimize the variance of $\hat{\theta}$. Solving this optimization problem yields the following solution:
$$\pi^* := \frac{\sigma_1}{\sigma_1 + \sigma_0}~.$$
This allocation, known as the \emph{Neyman Allocation}, assigns more individuals to the treatment which is more variable. Note that when $\sigma_0^2 = \sigma_1^2$, so that the variances of the potential outcomes are equal, the optimal proportion is $\pi^* = 0.5$, which corresponds to a standard equal treatment allocation. In general, implementing $\pi^*$ is infeasible without knowledge of $\sigma^2_0$ and $\sigma^2_1$. In light of this, if we had prior data $\{(Y_j, A_j)\}_{j=1}^m$ which allowed us to estimate $\sigma^2_0$ and $\sigma^2_1$, then we could use this data to estimate $\pi^*$, and then use this estimate to assign treatment in a subsequent wave of the study. The idea of sequentially updating estimates of unknown population quantities using past observations, in order to inform experimental design in subsequent stages, underlies many procedures developed in the literatures on response adaptive experiments and bandit problems, and is the main idea underpinning our proposed method.

\begin{remark}\label{rem:ethics}
Although the Neyman Allocation minimizes the variance of the difference-in-means estimator, it is entirely agnostic on the welfare of the individuals in the experiment itself. In particular, the Neyman Allocation could assign the majority of individuals in the experiment to the inferior treatment if that treatment has a much larger variance in outcomes (see \citealt{hu2006} for relevant literature in the context of clinical trials, as well as \cite{narita2018} for recent work on this issue in econometrics). While this feature of the Neyman Allocation may introduce ethical or logistical issues in some relevant applications, in this paper we focus exclusively on the problem of estimating the ATE as accurately as possible. 
\end{remark}

Next we repeat the above exercise with the addition of a discrete covariate $S \in \{1, 2, ... , K\}$ over which we stratify. We perform an experiment which produces a sample $\{(Y_i, A_i, S_i)\}_{i=1}^n$, where the sampling process is determined by i.i.d draws  $\{(Y_i(1), Y_i(0), S_i)\}_{i=1}^n$ and the treatment assignments $\{A_i\}_{i=1}^n$. For this example suppose that the $\{A_i\}_{i=1}^n$ are generated as follows: for each $k$, exactly $n_1(k) := \sum_{i=1}^n{\bf 1}\{S_i = k\}A_i = \lfloor{n(k)\pi(k)\rfloor}$ individuals are randomly assigned to treatment $A = 1$, with $n(k) := \sum_{i=1}^n{\bf 1}\{S_i = k\}$.

Note that when the assignment proportions $\pi(k)$ are not equal across strata, the difference-in-means estimator $\hat{\theta}^{(1)}$ is no longer consistent for $\theta$. Hence we consider the following weighted estimator of $\theta$:
$$\hat{\theta}^{(2)} := \sum_k\frac{n(k)}{n}\hat{\theta}(k)~,$$
where
$\hat{\theta}(k)$ is the difference-in-means estimator for $S = k$:
$$\hat{\theta}(k) := \frac{1}{n_1(k)}\sum_{i=1}^{n}Y_iA_i{\bf 1}\{S_i = k\} - \frac{1}{n(k) - n_1(k)}\sum_{i=1}^nY_i(1-A_i){\bf 1}\{S_i = k\}~.$$
In words, $\hat{\theta}^{(2)}$ is obtained by computing the difference in means for each $k$ and then taking a weighted average over each of these estimates. Note that when $K = 1$ (i.e. when $S$ can take only one value), this estimator simplifies to the difference-in-means estimator. It can be shown under appropriate conditions that
$$\sqrt{n}(\hat{\theta}^{(2)} - \theta) \xrightarrow{d} N(0, V^{(2)})~,$$
where
$$V^{(2)} := \sum_{k=1}^K P(S = k) \left[ \left(\frac{\sigma_0^2(k)}{1-\pi(k)} + \frac{\sigma_1^2(k)}{\pi(k)}\right) + \left(E[Y(1) - Y(0)|S = k] - E[Y(1)-Y(0)]\right)^2\right]~,$$
with
$\sigma^2_a(k) := E[Y(a)^2|S = k] - E[Y(a)|S = k]^2$. The first term in $V^{(2)}$ is the weighted average of the conditional variances of the difference in means estimator for each $S = k$. The second term in $V^{(2)}$ arises due to the additional variability in sample sizes for each $S = k$. We note that this variance takes the form of the semi-parametric efficiency bound derived by \cite{hahn1998} for estimators of the ATE which use the covariate $S$. Following a similar logic to what was proposed above without covariates, we could use first-wave data $\{(Y_j, A_j, S_j)\}_{j=1}^m$ to form a sample analog of $V^{(2)}$, and choose $\{\pi^*(k)\}_{k=1}^K$ to minimize this quantity. 

Now we introduce the setting that we consider in this paper: suppose we observe covariates $X \in \mathcal{X} \subset \mathbb{R}^d$, so that our covariate space is now multi-dimensional with potentially continuous components. How could we practically extend the logic of the previous examples to this setting? A natural solution is to \emph{discretize} (i.e. \emph{stratify}) $\mathcal{X}$ into $K$ categories (strata), by specifying a mapping $S:\mathcal{X} \rightarrow \{1, 2, 3, ..., K\}$, with $S_i := S(X_i)$, and then proceed as in the above example. As we argued in the introduction, stratified randomization is a popular technique in practice, and possesses several attractive theoretical and practical properties. In this paper we propose a method which uses first-wave data to estimate (1) the optimal stratification, and (2) the optimal assignment proportions within these strata. In other words, given first-wave data $\{(Y_j, A_j, X_j)\}_{j=1}^m$, where $X \in \mathcal{X} \subset \mathbb{R}^d$, we propose a method which selects $\{\pi(k)\}_{k=1}^K$ \emph{and} the function $S(\cdot)$, in order to minimize the variance of our estimator $\hat{\theta}^{(2)}$. In particular, our proposed solution selects a randomization procedure amongst the class of what we call \emph{stratification trees}, which we introduce in the next section.

\begin{remark}\label{rem:treatchoice}
Our focus on the minimization of asymptotic variance is motivated by standard asymptotic optimality results for regular estimators \citep[see for example Theorems 25.20, 25.21 in][]{van1998}, and our procedure directly impacts the size of confidence sets and the power of tests constructed using our normal approximation. However, accurate estimation of the ATE is not the only objective one could consider when designing an RCT. For example, we could instead consider designing the RCT with the ultimate goal of finding a treatment allocation which maximizes population welfare (see for example \citealt{manski2004}, \citealt{kitagawa2017} and references therein). Some recent work \citep[see for example][]{kasy2021} has focused on adaptive assignment mechanisms with this objective in mind. To what extent these objectives can be incorporated into our framework would be an interesting direction for future work. 
\end{remark}


\subsection{Notation and Definitions}\label{sec:notation}
In this section we establish our notation and define the class of randomization procedures that we will consider. As in Section \ref{sec:discuss} let $(Y(1), Y(0))$ be potential outcomes for a binary treatment $A \in \{0,1\}$ and let $X \in \mathcal{X} \subset \mathbb{R}^d$ denote a vector of observed pre-treatment covariates. Let $Q$ denote the distribution of $(Y(1), Y(0), X)$. Throughout the paper we assume that all of our observations are generated by i.i.d draws from $Q$. We restrict $Q$ as follows:
\begin{assumption}\label{ass:bounded}
Q satisfies the following properties:
\begin{itemize}[topsep = 1pt]
\item $Y(a) \in [-M, M]$ for some $M < \infty$, for $a \in \{0, 1\}$.
\item $X \in \mathcal{X} = \bigtimes_{j=1}^d [b_j, c_j]$, for some $\{b_j, c_j\}_{j=1}^d$ finite.
\item $X = (X_C, X_D)$, where $X_C \in \mathbb{R}^{d_1}$ for some $d_1 \in \{0, 1, 2, ..., d\}$ is continuously distributed with a bounded, strictly positive density. $X_D \in \mathbb{R}^{d - d_1}$ is discretely distributed with finite support.
\end{itemize}
\end{assumption}
\begin{remark}\label{rem:bound_X}
The assumptions imposed on $(Y(1),Y(0),X)$ in Assumption \ref{ass:bounded} are used frequently throughout the proofs of our results. However, they may be stronger than desirable in some applications. For example, the assumption that $X$ has only discrete or continuous components which are supported on a rectangle may fail in certain practical examples (see for example the set of covariates considered in Section \ref{sec:application}). However, in the simulations presented in Section \ref{sec:simulations} and Appendix \ref{sec:supp_sim} we consider designs where $Y(a)$ has unbounded support and $X$ is not supported on a rectangle, and these results suggest that Assumption \ref{ass:bounded} could be reasonably weakened.  We further note that although a user does not need to specify a choice of $M$ to implement our procedure, they are required to specify a choice of $\mathcal{X}$. We illustrate this point in the application presented in Section \ref{sec:application}.
 \end{remark}

Our quantity of interest is the average treatment effect (ATE) given by:
$$\theta := E[Y_i(1) - Y_i(0)]~.$$

An experiment on an i.i.d sample $\{(Y_i(1), Y_i(0), X_i)\}_{i=1}^n$ produces the following data:
$$\{W_i\}_{i=1}^n:= \{(Y_i, A_i, X_i)\}_{i=1}^n~,$$
whose joint distribution is determined by $Q$, the potential outcomes expression (\ref{eq:pot}), and the \emph{randomization procedure} which generates $\{A_i\}_{i=1}^n$. We focus on the class of stratified randomization procedures: these randomization procedures first stratify according to baseline covariates and then assign treatment status independently across each of these strata. However, we attempt to make minimal assumptions on the exact specification of the randomization procedure, and in particular we do \emph{not} require the treatment assignment within each stratum to be independent across observations.

We will now describe the structure we impose on the class of possible strata we consider. For $L$ a positive integer, let $K = 2^L$ and let $[K] := \{1, 2, ... ,K\}$. Consider a function $S: \mathcal{X} \rightarrow [K]$, then $\{S^{-1}(k)\}_{k=1}^K$ forms a partition of $\mathcal{X}$ with $K$ strata. For a given positive integer $L$, we work in the class $S(\cdot) \in \mathcal{S}_L$ of functions whose partitions form \emph{tree partitions} of depth $L$ on $\mathcal{X}$, which we now define. Our definition is recursive, so we begin with the definition for a tree partition of depth one:

\begin{definition}\label{def:tree_part1}
Let $\mathcal{X} = \bigtimes_{j=1}^d [b_j, c_j]$, and let $x = (x_1, x_2, ..., x_d) \in \mathcal{X}$. A tree partition of depth one on $\mathcal{X}$ is a partition $\{\mathcal{X}_D(j,\gamma), \mathcal{X}_U(j,\gamma)\}$ of $\mathcal{X}$, where 
$$\mathcal{X}_D(j,\gamma) := \{x \in \mathcal{X}: x_j \le \gamma\}~,$$
$$\mathcal{X}_U(j,\gamma) := \{x \in \mathcal{X}: x_j > \gamma\}~,$$
for some $j \in [d]$ and $\gamma \in (b_j, c_j)$. We call $\mathcal{X}_D(j,\gamma)$ and $\mathcal{X}_U(j,\gamma)$ leaves (or sometimes terminal nodes).
\end{definition}

\begin{example}\label{ex:tree_part1}
Figure \ref{fig:tree_part1} presents two different representations of a tree partition of depth one on $[0,1]^2$. The first representation we call \emph{graphical}: it depicts the partition on a square drawn in the plane. The second depiction we call a \emph{tree representation}: it illustrates how to describe a depth one tree partition as a yes or no question. In this case, the question is ``is $x_1$ less than or greater than 0.5?".
\end{example}

\begin{figure}[H]
\centering
\captionsetup{justification = centering}
\begin{tikzpicture}
\draw[transform canvas={yshift=5cm}];
\draw [fill = gray!20] (1,1) rectangle (3,5);
\draw [fill = gray!80] (3,1) rectangle (5,5);
\draw [very thick] (3,1) -- (3,5);
\node at (5.1,0.7) {$x_1$};
\node at (0.7,5.1) {$x_2$};
\node at (2,3) {$1$};
\node at (4,3) {$2$};
\node at (3,0.7) {\footnotesize{$0.5$}};
\end{tikzpicture}
\hspace{13mm}
\raisebox{15mm}{
\begin{tikzpicture}
  [
    grow                    = down,
    sibling distance        = 10em,
    level distance          = 3em,
    edge from parent/.style = {draw, -latex},
    every node/.style       = {font=\footnotesize},
    sloped
  ]
  \node [root]  {}
    child { node [leaf] {$1$}
      edge from parent node [above] {$x_1 \le 0.5$} }
    child { node [leaf] {$2$}
      edge from parent node [above] {$x_1 > 0.5$} };
 \end{tikzpicture}
 }
\caption{Two representations of a tree partition of depth 1 on $[0,1]^2$. \newline Graphical representation (left), tree representation (right).}\label{fig:tree_part1}
\end{figure}
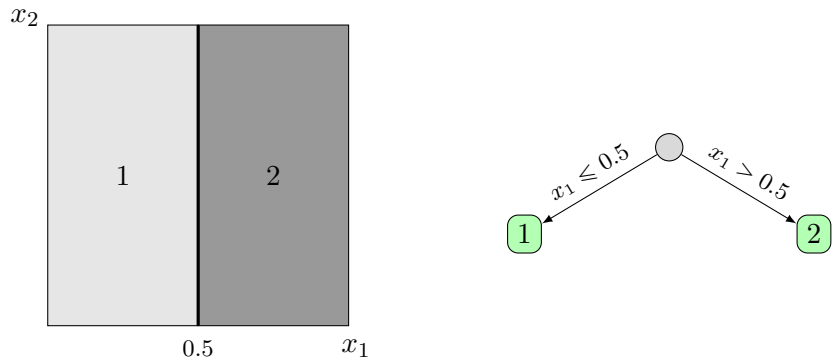

Next we define a tree partition of depth $L > 1$ recursively:

\begin{definition}\label{def:tree_partL}
A tree partition of depth $L > 1$ on $\mathcal{X} = \bigtimes_{j=1}^d [b_j, c_j]$ is a partition $\{\mathcal{X}^{(L-1)}_D(j, \gamma), \mathcal{X}^{(L-1)}_U(j, \gamma)\}$ of $\mathcal{X}$, where
$$\mathcal{X}^{(L-1)}_D(j, \gamma) \hspace{2mm} \text{is a tree partition of depth $L - 1$ on $\mathcal{X}_D(j, \gamma)$}~,$$
$$\mathcal{X}^{(L-1)}_U(j, \gamma) \hspace{2mm} \text{is a tree partition of depth $L - 1$ on $\mathcal{X}_U(j, \gamma)$}~,$$
for some $j \in [d]$ and $\gamma \in (b_j, c_j)$. We call $\mathcal{X}^{(L-1)}_D$ and $\mathcal{X}^{(L-1)}_U$ left and right subtrees, respectively.
\end{definition}

\begin{example}\label{ex:tree_part2}
Figure \ref{fig:tree_part2} depicts two representations of a tree partition of depth two on $[0,1]^2$. 
\end{example}

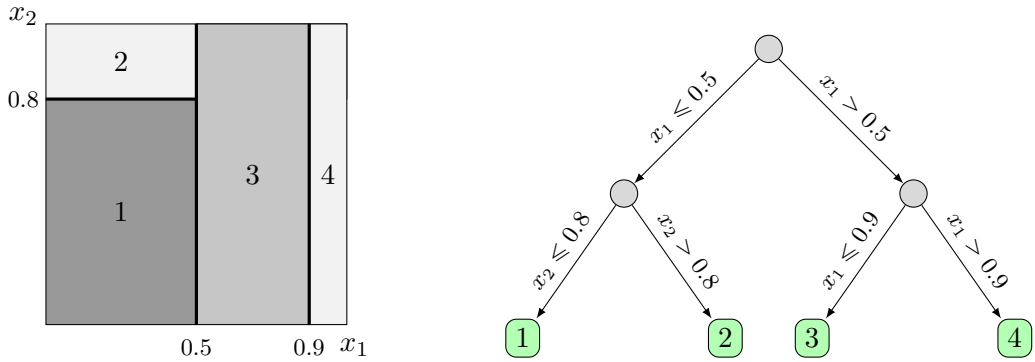
\begin{figure}[H]\center
\captionsetup{justification = centering}

      \begin{tikzpicture}
\draw [fill = gray!80] (1,1) rectangle (3,4);
\draw [fill = gray!10] (1,4) rectangle (5,5);
\draw [fill = gray!45] (3,1) rectangle (4.5,5);
\draw [fill = gray!10] (4.5,1) rectangle (5,5);
\draw [very thick] (3,1) -- (3,5);
\draw [very thick] (1,4) -- (3,4);
\draw [very thick] (4.5,1) -- (4.5,5);
\node at (5.1,0.7) {$x_1$};
\node at (0.7,5.1) {$x_2$};
\node at (2,2.5) {$1$};
\node at (2,4.5) {$2$};
\node at (3.75,3) {$3$};
\node at (4.75,3){$4$};
\node at (3,0.7) {\footnotesize{$0.5$}};
\node at (4.5,0.7){\footnotesize{$0.9$}};
\node at (0.7,4){\footnotesize{$0.8$}};
\end{tikzpicture}
\hspace{13mm}
\raisebox{1mm}{
\begin{tikzpicture}
  [
    grow                    = down,
    level distance          = 5em,
    edge from parent/.style = {draw, -latex},
    sloped,
    level 1/.style={sibling distance=10em,font=\footnotesize},
    level 2/.style={sibling distance=7em,font=\footnotesize},
  ]
  \node [root] {}
    child { node [root] {}
    	child {node [leaf]{$1$}
	  edge from parent node [above] {$x_2 \le 0.8$} }
	child {node [leaf]{$2$}
	  edge from parent node [above] {$x_2 > 0.8$} }
      edge from parent node [above] {$x_1 \le 0.5$} }
    child { node [root] {}
    	child {node[leaf]{$3$}
	  edge from parent node [above] {$x_1 \le 0.9$} }
	child {node[leaf]{$4$}
	  edge from parent node [above] {$x_1 > 0.9$} }
      edge from parent node [above] {$x_1 > 0.5$} };
      \end{tikzpicture}
}
\caption{Two representations of a tree partition of depth 2 on $[0,1]^2$. \newline Graphical representation (left), tree representation (right).}\label{fig:tree_part2}
\end{figure}

We focus on strata that form tree partitions for several reasons. First, relative to more ``flexible" partitions of $\mathcal{X}$, these types of strata are easy to represent and interpret, especially in higher dimensions, via their tree representations or as a series of yes or no questions. We argued in the introduction that this could be of particular importance in economic applications. Second, as we explain in Remarks \ref{rem:compdet} and \ref{rem:tech_challenge}, restricting ourselves to tree partitions helps with computational and theoretical tractability. In particular, computing an optimal stratification function is a difficult discrete optimization problem, but restricting ourselves to tree partitions allows us to employ an effective search heuristic known as an evolutionary algorithm. Third, the recursive aspect of tree partitions makes the targeting of subgroup-specific effects convenient, as we explain in Section \ref{sec:extensions}. 

For each $k \in [K]$, define $\pi := (\pi(k))_{k=1}^K$ to be the vector of target proportions of units assigned to treatment $1$ in each stratum. A \emph{stratification tree} is a pair $(S,\pi)$, where $S(\cdot)$ forms a tree partition, and $\pi$ specifies the target proportions in each stratum.  We denote the set of stratification trees of depth $L$ as $\mathcal{T}_L$.

\begin{remark}\label{rem:quotient}
To be precise, any element $T = (S, \pi) \in \mathcal{T}_L$ is equivalent to another element $T' = (S', \pi') \in \mathcal{T}_L$ whenever $T'$ can be realized as a re-labeling of $T$. For instance, if we consider Example \ref{ex:tree_part1} with the labels $1$ and $2$ reversed, the resulting tree is identical to the original except for this re-labeling. $\mathcal{T}_L$ should be understood as the quotient set that results from this equivalence.
\end{remark}
%

\begin{example}\label{ex:tree_part3}
Figure \ref{fig:strat_tree} depicts a representation of a stratification tree of depth two. Note that the terminal nodes of the tree have been replaced with labels that specify the target proportions in each stratum.
\end{example}

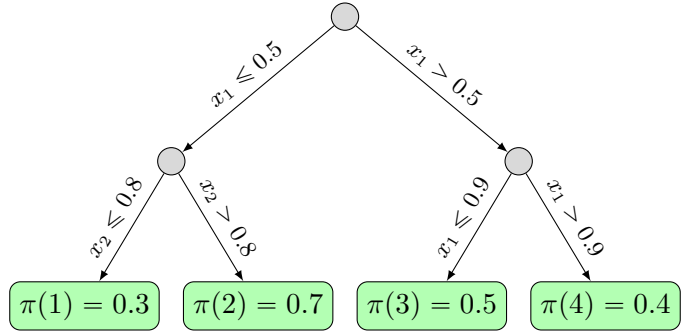
\begin{figure}[H]\center
\begin{tikzpicture}
  [
    grow                    = down,
    level distance          = 5em,
    edge from parent/.style = {draw, -latex},
    sloped,
    level 1/.style={sibling distance=12em,font=\footnotesize},
    level 2/.style={sibling distance=6em,font=\footnotesize},
  ]
  \node [root] {}
    child { node [root] {}
    	child {node [leaf]{$\pi(1) = 0.3$}
	  edge from parent node [above] {$x_2 \le 0.8$} }
	child {node [leaf]{$\pi(2) = 0.7$}
	  edge from parent node [above] {$x_2 > 0.8$} }
      edge from parent node [above] {$x_1 \le 0.5$} }
    child { node [root] {}
    	child {node[leaf]{$\pi(3) = 0.5$}
	  edge from parent node [above] {$x_1 \le 0.9$} }
	child {node[leaf]{$\pi(4) = 0.4$}
	  edge from parent node [above] {$x_1 > 0.9$} }
      edge from parent node [above] {$x_1 > 0.5$} };
      \end{tikzpicture}
\caption{Representation of a Stratification Tree of Depth 2}\label{fig:strat_tree}
\end{figure}

Note that in our definition, a stratification tree of depth $L$ has exactly $K = 2^L$  leaves, so that our trees are ``perfectly balanced". For the optimization problem we consider in our paper this is without loss of generality, because it can be shown \citep[for example using the argument from Theorem 6.1 in][]{bugni2021} that the asymptotic variance which results from a given stratification tree can be made weakly smaller by taking any leaf and subdividing it arbitrarily into two leaves with the same assignment target as their parent.

We further impose that the set of trees cannot have arbitrarily small cells, nor can they have arbitrarily extreme treatment assignment targets:

\begin{assumption}\label{ass:prop/cell_size}
We constrain the set of stratification trees $T = (S, \pi) \in \mathcal{T}_L$ such that, for some fixed $\nu > 0$ and $\delta > 0$, $\pi(k) \in [\nu, 1- \nu]$ and $P(S(X)=k) > \delta$.
\end{assumption}
We impose the restriction in Assumption \ref{ass:prop/cell_size} to ensure, for example, that $1/\pi(k)$ is uniformly bounded for all $T \in \mathcal{T}_L$. For technical reasons relating to the potential non-measurability of our estimator, we will impose one additional restriction on $\mathcal{T}_L$. 
\begin{assumption}\label{ass:cell_grid}
Let  $\mathcal{T}_L^{\dagger} \subset \mathcal{T}_L$ be a countable, closed subset of the set of stratification trees\footnotemark[1].  We then consider the set of stratification trees restricted to this subset. By an abuse of notation, we continue to denote the set of stratification trees we will consider as $\mathcal{T}_L$.
\end{assumption}

\begin{remark}
We emphasize that this assumption is \emph{only} used as a sufficient condition to guarantee measurability, in order to invoke Fubini's theorem. Note that, in practice, restricting the set of stratification trees to those constructed from a finite grid satisfies Assumption \ref{ass:cell_grid}. However, our results will also apply more generally.
\end{remark}

\footnotetext[1]{Here ``closed" is with respect to an appropriate topology on $\mathcal{T}_L$, see Appendix \ref{sec:rho} for details.}

For each $T \in \mathcal{T}_L$, and given an i.i.d sample $\{(Y_i(0),Y_i(1), X_i)\}_{i=1}^n$ of size $n$, an experimental assignment is described by a random vector $(A_i(T))_{i=1}^n$. For our purposes a \emph{randomization procedure} (or randomization scheme) is a family of such random vectors indexed by $T = (S, \pi) \in \mathcal{T}_L$.  For $T = (S, \pi)$, let $S_i:= S(X_i)$, $(S_i)_{i=1}^n$ be the random vector of stratification labels of the observed data. We impose two assumptions on the randomization procedure $(A_i(T))_{i=1}^n$.

First, we require the following exogeneity assumption:

\begin{assumption}\label{ass:indep_treat}
The randomization procedure is such that, for each $T = (S, \pi) \in \mathcal{T}_L$,
$$\left[\{(Y_i(0), Y_i(1), X_i)\}_{i=1}^n \indep (A_i(T))_{i=1}^n \right]\bigg| (S_i)_{i=1}^n~.$$
\end{assumption}
This assumption asserts that the randomization procedure can depend on the observables only through the strata labels. Next, let $p(k;T) := P(S_i = k)$ be the population proportions of each stratum, then we also require that  the randomization procedure satisfy the following ``consistency" property:

\begin{assumption}\label{ass:unif_treat}
The randomization procedure is such that
$$\sup_{T \in \mathcal{T}_L} \left|\frac{n_1(k;T)}{n} - \pi(k)p(k;T)\right| \xrightarrow{p} 0~,$$
for each $k \in [K]$, where
$$n_1(k;T) := \sum_{i=1}^n{\bf 1}\{A_i(T) = 1, S_i = k\}~.$$
\end{assumption}
This assumption asserts that the assignment procedure must approach the target proportion asymptotically, and do so in a uniform sense over all stratification trees in $\mathcal{T}_L$. 

Other than Assumptions \ref{ass:indep_treat} and \ref{ass:unif_treat}, we do not require any additional assumptions about how assignment is performed. Examples \ref{ex:iid_treat} and \ref{ex:sbr_treat} illustrate two randomization schemes which satisfy these assumptions and are popular in economics. \cite{bugni2017} make similar assumptions for a \emph{fixed} stratification and show that they are satisfied for a wide range of assignment procedures, including procedures often considered in the literature on clinical trials: see for example \cite{efron1971}, \cite{wei1978}, \cite{antognini2004}, and \cite{kuznetsova2011}. In Proposition \ref{prop:SBR_unif} below, we verify that Assumptions \ref{ass:indep_treat} and \ref{ass:unif_treat} hold for stratified block randomization (see Example \ref{ex:sbr_treat}), which is a common assignment procedure in economic applications.

 \begin{example}\label{ex:iid_treat}
\emph{Simple random assignment} assigns each individual within stratum $k$ to treatment via a coin-flip with weight $\pi(k)$. Formally, for each $T$, $(A_i(T))_{i=1}^n$ is a vector with independent components such that
$$P(A_i(T) = 1|S_i = k) = \pi(k)~.$$
Simple random assignment is theoretically convenient, and features prominently in papers on adaptive randomization. However, it is considered unattractive in practice because it results in a ``noisy" assignment for a given target $\pi(k)$, and hence could be far off the target assignment for any given random draw. Moreover, this extra noise increases the finite-sample variance of ATE estimators relative to other assignment procedures which target $\pi(k)$ more directly \citep[see for example the discussion in][]{kasy2013}.
\end{example}

\begin{example}\label{ex:sbr_treat}
\emph{Stratified block randomization} (SBR) assigns a fixed proportion $\pi(k)$ of individuals within stratum $k$ to treatment $1$. Formally, let $n(k)$ be the number of units in stratum $k$, and let $n_1(k)$ be the number of units assigned to treatment 1 in stratum $k$. In SBR, $n_1(k)$ is given by
$$n_1(k) = \lfloor n(k) \pi(k) \rfloor~.$$
SBR proceeds by randomly assigning $n_1(k)$ units to treatment $1$ for each $k$, where all 
$${n(k) \choose n_1(k)}~,$$
possible assignments are equally likely. This assignment procedure has the attractive feature that it targets the proportion $\pi(k)$ as directly as possible. An early discussion of SBR can be found in \cite{zelen1974}. 
\end{example}

We conclude this section by showing that Assumptions \ref{ass:indep_treat} and \ref{ass:unif_treat} are satisfied by SBR:

\begin{proposition}\label{prop:SBR_unif}
Suppose randomization is performed through SBR (see Example \ref{ex:sbr_treat}), then Assumptions \ref{ass:indep_treat} and \ref{ass:unif_treat} are satisfied.
\end{proposition}

\section{Results} \label{sec:results}
In this section we formally define our proposed procedure and present results about its asymptotic behavior. Section \ref{sec:mainres} sets up the problem and presents the main results about the asymptotic normality of our estimator. Section \ref{sec:extensions} considers several extensions: a cross-validation procedure to select the depth $L$ of the stratification tree, asymptotic results for a ``pooled" estimator of the ATE, and extensions for the targeting of subgroup specific effects.


\subsection{Main Results}\label{sec:mainres}
In this section we describe our procedure and present our main formal results. Recall our discussion at the end of Section \ref{sec:discuss}: given first-wave data, our goal is to estimate a stratification tree which minimizes the asymptotic variance in a certain class of ATE estimators, which we now introduce. For a fixed $T \in \mathcal{T}_L$, let $\{(Y_i, A_i, X_i)\}_{i=1}^n$ be an experimental sample generated from a randomized experiment with randomization procedure $(A_i(T))_{i=1}^n$. Consider estimation of the following equation by OLS:
$$Y_i = \sum_k \alpha(k ;T) {\bf 1}\{S_i = k\} + \sum_k \beta(k ;T) {\bf 1}\{A_i = 1, S_i = k\} + u_i~.$$
Then our ATE estimator is given by
$$\hat{\theta}(T) := \sum_k\frac{n(k; T)}{n}\hat{\beta}(k; T)~,$$
where $n(k; T) := \sum_i {\bf 1}\{S_i = k\}$.
In words, this estimator takes the difference in means between treatments within each stratum, and then averages these over the strata.
Given appropriate regularity conditions, the results in \cite{bugni2017} establish asymptotic normality for a \emph{fixed} $T = (S, \pi) \in \mathcal{T}_L$:
$$\sqrt{n}(\hat{\theta}(T) - \theta) \xrightarrow{d} N(0,V(T))~,$$
where 
$$V(T) := \sum_{k=1}^K P(S(X) = k) \left[\left(E[Y(1) - Y(0)|S(X)=k] - E[Y(1)-Y(0)]\right)^2 + \left(\frac{\sigma_0^2(k)}{1-\pi(k)} + \frac{\sigma_1^2(k)}{\pi(k)}\right)\right]~,$$
and 
$$\sigma_a^2(k) := E[Y(a)^2|S(X)=k] - E[Y(a)|S(X)=k]^2~.$$

Again we remark that this variance takes the form of the semi-parametric efficiency bound of \cite{hahn1998} amongst all estimators that use the strata indicators as covariates. We propose a two-stage adaptive randomization procedure which asymptotically achieves the minimal variance $V(T)$ across all $T \in \mathcal{T}_L$. For the rest of the paper, we denote the sample size of the first wave by $m$ and the sample size of the second wave by $n$. We index first-wave observations by $j = 1, \ldots, m$ and second-wave observations by $i = 1, \ldots, n$. In the first stage, we use first-wave data $\{(Y_j,A_j,X_j)\}_{j=1}^m$ to estimate some ``optimal" tree $\hat{T}_m$ which is designed to minimize $V(T)$. 
In the second stage, we perform a randomized experiment using stratified randomization with $(A_i(\hat{T}_m))_{i=1}^n$ to obtain second-wave data $\{(Y_i,A_i,X_i)\}_{i=1}^n$. Finally, to analyze the results of the experiment, we consider both the ``unpooled" estimator $\hat{\theta}(\hat{T}_m)$ defined above, which uses only the second-wave data to estimate the ATE, as well as a ``pooling" estimation strategy, which use both waves of data to construct an ATE estimator (see Section \ref{sec:extensions}).

\begin{remark}\label{rem:depth}
The depth $L$ of the set of stratification trees will remain fixed but arbitrary throughout Section \ref{sec:mainres}. The primary reason for this is technical: in order to allow for a wide variety of possible assignment procedures we leverage and extend the results in \cite{bugni2017}, but these results are derived in an asymptotic framework where the number of strata is fixed. Considering extensions of their results to settings where the number of strata grows is beyond the scope of this paper. However, we provide three arguments for why considering a fixed-$L$ asymptotic framework may not be a major limitation in our setting: (1) accurately estimating stratification trees with many strata is both computationally difficult and may result in poor finite-sample performance unless the first-wave sample size is unrealistically large (we return to the question of how to choose $L$ with this consideration in mind in Section \ref{sec:extensions}), (2) in Appendix \ref{sec:supp_sim} we provide some preliminary simulation evidence which suggests that there are decreasing returns to increasing $L$, and (3) a simple compromise for practitioners who wish to stratify more finely is to perform ad-hoc stratification \emph{within} each of the strata (with accompanying assignment proportions) produced by our method. By arguing as in Theorem 6.1 in \cite{bugni2021} it can be shown that this is guaranteed to weakly decrease the asymptotic variance of our estimator, and moreover, as long as the ad-hoc stratification is not \emph{too} fine, Proposition \ref{prop:theta_robust} establishes that our proposed inference procedure will still be valid. If a practitioner wishes to stratify \emph{very} finely, then one possibility is to perform the optimal blocking procedure derived in \cite{bai2019} \emph{within} each of the strata produced by our method. However in this case our proposed inference procedure may no longer be appropriate. See \cite{bai2019} for details.
\end{remark}

We now present the main theoretical properties of our method. First, we establish conditions under which the estimator $\hat{\theta}(\hat{T}_m)$ constructed using the second wave of data is asymptotically normal, with minimal variance in the class of estimators defined above. Additionally, we provide a consistent estimator of the asymptotic variance of our estimator, and establish a form of ``robustness" of our estimator to potential inconsistency of $\hat{T}_m$.  

%
%

From now on, to be concise, we will call data from the first-wave the \emph{pilot} data, and data from the second-wave the \emph{main} data. As in the paragraph above, denote the pilot data as $\{W_j\}_{j=1}^m := \{(Y_j,X_j,A_j)\}_{j=1}^m$. Given this pilot sample, we require the following high-level consistency property for our estimator $\hat{T}_m$:

\begin{assumption}\label{ass:tree_estimate}
The estimator $\hat{T}_m$ is a $\sigma\{(W_j)_{j=1}^m\}/\mathcal{B}(\mathcal{T}_L)$ measurable function of the pilot data\footnotemark[2]  and satisfies
$$|V(\hat{T}_m) - V^*| \xrightarrow{p} 0~,$$
where 
$$V^* := \inf_{T \in \mathcal{T}_L} V(T)~,$$
as $m \rightarrow \infty$.
\end{assumption}

\footnotetext[2]{$\mathcal{B}(\mathcal{T}_L)$ is the Borel-sigma algebra on $\mathcal{T}_L$ generated by an appropriate topology and $\sigma\{(W_i)_{i=1}^m\}$ is the sigma-algebra generated by the pilot data. See Appendix \ref{sec:rho} for details.}

Note that Assumption \ref{ass:tree_estimate} does not require that $V^*$ is \emph{uniquely} minimized at some $T \in \mathcal{T}_L$. Moreover, Assumption \ref{ass:tree_estimate} imposes no explicit restrictions on how $\hat{T}_m$ is constructed, or even on the nature of the pilot data itself. In Appendix \ref{sec:rho}, Lemma \ref{lem:S_const} we show that Assumption \ref{ass:tree_estimate} is sufficient to guarantee that the sequence of trees $\hat{T}_m$ ``approaches" the set of minimizers of $V(\cdot)$, but this does not guarantee that this sequence converges to any \emph{fixed} tree within that set. Similar results have been derived in a maximum-likelihood context in \cite{redner1981}.

In Proposition \ref{prop:emp_const} below, we establish sufficient conditions on the pilot data under which an appropriate $\hat{T}_m$ can be constructed by solving the following empirical minimization problem:
$$\hat{T}_m^{EM} \in \arg\min_{T \in \mathcal{T}_L} \widetilde{V}_m(T)~,$$
where
\[\widetilde{V}_m(T) := \sum_{k=1}^K\frac{m(k;T)}{m}\left[\left(\hat{E}[Y(1) - Y(0)|S(X) = k] - \hat{E}[Y(1) - Y(0)]\right)^2 + \left(\frac{\hat{\sigma}^2_{0,S}(k)}{1 - \pi(k)} + \frac{\hat{\sigma}^2_{1,S}(k)}{\pi(k)}\right)\right]\]
with
\begin{align*}
\hat{\sigma}^2_{a,S}(k) &:= \hat{E}[Y(a)^2|S(X)=k] - \hat{E}[Y(a)|S(X)=k]^2~, \\
\hat{E}[Y(a)^t|S(X) = k] &:= \frac{1}{m_a(k;T)}\sum_{j=1}^mY_j^t{\bf 1}\{A_j = a\}{\bf 1}\{S(X_j)=k\} \text{ for $t \in \{1, 2\}$}~,\\
\hat{E}[Y(a)] &:= \frac{1}{m_a}\sum_{j=1}^mY_j{\bf 1}\{A_j = a\}~.
\end{align*}

In general, computing $\hat{T}_m^{EM}$ involves solving a nonlinear discrete optimization problem. Although this problem does not have a unique solution, Proposition \ref{prop:emp_const} establishes conditions such that \emph{any} (measurable) sequence of minimizers which results from solving this minimization problem will satisfy Assumption \ref{ass:tree_estimate}, due to the fact that the empirical objective $\widetilde{V}_m(\cdot)$ approximates $V(\cdot)$ uniformly well as $m \rightarrow \infty$.

\begin{remark}\label{rem:compdet}
In Appendix \ref{sec:compdet} we describe an evolutionary algorithm which effectively performs a stochastic search for the global minimizer of the empirical minimization problem, and provide rough guidelines for implementation. We make a few comments here about the effectiveness of this algorithm in practice: first, the algorithm finds the global minimum in simple verified examples. Second, the algorithm always returns the \emph{same} tree in repeated runs of the algorithm (up to negligible perturbations), with appropriate tuning parameters. Third, our current implementation of the algorithm (implemented in Julia 1.6) terminates fairly quickly for moderate depths and sample sizes: typically in less than one hour on a personal computer.
\end{remark}

In Proposition \ref{prop:emp_const} we verify Assumption \ref{ass:tree_estimate} for $\hat{T}_m^{EM}$ when the pilot data comes from a stratified RCT with equal assignment proportions across strata:

\begin{proposition}\label{prop:emp_const}
Suppose the pilot data come from a RCT performed using stratified randomization with stratification function $\zeta: \mathcal{X} \rightarrow \{1, 2, ..., Z\}$ (where $Z < \infty$ is a fixed number) such that $P(\zeta(X) = z) > 0$ for all $z = 1, \ldots, Z$. Let $m(z; \zeta) := \sum_{j = 1}^m {\bf 1}\{\zeta(X) = z\}$, $m_a(z; \zeta) := \sum_{j = 1}^m{\bf 1}\{A_j = a, \zeta(X) = z\}$. Suppose that the randomization procedure for the pilot data satisfies:
\[\left[\{(Y_j(0), Y_j(1), X_j)\}_{j=1}^m \indep (A_j)_{j =1}^m\right] \Bigg| (\zeta_j)_{j =1}^m ~,\]
\[\frac{m_1(z; \zeta)}{m(z;\zeta)} \xrightarrow{p} \pi \hspace{1mm} \text{for some $\pi \in (0, 1)$, for all $z = 1, \ldots Z$}~.\]
Under Assumptions \ref{ass:bounded}, \ref{ass:prop/cell_size}, and \ref{ass:cell_grid}, Assumption \ref{ass:tree_estimate} is satisfied for $\hat{T}_m^{EM}$.
\end{proposition}

We now state the first main result of the paper: an optimality result for the estimator $\hat{\theta}(\hat{T}_m)$. In Remark \ref{rem:tech_challenge} we comment on some of the technical challenges that arise in the proof of the result.

\begin{theorem}\label{thm:main}
Given Assumptions \ref{ass:bounded}, \ref{ass:prop/cell_size}, \ref{ass:cell_grid}, \ref{ass:indep_treat}, \ref{ass:unif_treat}, and \ref{ass:tree_estimate}, we have that
$$\sqrt{n}(\hat{\theta}(\hat{T}_m) - \theta) \xrightarrow{d} N(0,V^*)~,$$
as $m, n \rightarrow \infty$.
\end{theorem}

\begin{remark}\label{rem:tech_challenge}
Here we comment on some of the technical challenges that arise in proving Theorem \ref{thm:main}. First, we develop a theory of convergence for stratification trees by defining a novel metric on $\mathcal{S}_L$ based on the Frechet-Nikodym metric, and establish basic properties about the resulting metric space. In particular, we use this construction to show that a set of minimizers of $V(T)$ exists given our assumptions, and that $\hat{T}_m$ converges to this set of minimizers in an appropriate sense. For these results we frequently exploit the fact that for a fixed index $k \in [K]$, the class of sets $\{S^{(-1)}(k): S \in \mathcal{S}_L\}$ consists of rectangles, and hence forms a VC class.

Next, because Assumptions \ref{ass:indep_treat} and \ref{ass:unif_treat} impose so little on the dependence structure of the randomization procedure, it is not clear how to apply standard central limit theorems. When the stratification is fixed, \cite{bugni2017} establish asymptotic normality by essentially re-writing the sampling distribution of the estimator as a partial-sum process. In our setting the stratification is \emph{random}, and so to prove our result we generalize their construction in a way that allows us to re-write the sampling distribution of the estimator as a \emph{sequential empirical process} \citep[see][Section 2.12.1 for a definition]{van1996}. We then exploit the asymptotic equicontinuity of this process to establish asymptotic normality (see Lemma \ref{lem:main_2}). 
\end{remark}

Next we construct a consistent estimator for the variance $V^*$. Let
$$\widehat{V}_H(T) := \sum_{k=1}^K \frac{n(k;T)}{n}\left(\hat{\beta}(k;T) - \hat{\theta}(T)\right)^2~,$$
and let 
$$\widehat{V}_Y(T) := \hat{R}'(T)\hat{V}_{hc}(T)\hat{R}(T)~,$$
where $\hat{V}_{hc}(T)$ is the robust variance estimator for the parameters in the saturated regression, and $\hat{R}(T)$ is following vector with $K$ ``leading" zeros:
$$\left(\hat{R}(T)\right)' := \left[0, 0, 0, \ldots, 0, \frac{n(1; T)}{n}, \ldots, \frac{n(K; T)}{n}\right]~.$$

We obtain the following consistency result:
\begin{theorem}\label{thm:var_const}
Given Assumptions \ref{ass:bounded}, \ref{ass:prop/cell_size}, \ref{ass:cell_grid}, \ref{ass:indep_treat}, \ref{ass:unif_treat}, and \ref{ass:tree_estimate}, then
$$\widehat{V}(\hat{T}_m) \xrightarrow{p} V^*~,$$
where
$$\widehat{V}(T) := \widehat{V}_H(T) + \widehat{V}_Y(T)~,$$
as $m, n \rightarrow \infty$.
\end{theorem}

Although Theorem \ref{thm:main} guarantees that $\hat{\theta}(\hat{T}_m)$ is asymptotically normal as $m, n \rightarrow \infty$, we may be concerned about the validity of this approximation when conducting inference in settings where the pilot sample size is not large. Accordingly, we finish this section by presenting a result about the asymptotic validity of hypothesis tests constructed from  $\hat{\theta}(\hat{T}_m)$ and $\widehat{V}(\hat{T}_m)$ when $\hat{T}_m$ is not necessarily itself consistent in the sense of Assumption \ref{ass:tree_estimate}. Consider the problem of testing
\begin{equation}\label{eq:null}
H_0 : \theta = \theta_0 \hspace{2mm} \text{versus} \hspace{2mm} H_1:\theta \ne \theta_0~,
\end{equation}
at level $\alpha \in (0,1)$, using a standard test given by
\[\phi_n(T) := {\bf 1}\{\left|W(T)\right| > z_{1 - \frac{\alpha}{2}}\}~,\]
where 
\[W(T) := \frac{\sqrt{n}(\hat{\theta}(T) - \theta_0)}{\sqrt{\hat{V}(T)}}~,\]
and $z_{1-\frac{\alpha}{2}}$ is the $1 - \frac{\alpha}{2}$ quantile of a standard normal random variable.

\begin{proposition}\label{prop:theta_robust}
Let $\widetilde{T}_m$ be any sequence of trees constructed from the pilot data. Suppose that 
\[V(T) > 0 \hspace{2mm} \text{for all} \hspace{2mm} T \in \mathcal{T}_L~.\] Given Assumptions \ref{ass:bounded}, \ref{ass:prop/cell_size}, \ref{ass:cell_grid}, \ref{ass:indep_treat}, and \ref{ass:unif_treat}, under the null hypothesis given by (\ref{eq:null}),
\[\lim_{m,n \rightarrow \infty} E[\phi_n(\widetilde{T}_m)] = \alpha~.\]
\end{proposition}

Note that Proposition \ref{prop:theta_robust} also accommodates the case where the pilot sample size is fixed at some number $m'$ by simply defining the sequence $\widetilde{T}_m$ to be equal to $\widetilde{T}_{m'}$ for $m \ge m'$. We conclude from Proposition \ref{prop:theta_robust} that, regardless of whether or not $\hat{T}_m$ is consistent for an optimal tree, we can use $W(\hat{T}_m)$ and the critical values from a standard normal distribution to conduct asymptotically valid inference. Indeed, we will see in the simulations of Section \ref{sec:simulations} that even in situations where $\hat{T}_m$ is a very poor estimate of an optimal tree, the coverage of a confidence interval or size of a test constructed using $\phi_n(\hat{T}_m)$ are close to the nominal level.


\subsection{Extensions}\label{sec:extensions}

In this section we present some extensions to the main results. First we present a version of $\hat{T}_m$ whose depth is selected by cross-validation. Second, we describe a method to combine estimates of the ATE from both waves of data, and establish properties of the resulting ``pooled" estimator. Finally, we explain how to accommodate the targeting of subgroup-specific effects. 

\subsubsection{Cross-validation to select $L$}
In this subsection we present a method to help select the depth $L$ in practice. To choose $L$, we revisit the first-stage estimation problem via the lens of the general model-selection paradigm described in, for example, \cite{arlot2010}. The tradeoff which arises when choosing between various choices of $L$ in the first-stage estimation problem can be framed as a classical tradeoff between \emph{approximation} error and  \emph{estimation} error, as we now describe. For each $L$, let 
\[V^*_L := \min_{T \in \mathcal{T}_L}V(T)~.\]
On one hand, using a larger $L$ allows us to attain a (weakly) lower value for the asymptotic variance $V^*_L$. On the other hand, using a larger $L$ makes the set of trees $\mathcal{T}_L$ more complex, and thus makes estimation of the optimal tree more difficult in a finite sample, since we run the risk of ``overfitting". More formally, let $\bar{L}$ be some upper bound on the depth of trees to be considered (in practice, this would correspond to some computational limit, or potentially to some exogenous logistical constraint), and let $[\bar{L}] = \{0, 1, 2, \ldots, \bar{L}\}$, where we understand $L = 0$ to mean no stratification. Let $\hat{T}_m^{(L)}$ be a stratification tree of depth $L$ estimated from the pilot data. We can then decompose the excess variance obtained by using $\hat{T}_m^{(L)}$ relative to $V^*_{\bar{L}}$ as
\[V(\hat{T}_m^{(L)}) - V^*_{\bar{L}} = \left(V^*_{L} - V^*_{\bar{L}}\right) +  \left(V(\hat{T}_m^{(L)})- V^*_{L}\right)~.\]
The first term on the right-hand-side of this expression can be understood as the approximation error that results from optimizing in the class of trees $\mathcal{T}_L$ instead of the class of trees $\mathcal{T}_{\bar{L}}$. The second term on the right-hand-side of the expression can be understood as the estimation error that results from using the estimated tree $\hat{T}_m^{(L)}$ instead of the optimal tree for the class $\mathcal{T}_L$. Using this decomposition we see immediately that by Assumption \ref{ass:tree_estimate} we are guaranteed that setting $L = \bar{L}$ achieves the smallest possible excess variance asymptotically, however, our goal here is to attempt to balance these two tradeoffs in \emph{finite samples}. In particular, the ideal ``oracle depth" $L^*_m$ for a pilot sample of size $m$ is given by 
\[L^*_m := \arg\min_{L \in [\bar{L}]} \left(V(\hat{T}_m^{(L)}) - V^*_{\bar{L}}\right)~,\]
which exactly balances the tradeoff between the estimation and approximation errors. Of course, this is infeasible, and so instead we propose selecting a depth $\hat{L}_m$  via a standard ``$B$-fold" cross validation procedure which is designed to balance the estimation and approximation tradeoffs described above.

For simplicity we describe $2$-fold cross validation, but we comment on other choices of $B$ in Remark \ref{rem:vfold} below. The cross-validation procedure proceeds as follows. First, split the pilot sample randomly into two halves and denote these by $\mathcal{D}_1$ and $\mathcal{D}_2$. For each $L$, let $\hat{T}^{(L,1)}_m$ and $\hat{T}^{(L,2)}_m$ be stratification trees of depth $L$ estimated on $\mathcal{D}_1$ and $\mathcal{D}_2$, respectively. Let $\widetilde{V}^{(1)}_m(\cdot)$ and $\widetilde{V}^{(2)}_m(\cdot)$ be the empirical variances computed on $\mathcal{D}_1$ and $\mathcal{D}_2$ (where, in the event that a cell in the tree partition is empty, we assign a value of infinity to the empirical variance). Then we define the following cross-validation criterion:
\[\widetilde{V}^{CV}_L := \frac{1}{2}\left(\widetilde{V}^{(1)}_m\left(\hat{T}^{(L,2)}_m \right) + \widetilde{V}^{(2)}_m\left(\hat{T}^{(L,1)}_m\right)\right)~.\]
In words, for each $L$, we estimate a stratification tree on each half of the sample, compute the empirical variance of these estimates by using the \emph{other} half of the sample, and then average the results. Intuitively, as we move from small values of $L$ to large values of $L$, we would expect that this cross-validation criterion should generally decrease with $L$, and then eventually increase, in accordance with the tradeoff between estimation error and approximation error.   Given $\widetilde{V}^{CV}_L$, we select our depth $\hat{L}_m$ as follows: $$\hat{L}_m := \arg\min_{L\in [\bar{L}]} \widetilde{V}^{CV}_L~,$$
where in the event of a tie we choose the smallest such $L$. Accordingly, our cross-validated stratification tree is defined as
\[\hat{T}_m^{CV} := \hat{T}^{(\hat{L}_m)}_m~,\]
i.e. $\hat{T}_m^{CV}$ is chosen to be the stratification tree whose depth minimizes the cross-validation criterion $\widetilde{V}^{CV}_L$. 


We assess the finite-sample performance of $\hat{T}_m^{CV}$ via simulation in Section \ref{sec:simulations}, and note there that trees constructed using the cross-validation procedure can outperform trees constructed using $\bar{L}$ when the pilot is small, and perform similarly when the pilot is large (we prove a formal result about the large pilot behavior of the cross-validation procedure in Appendix \ref{sec:CV_large}). As a result we recommend that researchers fix some maximum allowable depth $\bar{L}$ (again, this will frequently correspond to a computational limit), and then use cross-validation to select an appropriate depth between $0$ and $\bar{L}$. In Section \ref{sec:application}, we use this cross-validation procedure to select the depth of the stratification trees we estimate for the experiment undertaken in \cite{karlan2017}.

\begin{remark}\label{rem:vfold}
Our description of cross-validation above defines $B$-fold cross-validation for $B = 2$. It is straightforward to extend this to general $B$, where the dataset is split into $B$ folds. In many statistical applications $5$ or $10$ folds has become the practical standard. However, we illustrate in Appendix \ref{sec:supp_sim} that, when the first-wave sample size is small, using more than $2$ folds can potentially \emph{reduce} performance. With that in mind, increasing the number of folds can be beneficial when the first-wave sample size is not too small; in particular, using more folds results in more stability, in the sense that $\hat{L}$ fluctuates less across different splits of the data. \end{remark}

\subsubsection{A pooling estimator of the ATE}\label{sec:pool}
In this subsection we study an estimator which allows us to ``pool" data from both datasets when estimating the ATE. Pooling may be particularly useful in formal two-stage randomized experiments where the first wave sample-size is large relative to the total sample-size (for example, in the application we consider in Section \ref{sec:application}). 

Let $\hat{\theta}_1$ be an estimator of the ATE constructed from the pilot data, and let $\hat{\theta}(\hat{T}_m)$ be the estimator defined in Section \ref{sec:mainres}. 
We impose the following high level assumption on the asymptotic behavior of $\hat{\theta}_1$:

\begin{assumption}\label{ass:pilot_normal}
$\hat{\theta}_1$ is an asymptotically normal estimator for the ATE:
\[\sqrt{m}(\hat{\theta}_1 - \theta) \xrightarrow{d} N(0, V_1)~,\]
as $m \rightarrow \infty$.
\end{assumption}

Assumption \ref{ass:pilot_normal} holds for a variety of standard estimators under various assignment schemes: see for example the results in \cite{bugni2015}, \cite{bugni2017}, and \cite{bai2018}. We also impose the following assumption on the relative rates of growth of the pilot and main sample. 

\begin{assumption}\label{ass:m/N}
Let $m$ be the pilot data sample size, $n$ the main data sample size, and $N = m + n$. We assume that
\[\frac{m}{N} \rightarrow \lambda~,\]
for some $\lambda \in [0, 1]$.
\end{assumption}

We propose the following sample-size weighted estimator:
\[\hat{\theta}_{MW} := \widehat{\lambda}\hat{\theta}_1 + (1 - \widehat{\lambda})\hat{\theta}(\hat{T}_m)~,\]
where $\hat{\lambda} := m/N$. Theorem \ref{thm:pool_est} derives the limiting distribution of this estimator:

\begin{theorem}\label{thm:pool_est}
Given Assumptions \ref{ass:bounded}, \ref{ass:prop/cell_size}, \ref{ass:cell_grid}, \ref{ass:indep_treat}, \ref{ass:unif_treat}, \ref{ass:tree_estimate}, \ref{ass:pilot_normal}, and \ref{ass:m/N}, we have that
\[\sqrt{N}(\hat{\theta}_{MW} - \theta) \xrightarrow{d} N(0,  V^*_\lambda)~,\]
where $N := n+m$ and $V^*_\lambda := \lambda V_1 + (1 - \lambda)V^*$, as $m, n \rightarrow \infty$.
\end{theorem}

In words, we see that the pooled estimator $\hat{\theta}_{MW}$ has an asymptotic variance which is a weighted combination of the optimal variance and the variance from estimation in the pilot experiment, with weights which correspond to their relative sizes. In the asymptotic regime where $\lambda = 0$, $V^*_\lambda = V^*$, and hence pooling has no impact on the asymptotic behavior of the estimator. In contrast, in an asymptotic regime where $\lambda \ne 0$, $\hat{\theta}_{MW}$ and $\hat{\theta}(\hat{T}_m)$ are computed on sample sizes which differ asymptotically. To compare their variances, note that 
\[\text{Var}\left(\hat{\theta}_{MW}\right) \approx \frac{V^*_\lambda}{N}~,\]
whereas
\[\text{Var}\left(\hat{\theta}(\hat{T}_m)\right) \approx \frac{V^*}{n}~,\]
so that, asymptotically, pooling will be beneficial when $(1 - \lambda)V^*_\lambda < V^*$. In practice, we expect that this will often be the case when the pilot data make up a large proportion of the total sample size (as in for example the application in Section \ref{sec:application}).


\begin{remark}\label{rem:pool}
The pooled estimator we present in this section is myopic, in the sense that $\hat{T}_m$ and $\hat{\theta}_{MW}$ are estimated as if the researcher did not anticipate that they would pool the data in the second stage. This has the benefit of being straightforward to analyze under very general assumptions on the pilot experiment. However, we could also consider  ``smart" versions of pooling, where the researcher estimates $\hat{T}_m$ taking into account that pooling will occur in the second stage. We present a preliminary discussion of such a strategy in Appendix \ref{sec:smart_pool}, but due to the increased technical complications of this approach we do not pursue a formal analysis of the procedure in this paper. 
\end{remark}

\subsubsection{Stratification Trees for Subgroup Targeting}
In this subsection we explain how the method can flexibly accommodate the problem of variance reduction for estimators of subgroup-specific ATEs, while still minimizing the variance of the unconditional ATE estimator in a restricted set of trees. It is common practice in RCTs for the strata to be specified such that they are the subgroups that a researcher is interested in studying \citep[see for example the recommendations in][]{glennerster2013}. This serves two purposes: the first is that it enforces a pre-specification of the subgroups of interest, which guards against ex-post data mining. Second, it allows the researcher to improve the efficiency of the subgroup specific estimates. 

Let $S' \in \mathcal{S}_{L'}$ be a tree of depth $L' < L$, whose terminal nodes represent the subgroups of interest. Suppose these nodes are labelled by $g = 1, 2, ..., G$, and that $P(S'(X) = g) > 0$ for each $g$. The subgroup-specific ATEs are defined as follows:
$$\theta^{(g)} := E[Y(1) - Y(0)|S'(X) = g]~.$$
We introduce the following new notation: let $\mathcal{T}_L(S') \subset \mathcal{T}_L$ be the set of stratification trees of depth $L$ which can be constructed as \emph{extensions} of $S'$. For a given $T \in \mathcal{T}_L(S')$, let $\mathcal{K}_g(T) \subset [K]$ be the set of terminal nodes of $T$ which pass through the node $g$ in $S'$ (see Figure \ref{fig:subtree} for an example). 
\vspace{5mm}
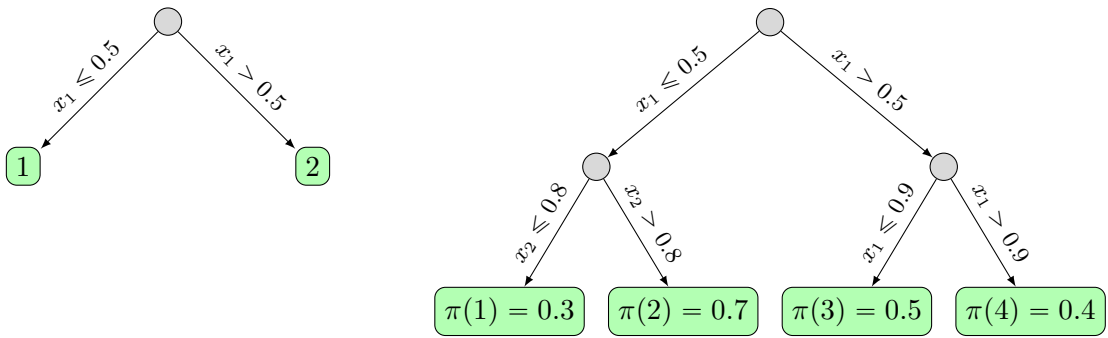
\begin{figure}[H]
\centering
\captionsetup{justification = centering}
\raisebox{20 mm}{
\begin{tikzpicture}
  [
    grow                    = down,
    sibling distance        = 10em,
    level distance          = 5em,
    edge from parent/.style = {draw, -latex},
    every node/.style       = {font=\footnotesize},
    sloped
  ]
  \node [root]  {}
    child { node [leaf] {$1$}
      edge from parent node [above] {$x_1 \le 0.5$} }
    child { node [leaf] {$2$}
      edge from parent node [above] {$x_1 > 0.5$} };
 \end{tikzpicture}
}
\hspace{10mm}
\begin{tikzpicture}
  [
    grow                    = down,
    level distance          = 5em,
    edge from parent/.style = {draw, -latex},
    sloped,
    level 1/.style={sibling distance=12em,font=\footnotesize},
    level 2/.style={sibling distance=6em,font=\footnotesize},
  ]
  \node [root] {}
    child { node [root] {}
    	child {node [leaf]{$\pi(1) = 0.3$}
	  edge from parent node [above] {$x_2 \le 0.8$} }
	child {node [leaf]{$\pi(2) = 0.7$}
	  edge from parent node [above] {$x_2 > 0.8$} }
      edge from parent node [above] {$x_1 \le 0.5$} }
    child { node [root] {}
    	child {node[leaf]{$\pi(3) = 0.5$}
	  edge from parent node [above] {$x_1 \le 0.9$} }
	child {node[leaf]{$\pi(4) = 0.4$}
	  edge from parent node [above] {$x_1 > 0.9$} }
      edge from parent node [above] {$x_1 > 0.5$} };
      \end{tikzpicture}

\caption{On the left: a tree $S'$ whose nodes represent the subgroups of interest. \newline
On the right: an extension $T \in \mathcal{T}_2(S')$. Here $\mathcal{K}_1(T) = \{1,2\}, \mathcal{K}_2(T) = \{3,4\}$}\label{fig:subtree}
\end{figure}

Given a tree $T \in \mathcal{T}_L(S')$, a natural estimator of $\theta^{(g)}$ is then given by
$$\hat{\theta}^{(g)}(T) := \sum_{k \in \mathcal{K}_g}\frac{n(k; T)}{n'(g)}\hat{\beta}(k; T)~,$$
where $n'(g) = \sum_{i=1}^n {\bf 1}\{S'(X_i) = g\}$ and $\hat{\beta}(k)$ are the regression coefficients of the saturated regression over $T$. It is straightforward to see from the recursive structure of stratification trees that choosing $T$ as a solution to the following problem:
$$\min_{T \in \mathcal{T}_L(S')} V(T)~,$$
will minimize the asymptotic variance of the subgroup specific estimators $\hat{\theta}^{(g)}$, while still minimizing the variance of the global ATE estimator $\hat{\theta}$ in the restricted set of trees $\mathcal{T}_L(S')$. Moreover, to compute a minimizer of $V(T)$ over $\mathcal{T}_L(S')$, it suffices to compute the optimal tree for each subgroup, and then append these to $S'$ to form the stratification tree. 

In Section \ref{sec:application} we illustrate the application of this idea to the setting in \cite{karlan2017}. In their paper, they study the effect of information about a charity's effectiveness on subsequent donations to the charity, and in particular the treatment effect heterogeneity between large and small prior donors. For their application we specify $S'$ to be a tree of depth $1$, whose terminal nodes correspond to the subgroups of large and small prior donors. We then estimate the optimal tree for each of these subgroups and append them to $S'$ to form a stratification tree which simultaneously minimizes the variance of the subgroup-specific estimators, while still minimizing the variance of the global estimator in this restricted class.

\section{Simulations}\label{sec:simulations}
In this section we analyze the finite sample behaviour of our method via a simulation study, and in particular analyze the performance of the cross-validation procedure presented in Section \ref{sec:extensions}. We consider three DGPs in the spirit of the designs considered in \cite{athey2016}. We emphasize that although these designs are artificial, they highlight several interesting qualitative patterns. For all three designs in this section, the outcomes are specified as follows:
$$Y_i(a) = \kappa_a(X_i) + \nu_a(X_i)\cdot\epsilon_{a,i}~.$$
Where the $\epsilon_{a,i}$ are i.i.d $N(0, 0.1)$, and $\kappa_a(\cdot)$, $\nu_a(\cdot)$ are specified individually for each DGP below. In all cases, $X_i \in [0,1]^d$, with components independently and identically distributed as $Beta(2,5)$.
The specifications are given by:

\noindent {\bf Model 1}: $d = 2$, $\kappa_0(x) = 0.2$, $\nu_0(x) = 5$, 
$$\kappa_1(x) = 10x_1^2{\bf 1}\{x_1 > 0.4\} - 5x_2^2{\bf 1}\{x_2 > 0.4\}~,$$
$$\nu_1(x)  = 1 + 10x_1^2{\bf 1}\{x_1 > 0.6\} + 5x_2^2{\bf 1}\{x_2 > 0.6\}~.$$
This is a ``low-dimensional" design with two covariates. The first covariate is given a higher weight than the second in the outcome equation for $Y(1)$.

\noindent {\bf Model 2}: $d = 10$, $\kappa_0(x) = 0.5$, $\nu_0(x) = 5$, 
$$\kappa_1(x) = \sum_{j = 1}^{10} (-1)^{j-1}10^{-j+2}x_j^2{\bf 1}\{x_j > 0.4\}~,$$
$$\nu_1(x) =  1 + \sum_{j = 1}^{10} 10^{-j+2}x_j^2{\bf 1}\{x_j > 0.6\}~.$$
This is a ``moderate-dimensional" design with ten covariates. Here the first covariate has the largest weight in the outcome equation for $Y(1)$, and the weight of subsequent covariates decreases quickly. 

\noindent {\bf Model 3}: $d = 10$, $\kappa_0(x) = 0.2$, $\nu_0(x) = 9$, 
$$\kappa_1(x) = \sum_{j = 1}^{3} (-1)^{j-1}10x_j^2\cdot{\bf 1}\{x_j > 0.4\} +  \sum_{j = 4}^{10}(-1)^{j-1}5x_j^2\cdot{\bf 1}\{x_j > 0.4\}~,$$
$$\nu_1(x) =  1 + \sum_{j = 1}^{3}10x_j^2\cdot{\bf 1}\{x_j > 0.6\} +  \sum_{j = 4}^{10}5x_j^2\cdot{\bf 1}\{x_j > 0.6\}~.$$
This is a ``moderate-dimensional" design with ten covariates. Here the first three covariates have similar weight in the outcome equation for $Y(1)$, and the next seven covariates have a smaller but still significant weight.

In each case, $\kappa_0(\cdot)$ is calibrated so that the average treatment effect is close to $0.1$, and $\nu_0(\cdot)$ is calibrated so that $Y_i(1)$ and $Y_i(0)$ have similar unconditional variances (see Appendix \ref{sec:compdet} for details). In each simulation we test six different methods of stratification. In all cases, when we stratify we consider a maximum of $8$ strata (which corresponds to a stratification tree of depth 3). In all cases we use SBR to perform assignment. We consider the following methods of stratification:

\begin{itemize}[topsep = 1pt]
\item No Stratification: Here we assign the treatment to half the sample, with no stratification.
\item Ad-hoc: Here we stratify in an ``ad-hoc" fashion and then assign treatment to half the sample in each stratum. To construct the strata we iteratively select a covariate and stratum at random, and stratify on the midpoints of the currently defined stratum. 
\item Ad-hoc $+$ Neyman: Here we split the sample and perform a pilot experiment to estimate the Neyman allocation for each stratum defined in Ad-hoc. We then use the resulting stratification to assign treatment in the second wave.  
\item Stratification Tree: Here we split the sample and perform a pilot experiment to estimate a stratification tree. We then use this tree to assign treatment in the second wave.
\item Cross-Validated Tree: Here we estimate a stratification tree as above, while selecting the depth via 2-fold cross validation. If the depth of the resulting tree is less than 3, then we perform ad-hoc stratification within each of the computed leaves as described in Remark \ref{rem:depth}.
\item Infeasible Optimal Tree: Here we estimate an ``optimal" tree by using a large auxiliary sample (see Appendix \ref{sec:compdet} for details). We then split the sample and perform a pilot experiment in the first wave, while using the optimal tree to assign treatment in the second wave.
\end{itemize}

We perform the simulations with a sample size of $5,000$, and consider three different splits of the total sample for the pilot experiment and main experiment. The pilot experiment was performed using SBR with ad-hoc stratification. To estimate the stratification trees we minimize an empirical analog of the asymptotic variance as described in Section \ref{sec:mainres}. The estimator of the ATE we use throughout is the pooled estimator described in Section \ref{sec:extensions}.

We assess the performance of the randomization procedures through the following criteria: the empirical coverage of a $95\%$ confidence interval formed using a normal approximation, the percentage reduction in average length of the $95\%$ CI relative to no stratification, the power of a $t$-test for an ATE of 0, and the percentage reduction in root mean-squared error (RMSE) relative to no stratification. For each design we perform $6,000$ Monte Carlo iterations. Tables \ref{tab:model1}, \ref{tab:model2}, and \ref{tab:model3} below present our simulation results.

\begin{remark}\label{rem:supp_sim}
In Appendix \ref{sec:supp_sim} we present additional simulation results. In particular, we explore alternative choices of $B$ in our $B$-fold cross validation procedure and alternative choices for the maximum number of strata. 
\end{remark}

\begin{table}[p]
  \centering
    \begin{tabular}{ccccccc}
       \toprule
    \multicolumn{2}{c}{Sample Size} & \multirow{2}[4]{*}{Randomization Procedure} & \multicolumn{4}{c}{Criteria} \\
\cmidrule{1-2}\cmidrule{4-7}    Pilot & Main  &       & Coverage &  $\%\Delta$Length & Power & $\%\Delta$RMSE \\
    \midrule
    \multirow{6}[2]{*}{100} & \multirow{6}[2]{*}{4900} & No Stratification & 95.3  & 0.0   & 77.2  & 0.0 \\
          &       & Ad-Hoc & 95.0  & -6.9  & 82.5  & -7.1 \\
          &       & Ad-Hoc Neyman & 94.9  & -6.1  & 82.7  & -6.3 \\
          &       & Strat. Tree & 94.9  & -1.2  & 78.9  & -0.4 \\
          &       & CV Tree & 95.0  & -9.6  & 85.0  & -9.2 \\
          &       & Optimal Tree & 95.0  & -18.7 & 91.3  & -18.6 \\
    \midrule
    \multirow{6}[2]{*}{500} & \multirow{6}[2]{*}{4500} & No Stratification & 95.1  & 0.0   & 78.1  & 0.0 \\
          &       & Ad-Hoc & 94.8  & -6.9  & 83.6  & -6.2 \\
          &       & Ad-Hoc Neyman & 94.5  & -8.9  & 85.4  & -7.6 \\
          &       & Strat. Tree & 94.8  & -14.3 & 89.2  & -12.3 \\
          &       & CV Tree & 94.3  & -14.1 & 87.9  & -11.3 \\
          &       & Optimal Tree & 94.5  & -17.6 & 91.1  & -15.5 \\
    \midrule
    \multirow{6}[2]{*}{1500} & \multirow{6}[2]{*}{3500} & No Stratification & 94.3  & 0.0   & 77.0  & 0.0 \\
          &       & Ad-Hoc & 94.8  & -6.9  & 83.3  & -7.8 \\
          &       & Ad-Hoc Neyman & 94.8  & -8.8  & 85.0  & -10.2 \\
          &       & Strat. Tree & 95.0  & -14.3 & 88.4  & -14.8 \\
          &       & CV Tree & 94.9  & -13.8 & 88.6  & -15.3 \\
          &       & Optimal Tree & 95.0  & -15.1 & 89.4  & -16.2 \\
    \bottomrule
    \end{tabular}%
   \caption{Simulation Results for Model 1}
  \label{tab:model1}%
\end{table}%
  

\begin{table}[p]
  \centering
    \begin{tabular}{ccccccc}
    \toprule
    \multicolumn{2}{c}{Sample Size} & \multirow{2}[4]{*}{Randomization Procedure} & \multicolumn{4}{c}{Criteria} \\
\cmidrule{1-2}\cmidrule{4-7}    Pilot & Main  &       & Coverage &  $\%\Delta$Length & Power & $\%\Delta$RMSE \\
    \midrule
    \multirow{6}[2]{*}{100} & \multirow{6}[2]{*}{4900} & No Stratification & 95.2  & 0.0   & 56.4  & 0.0 \\
          &       & Ad-Hoc & 95.1  & -2.1  & 59.3  & -1.2 \\
          &       & Ad-Hoc Neyman & 95.0  & 2.4   & 54.2  & 3.2 \\
          &       & Strat. Tree & 94.8  & 8.0   & 50.6  & 10.2 \\
          &       & CV Tree & 94.9  & -7.8  & 63.7  & -5.8 \\
          &       & Optimal Tree & 94.8  & -19.1 & 74.5  & -17.6 \\
    \midrule
    \multirow{6}[2]{*}{500} & \multirow{6}[2]{*}{4500} & No Stratification & 94.7  & 0.0   & 56.4  & 0.0 \\
          &       & Ad-Hoc & 95.0  & -2.1  & 58.7  & -3.6 \\
          &       & Ad-Hoc Neyman & 94.9  & -1.4  & 57.0  & -2.4 \\
          &       & Strat. Tree & 94.6  & -12.8 & 67.4  & -12.2 \\
          &       & CV Tree & 94.6  & -14.0 & 69.3  & -13.9 \\
          &       & Optimal Tree & 95.0  & -17.5 & 73.1  & -18.0 \\
    \midrule
    \multirow{6}[2]{*}{1500} & \multirow{6}[2]{*}{3500} & No Stratification & 95.0  & 0.0   & 55.5  & 0.0 \\
          &       & Ad-Hoc & 94.7  & -2.1  & 57.5  & -1.1 \\
          &       & Ad-Hoc Neyman & 94.6  & -2.3  & 58.7  & -0.9 \\
          &       & Strat. Tree & 94.9  & -12.8 & 68.0  & -12.2 \\
          &       & CV Tree & 94.9  & -12.5 & 67.7  & -12.0 \\
          &       & Optimal Tree & 95.0  & -13.9 & 69.0  & -13.1 \\
    \bottomrule
    \end{tabular}%
\caption{Simulation Results for Model 2}
  \label{tab:model2}%
\end{table}%

\begin{table}[htbp]
  \centering
    \begin{tabular}{ccccccc}
    \toprule
    \multicolumn{2}{c}{Sample Size} & \multirow{2}[4]{*}{Stratification Method} & \multicolumn{4}{c}{Criteria} \\
\cmidrule{1-2}\cmidrule{4-7}    Pilot & Main  &       & Coverage &  $\%\Delta$Length & Power & $\%\Delta$RMSE \\
    \midrule
    \multirow{6}[2]{*}{100} & \multirow{6}[2]{*}{4900} & No Stratification & 95.4  & 0.0   & 30.6  & 0.0 \\
          &       & Ad-Hoc & 95.3  & -2.0  & 31.3  & -2.0 \\
          &       & Ad-Hoc Neyman & 95.1  & 0.9   & 30.5  & 1.5 \\
          &       & Strat. Tree & 94.7  & 15.5  & 23.9  & 17.8 \\
          &       & CV Tree & 95.4  & -1.4  & 30.9  & -0.8 \\
          &       & Optimal Tree & 95.5  & -7.1  & 34.1  & -8.2 \\
    \midrule
    \multirow{6}[2]{*}{500} & \multirow{6}[2]{*}{4500} & No Stratification & 95.0  & 0.0   & 30.5  & 0.0 \\
          &       & Ad-Hoc & 94.6  & -2.1  & 31.4  & -1.2 \\
          &       & Ad-Hoc Neyman & 95.3  & -1.4  & 31.2  & -2.0 \\
          &       & Strat. Tree & 94.8  & -2.2  & 31.3  & -2.2 \\
          &       & CV Tree & 95.0  & -2.8  & 31.0  & -3.6 \\
          &       & Optimal Tree & 94.5  & -6.7  & 34.4  & -5.0 \\
    \midrule
    \multirow{6}[2]{*}{1500} & \multirow{6}[2]{*}{3500} & No Stratification & 94.9  & 0.0   & 31.7  & 0.0 \\
          &       & Ad-Hoc & 94.5  & -2.0  & 31.0  & -1.3 \\
          &       & Ad-Hoc Neyman & 95.1  & -1.9  & 32.4  & -1.6 \\
          &       & Strat. Tree & 94.7  & -4.5  & 33.2  & -3.9 \\
          &       & CV Tree & 94.7  & -3.9  & 33.4  & -3.5 \\
          &       & Optimal Tree & 94.3  & -5.6  & 34.3  & -4.3 \\
    \bottomrule
    \end{tabular}%
     \caption{Simulation Results for Model 3}
  \label{tab:model3}%
\end{table}%


For all three designs, we find that both the stratification tree and CV tree generally outperform no stratification and ad-hoc stratification, with particularly sizable gains for Model 2. However, when using a small pilot, the performance of the stratification tree without cross validation is often found to be worse than not stratifying at all. In such settings, we find that the CV tree effectively protects against overfitting. With larger sized pilots, we see that both trees perform comparably to the optimal tree in all three designs. Overall, we conclude that our proposed cross-validation procedure does a good job of protecting against overfitting, and we recommend that practitioners use cross validation to help select the depth of their trees in practice. With that in mind, we would still caution against using our method with small pilots even when using cross validation to select the depth.

\section{An Application}\label{sec:application}
In this section we study the behavior of our method in an application, using the experimental data from \cite{karlan2017}. First we provide a brief review of the empirical setting: \cite{karlan2017} study how donors to the charity Freedom from Hunger respond to new information about the charity's effectiveness. The experiment, which proceeded in two separate waves corresponding to regularly scheduled fundraising campaigns, randomly mailed one of two different marketing solicitations to previous donors, with one solicitation emphasizing the scientific research on FFH's impact, and the other emphasizing an emotional appeal to a specific beneficiary of the charity. The outcome of interest was the amount donated in response to the mailer. \cite{karlan2017} found that, although the effect of the research insert was small and insignificant, there was substantial heterogeneity in response to the treatment: for those who had given a large amount of money in the past, the effect of the research insert was positive, whereas for those who had given a small amount, the effect was negative. They argue that this evidence is consistent with the behavioral mechanism proposed by \cite{kahneman2003}, where small prior donors are driven by a ``warm-glow" of giving (akin to Kahneman's System I decision making), in contrast to large prior donors, who are driven by altruism (akin to Kahneman's System II decision making). However, the resulting confidence intervals of their estimates are wide, and often contain zero \citep[see for example Figure 1 in][]{karlan2017}. The covariates available in the dataset for stratification are as follows:

\begin{itemize}[topsep = 1pt]
\item Total amount donated prior to mailer 
\item Amount of most recent donation prior to mailer (denoted {\tt pre gift} below)
\item Amount of largest donation prior to mailer 
\item Number of years as a donor (denoted {\tt \# years} below)
\item Number of donations per year (denoted {\tt freq} below)
\item Average years of education in census tract
\item Median zipcode income
\item Prior giving year (either 2004/05 or 2006/07) (denoted {\tt p.year} below)
\end{itemize} 

As a basis for comparison, Figure \ref{fig:app_kw} depicts the stratification used for the first wave in \cite{karlan2017}. 

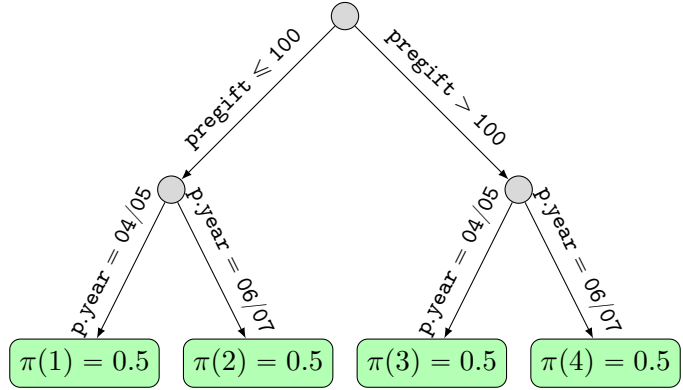
\begin{figure}[H]\center
\begin{tikzpicture}
  [
    grow                    = down,
    level distance          = 6em,
    edge from parent/.style = {draw, -latex},
    sloped,
    level 1/.style={sibling distance=12em,font=\footnotesize},
    level 2/.style={sibling distance=6em,font=\footnotesize},
  ]
  \node [root] {}
    child { node [root] {}
    	child {node [leaf]{$\pi(1) = 0.5$}
	  edge from parent node [above] {${\tt p.year = 04/05}$} }
	child {node [leaf]{$\pi(2) = 0.5$}
	  edge from parent node [above] {${\tt p.year = 06/07}$} }
      edge from parent node [above] {${\tt pre gift \le 100}$} }
    child { node [root] {}
    	child {node[leaf]{$\pi(3) = 0.5$}
	  edge from parent node [above] {${\tt p.year = 04/05}$} }
	child {node[leaf]{$\pi(4) = 0.5$}
	  edge from parent node [above] {${\tt p.year = 06/07}$} }
      edge from parent node [above] {${\tt pre gift > 100}$} };
      \end{tikzpicture}
 \caption{Stratification used in \cite{karlan2017}}\label{fig:app_kw}
 \end{figure}
 
We estimate two different stratification trees using data\footnotemark[3] from the first wave of the experiment (with a sample size of $10,869$), that illustrate stratifications which could have been used to assign treatment in the second wave. We compute the trees by minimizing an empirical analog of the variance, as described in Section \ref{sec:results}. The first tree is fully unconstrained, and hence targets efficient estimation of the unconditional ATE estimator, while the second tree is constrained in accordance with Section \ref{sec:extensions} to efficiently target estimation of the subgroup-specific effects for large and small prior donors (see below for a precise definition). In both cases, the depth of the stratification tree was selected using $2$-fold cross validation as described in Section \ref{sec:extensions}, with a maximal depth of $\bar{L} = 5$ (which corresponds to a maximum of $32$ strata). When computing our trees, given that some of these covariates do not have upper bounds a-priori, we impose an upper bound on the allowable range for the strata to be considered (we set the upper bound as roughly the 97th percentile in the dataset, although in practice this could be set using historical data).
\footnotetext[3]{Replication data is available by request from Innovations for Poverty Action. Observations with missing data on median income, average years of education, and those receiving the ``story insert" were dropped.}

Figure \ref{fig:app_ft} depicts the unrestricted tree estimated via cross-validation. We see that the cross-validation procedure selects a tree of depth one, which may suggest that the covariates available to us for stratification are not especially relevant for decreasing the variance of the estimator. However, we do see a wide discrepancy in the assignment proportions for the selected strata. In words, the subgroup of respondents who have been donors for more than $16$ years have a larger variance in outcomes when receiving the research mailer than the control mailer. In contrast the subgroup of respondents who have been donors for less than $16$ years have roughly equal variances in outcomes under both treatments. 
 
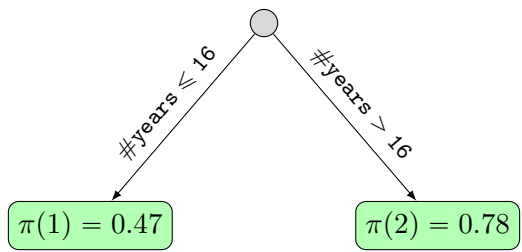
\begin{figure}[H]\center
\begin{tikzpicture}
  [
    grow                    = down,
    level distance          = 7em,
    edge from parent/.style = {draw, -latex},
    sloped,
    level 1/.style={sibling distance=12em,font=\footnotesize},
    level 2/.style={sibling distance=6em,font=\footnotesize},
  ]
  \node [root]  {}
    child { node [leaf] {$\pi(1) = 0.47$}
      edge from parent node [above] {${\tt \# years \le 16}$} }
    child { node [leaf] {$\pi(2) = 0.78$}
      edge from parent node [above] {${\tt\# years > 16}$} };
      \end{tikzpicture}
\caption{Unrestricted Stratification Tree estimated from \cite{karlan2017} data}\label{fig:app_ft}
\end{figure}

Next, we estimate the restricted stratification tree which targets the subgroup-specific treatment effects for large and small prior donors. We specify a large donor as someone who's most recent donation prior to the experiment was larger than $\$100$. We proceed by estimating each subtree using cross-validation. Figure \ref{fig:app_rt} depicts the estimated tree. We see that the cross-validation procedure selects a stratification tree of depth 1 in the left subtree and a tree of depth 0 (i.e. no stratification) in the right subtree, which further reinforces that the covariates we have available may be uninformative for decreasing variance. 



\vspace{5mm}
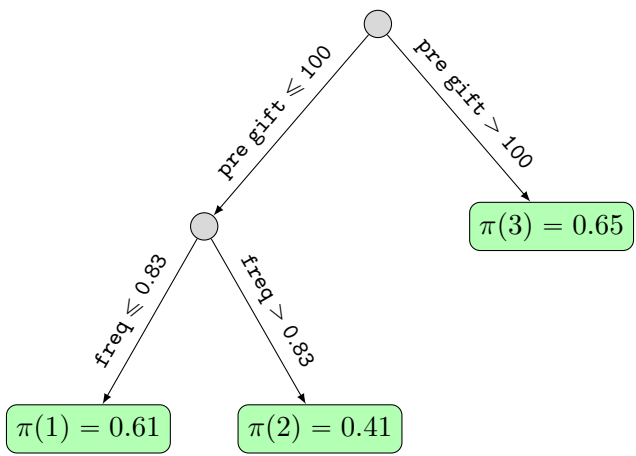
\begin{figure}[H]\center
\begin{tikzpicture}
  [
    grow                    = down,
    level distance          = 7em,
    edge from parent/.style = {draw, -latex},
    sloped,
    level 1/.style={sibling distance=12em,font=\footnotesize},
    level 2/.style={sibling distance=8em,font=\footnotesize},
  ]
  \node [root]  {}
    child { node [root] {}
    	child {node [leaf]{$\pi(1) = 0.60$}
	  edge from parent node [above] {${\tt freq \le 0.83}$}}
	child {node [leaf]{$\pi(2) = 0.42$}
	  edge from parent node [above] {${\tt freq > 0.83}$}}
      edge from parent node [above] {${\tt pre\hspace{1mm}gift \le 100}$}}
    child { node [leaf] {$\pi(3) = 0.65$}
      edge from parent node [above] {${\tt pre\hspace{1mm}gift > 100}$} };
      \end{tikzpicture}
\caption{Restricted Stratification Tree estimated from \cite{karlan2017} data}\label{fig:app_rt}
\end{figure}

These results are not necessarily surprising given the nature of the experiment: with very high probability, a recipient of either mailer is likely to make no donation at all, and hence we might expect limited heterogeneity in the potential outcomes with respect to our observable characteristics. 

\begin{remark}
In Appendix \ref{sec:supp_sim} we repeat the simulation exercise of Section \ref{sec:simulations} with an application-based simulation design. There we find that our method obtains very modest gains in precision relative to the stratification used in \cite{karlan2017}, which may not be surprising given the nature of the experiment. 
\end{remark}

\section{Conclusion}\label{sec:conclusion}
In this paper we proposed an adaptive randomization procedure for two-stage randomized controlled trials, which uses the data from a first-wave experiment to assign treatment in a second wave of the RCT. Our method uses the first-wave data to estimate a stratification tree: a stratification of the covariate space into a tree partition along with treatment assignment probabilities for each of these strata. 

Going forward, there are several extensions of the paper that we would like to consider. First, although we have argued throughout this paper that we find tree partitions to be a natural and convenient constraint, there are serious theoretical and empirical questions about whether or not weakening this restriction could lead to large decreases in asymptotic variance. A potentially easy compromise would be to consider \emph{oblique} tree partitions, where each split of the tree is determined by a linear-index of the covariates. It would be interesting to know to what extent our results generalize to this setting. Second, many RCTs are performed as \emph{cluster} RCTs, that is, where treatment is assigned at a higher level of aggregation such as a school or city. Extending the results of the paper to this setting could be a worthwhile next step. Similarly, we could extend the results to settings with non-compliance by leveraging recent results in \cite{bugni2021}.  Another avenue to consider would be to combine our randomization procedure with other aspects of the experimental design. For example, \cite{carneiro2016} set up a statistical decision problem to optimally select the sample size, as well as the number of covariates to collect from each participant in the experiment, given a fixed budget. It may be interesting to embed our randomization procedure into a similar decision problem. Finally, although our method employs stratified randomization, we assumed throughout that the experimental sample is an i.i.d sample. Further gains may be possible by considering a setting where we are able to conduct stratified \emph{sampling} in the second wave as well as stratified randomization. To that end, \cite{song2014} develop estimators and semi-parametric efficiency bounds for stratified sampling which may be useful.
\pagebreak

\noindent{\Large Acknowledgments}

I am grateful for advice and encouragement from Ivan Canay, Joel Horowitz, and Chuck Manski. I would also like to thank three anonymous referees, Eric Auerbach, Yuehao Bai, Lori Beaman, Stephane Bonhomme, Federico Bugni, Ivan Fernandez-Val, Hidehiko Ichimura, Sasha Indarte, Seema Jayachandran, Vishal Kamat, Dean Karlan, Cynthia Kinnan, Dennis Kristensen, Ryan Lee, Eric Mbakop, Matt Masten, Francesca Molinari, Denis Nekipelov, Sam Norris, Susan Ou, Azeem Shaikh, Mikkel Solvsten, Imran Rasul, Alex Torgovitsky, Chris Udry, Takuya Ura, Andreas Wachter, Joachim Winter, and seminar participants at many institutions for helpful comments and discussions. This research was supported in part through the computational resources and staff contributions provided for the Quest high performance computing facility at Northwestern University, and the Acropolis computing cluster at the University of Chicago.

\pagebreak
\begin{small}
\clearpage
\bibliography{references.bib}
\clearpage
\appendix
\section{Proofs of Main Results}\label{sec:appendixA}
The proof of Theorem \ref{thm:main} requires some preliminary machinery which we develop in Appendix \ref{sec:rho}. In this section we take the following facts as given:

\begin{itemize}[topsep = 1pt]
\item We select a representative out of every equivalence class $T \in \mathcal{T}_L$ by defining an explicit labeling of the leaves, which we call the \emph{canonical labeling} (Definition \ref{def:canon}). 
\item We endow $\mathcal{T}_L$ with a metric $\rho(\cdot,\cdot)$ that makes $(\mathcal{T}_L, \rho)$ a compact metric space (Definition \ref{def:rho1}, Lemma \ref{lem:rho_complete}, Lemma \ref{lem:rho_tb}).
\item We prove that $V(\cdot)$ is continuous in $\rho$ (Lemma \ref{lem:V_cont}).
\item Let $\mathcal{T}^*_L$ be the set of minimizers of $V(\cdot)$, then this set is compact (in the topology induced by $\rho$), and it is the case given our assumptions that 
$$\inf_{T^* \in \mathcal{T}^*_L} \rho(\hat{T}_m,T^*) \xrightarrow{p} 0~,$$
as $m \rightarrow \infty$. Furthermore, there exists a sequence of $\sigma\{(W_i)_{i=1}^m\}/\mathcal{B}(\mathcal{T}_L)$-measurable trees $\bar{T}_m \in \mathcal{T}^*_L$ such that
$$\rho(\hat{T}_m,\bar{T}_m) \xrightarrow{p} 0~.$$
(Lemma \ref{lem:S_const})
\end{itemize}


\noindent{\bf Proof of Theorem \ref{thm:main}}
\begin{proof}
We want to show that for all $t \in \mathbb{R}$,
\[E\left[{\bf 1}\{\sqrt{n}(\hat{\theta}(\hat{T}_m) - \theta) \le t\}\right] \rightarrow \Phi^*(t)~,\] 
where $\Phi^*(t)$ is the CDF of a $N(0, V^*)$ random variable. Let $E_1[\cdot]$ and $E_2[\cdot]$ denote the expectations with respect to the first wave and second wave data, respectively. Fix a strictly increasing indexing $(n_1, m_1) < \ldots < (n_{\ell}, m_{\ell}) < \ldots$ (where the inequality is interpreted componentwise). By Lemma \ref{lem:S_const}, 
\[\rho(\hat{T}_{m_\ell}, \bar{T}_{m_\ell}) \xrightarrow{p} 0~,\]
and hence there exists a subsequence of this sequence (which by an abuse of notation we continue to index by $m_\ell$ with corresponding index $n_\ell$) for which this convergence holds almost surely. By Lemma \ref{lem:main_1},
\[E_2[{\bf 1}\{\sqrt{n_\ell}(\hat{\theta}(\hat{T}_{m_\ell}) - \theta) \le t\}] \xrightarrow{a.s} \Phi^*(t)~.\]
By the dominated convergence theorem, we get that
\[E_1[E_2[{\bf 1}\{\sqrt{n_\ell}(\hat{\theta}(\hat{T}_{m_\ell}) - \theta) \le t\}]] \rightarrow \Phi^*(t)~.\]
By Fubini's theorem,
\[E[{\bf 1}\{\sqrt{n_\ell}(\hat{\theta}(\hat{T}_{m_\ell}) - \theta) \le t\}] = E_1[E_2[{\bf 1}\{\sqrt{n_\ell}(\hat{\theta}(\hat{T}_{m_\ell}) - \theta) \le t\}]] \rightarrow \Phi^*(t)~.\]
Hence by Lemma \ref{lem:double_index}, the result follows.
\end{proof}

\begin{lemma}\label{lem:main_1}
Let $\{T^{(1)}_m\}_m$ be a sequence of trees such that there exists a sequence $\{T^{(2)}_m\}_m$ where $\rho(T^{(1)}_m, T^{(2)}_m) \rightarrow 0$, and $T^{(2)}_m \in \mathcal{T}^*_L$ for all $m$. Given the Assumptions required for Theorem \ref{thm:main}, 
\[\sqrt{n}(\hat{\theta}(T^{(1)}_m) - \theta) \xrightarrow{d} N(0, V^*)~.\]
\end{lemma}
\begin{proof}
By the derivation in the proof of Theorem 3.1 in \cite{bugni2017}, we have that 
$$\sqrt{n}(\hat{\theta}(T^{(1)}_m) - \theta) = \sum_{k=1}^{K}\left[\Omega_1(k; T^{(1)}_m) - \Omega_0(k; T^{(1)}_m)\right]+ \sum_{k=1}^{K}\Theta(k; T^{(1)}_m)~,$$
where
$$\Omega_a(k; T) := \frac{n(k;T)}{n_a(k;T)}\left[\frac{1}{\sqrt{n}}\sum_{i=1}^n{\bf 1}\{A_i(T) = a, S_i = k\}\psi_i(a; T)\right]~,$$ 
with the following definitions:
$$\psi_i(a; T) := Y_i(a) - E[Y_i(a)|S(X)]~,$$ 
$$n(k; T) := \sum_{i=1}^n{\bf 1}\{S_i = k\}~,$$
$$n_a(k; T) := \sum_{i=1}^n {\bf 1}\{A_i(T) = a, S_i = k\}~,$$
and
$$\Theta(k;T) := \sqrt{n}\left(\frac{n(k;T)}{n} - p(k; T)\right)\left[E(Y(1)|S(X) = k) - E(Y(0)|S(X)=k)\right]^2~.$$ 
To prove our result, we study the process
\[ \mathbb{O}_n(T) := \left[\Omega_{0}(1; T) \hspace{2mm} \Omega_{1}(1; T)  \hspace{2mm} \Omega_{0}(2; T) \hspace{1mm} \ldots \hspace{1mm} \Omega_{1}(K; T) \hspace{2mm} \Theta(1; T) \hspace{1mm} \ldots \hspace{1mm} \Theta(K; T)\right]'~.\]
By Lemma \ref{lem:main_2}, we have that 
$$\mathbb{O}_n(T^{(1)}_m)  \,{\buildrel d \over =}\, \bar{\mathbb{O}}_n(T^{(2)}_m) + o_P(1)~,$$
where $\bar{\mathbb{O}}_n(\cdot)$ is defined in Lemma \ref{lem:main_2}. Hence
$$\sqrt{n}(\hat{\theta}(T^{(1)}_m) - \theta)  \,{\buildrel d \over =}\, \Psi_n(T^{(2)}_m) + o_P(1)~,$$
where
\[\Psi_n(T^{(2)}_m) = B'\bar{\mathbb{O}}_n(T^{(2)}_m)~,\]
and $B$ is the appropriate vector of ones and negative ones such that $B'\mathbb{O}_n(T) = \sqrt{n}(\hat{\theta}(T) - \theta)$.
It remains to show that $\Psi_n(T^{(2)}_m) \xrightarrow{d} N(0,V^*)$, and then the result will follow. To that end, fix a strictly increasing indexing $(n_1, m_1) < ... < (n_\ell, m_\ell) < ... $ (where the inequality is to be interpreted componentwise). By  the compactness of $\mathcal{T}^*_L$ (Lemma \ref{lem:S_const}), $\{T^{(2)}_{m_\ell}\}$ contains a convergent subsequence (which by an abuse of notation we continue to index by $m_\ell$, with corresponding index $n_\ell$), so that:
$$T^{(2)}_{m_{\ell}} \rightarrow T^*~,$$ 
for some $T^* \in \mathcal{T}^*_L$.
By the asymptotic equicontinuity of $\bar{\mathbb{O}}_n(\cdot)$ established in Lemma \ref{lem:main_2}, we have that
\[\bar{\mathbb{O}}_{n_\ell}(T^{(2)}_{m_{\ell}}) = \bar{\mathbb{O}}_{n_\ell}(T^*) + o_P(1)~,\]
and by the partial sum arguments in Lemma C.1. of \cite{bugni2017}, it follows that
\[\Psi_{n_\ell}(T^*) \xrightarrow{d} N(0,V(T^*))~,\]
and $V(T^*) = V^*$ since $T^*$ is an optimal tree.
Hence we have that
\[\Psi_{n_\ell}(T^{(2)}_{m_{\ell}}) \xrightarrow{d} N(0,V^*)~.\]
By Lemma \ref{lem:double_index} (applied to the CDFs), we get that
\[\Psi_n(T^{(2)}_m) \xrightarrow{d} N(0,V^*)~,\]
as $m, n \rightarrow \infty$, and so the result follows.
\end{proof}

\begin{lemma}\label{lem:main_2}
Let $\{T^{(1)}_m\}_m$ be a sequence of trees such that there exists a sequence $\{T^{(2)}_m\}_m$ where $\rho(T^{(1)}_m, T^{(2)}_m) \rightarrow 0$, and $T^{(2)}_m \in \mathcal{T}^*_L$ for all $m$. Given the Assumptions required for Theorem \ref{thm:main}, 
$$\mathbb{O}_n(T^{(1)}_m)  \,{\buildrel d \over =}\, \bar{\mathbb{O}}_n(T^{(2)}_m) + o_P(1)~,$$
as $n \rightarrow \infty$, where $\mathbb{O}_n(\cdot)$ is defined in the proof of Lemma \ref{lem:main_1} and $\bar{\mathbb{O}}_n(\cdot)$ is defined in the proof of this result.
\end{lemma}
\begin{proof}
By the argument in Lemma C1 in \cite{bugni2017}, we have that 
$$\mathbb{O}_n(T) \,{\buildrel d \over =}\, \widetilde{\mathbb{O}}_n(T)~,$$ 
where 
\[ \widetilde{\mathbb{O}}_n(T) := \left[\widetilde{\Omega}_{0}(1; T) \hspace{2mm} \widetilde{\Omega}_{1}(1; T)  \hspace{2mm} \widetilde{\Omega}_{0}(2; T) \hspace{1mm} \ldots \hspace{1mm} \widetilde{\Omega}_{1}(K; T) \hspace{2mm} \Theta(1; T) \hspace{1mm} \ldots \hspace{1mm} \Theta(K; T)\right]'~.\]
with
$$\widetilde{\Omega}_a(k; T) := \frac{n(k;T)}{n_a(k;T)}\left[\frac{1}{\sqrt{n}}\sum_{i=n(\hat{F}(k;T)+\hat{F}_a(k;T))+1}^{n(\hat{F}(k;T)+\hat{F}_{a+1}(k;T))}G^k_a(U_{i,(a)}(k); T)\right]~,$$
with the following definitions: $\{U_{i,(a)}(k)\}_{i=1}^N$ are i.i.d $U[0,1]$ random variables generated independently of everything else, and independently across pairs $(a,k)$, $G^{k}_a(\cdot \hspace{1mm}; T)$ is the quantile function of the distribution of $\psi(a; T)|S(X) = k$ \citep[where here we use the quantile transformation of a uniform random variable: see][Chapter 21]{van1998}, $\hat{F}(k;T) := \frac{1}{n}\sum_{i=1}^n {\bf 1}\{S_i < k\}$, and $\hat{F}_a(k;T) := \frac{1}{n}\sum_{i=1}^n{\bf 1}\{S_i=k, A_i < a\}$.

Let us focus on the term in brackets. Fix some $a$ and $k$ for the time being, and let $$\mathcal{G} := \{G^{k}_a(\cdot \hspace{1mm}; T): T \in \mathcal{T}\}$$ be the class of all the inverse CDFs defined above, then the empirical process $\eta_n:[0,1]\times\mathcal{G}\rightarrow\mathbb{R}$ defined by
$$\eta_n(u,f) := \frac{1}{\sqrt{n}}\sum_{i=1}^{\lfloor nu \rfloor}f(U_i)~,$$
is known as the \emph{sequential empirical process} (see \cite{van1996}) (note that by construction $E[f(U_i)] = 0$). By Theorem 2.12.1 in \cite{van1996}, $\eta_n$ converges in distribution to a tight limit in $\ell^\infty([0,1]\times\mathcal{G})$ if $\mathcal{G}$ is Donsker, which follows by Lemma \ref{lem:G_donsker}. It follows that $\eta_n$ is asymptotically equicontinuous in the natural (pseudo) metric $$d\left((u,f),(v,g)\right) = |u-v| + \rho_P(f,g)~,$$ where $\rho_P$ is the variance pseudometric. Note that since $U_i \sim U[0,1]$ and $E[f(U_i)] = 0$ for all $f \in \mathcal{G}$, $\rho_P$ is equal to the $L^2$ norm $||\cdot||$. Define  $F(k;T) := P(S(X) < k)$ and $F_a(k;T) := \sum_{j < a}p(k;T)\pi_j(k)$, where $\pi_0(k) := 1-\pi(k)$, $\pi_1(k) := \pi$, then it follows by Lemmas \ref{lem:F_limit}, and \ref{lem:G_limit} that:
$$|\hat{F}_a(k;T^{(1)}_m) - F_a(k;T^{(2)}_m)| \xrightarrow{p} 0~,$$
$$|\hat{F}(k;T^{(1)}_m) - F(k;T^{(2)}_m)| \xrightarrow{p} 0~,$$
$$||G^{k}_a(\cdot \hspace{1mm}; T^{(1)}_m) - G^{k}_a(\cdot \hspace{1mm}; T^{(2)}_m)|| \rightarrow 0~,$$
as $m \rightarrow \infty$.
Hence we have by asymptotic equicontinuity that
$$\eta_n\left(\hat{F}(k;T^{(1)}_m) + \hat{F}_a(k;T^{(1)}_m),G^{k}_a(\cdot \hspace{1mm};T^{(1)}_m)\right) = \eta_n\left(F(k;T^{(2)}_m) + F_a(k;T^{(2)}_m),G^{k}_a(\cdot \hspace{1mm}; T^{(2)}_m)\right) + o_P(1)~.$$
By Lemma \ref{lem:N_limit},
$$\frac{n(k;T^{(1)}_m)}{n_a(k;T^{(1)}_m)} = \frac{1}{\pi(k; T^{(2)}_m)} + o_P(1)~.$$
Using the above two expressions, it can be shown that
$$\widetilde{\Omega}_a(k; T^{(1)}_m) = \bar{\Omega}_a(k;T^{(2)}_m) + o_P(1)~,$$
where
$$ \bar{\Omega}_a(k; T) := \frac{1}{\pi(k; T)}\left[\frac{1}{\sqrt{n}}\sum_{i=\lfloor n(F(k;T) + F_a(k;T))\rfloor+1}^{\lfloor n(F(k;T)+F_{a+1}(k;T))\rfloor}G^{k}_{a}(U_{i,(a)}(k); T)\right]~.$$
Now we turn our attention to $\Theta(k; T)$. By standard empirical process results for 
$$\sqrt{n}\left(\frac{n(k;T)}{n} - p(k; T)\right)~,$$
it can be shown that
$$\Theta(k; T^{(1)}_m) = \Theta(k;T^{(2)}_m) + o_P(1)~,$$
since the class of indicators $\{{\bf 1}\{S(X) = k\}:S \in \mathcal{S}\}$ is Donsker for each $k$ (since the partitions are rectangles and hence for a fixed $k$ we get a VC class).
Finally, let 
  \[ \bar{\mathbb{O}}_n(T) := \left[\bar{\Omega}_{0}(1; T) \hspace{2mm} \bar{\Omega}_{1}(1; T)  \hspace{2mm} \bar{\Omega}_{0}(2; T) \hspace{1mm} \ldots \hspace{1mm} \bar{\Omega}_{1}(K; T) \hspace{2mm} \Theta(1; T) \hspace{1mm} \ldots \hspace{1mm} \Theta(K; T)\right]'~.\]
then we have shown that
$$\mathbb{O}_n(T^{(1)}_m) \,{\buildrel d \over =}\, \bar{\mathbb{O}}_n(T^{(2)}_m) + o_P(1),$$ 
as desired.
\end{proof}

%

\noindent{\bf Proof of Theorem \ref{thm:var_const}}
\begin{proof}
As in the proof of Lemma \ref{lem:main_1}, let $\{T_m^{(1)}\}_m$ be a sequence of trees such that there exists a sequence $\{T_m^{(2)}\}_m$ where $\rho(T_m^{(1)}, T_m^{(2)}) \rightarrow 0$, and $T_m^{(2)} \in \mathcal{T}^*_L$ for all $m$. By adapting the derivation in Theorem 3.3 of \cite{bugni2017} using the same techniques developed in the proof of Theorem \ref{thm:main}, it can be shown that
$$\hat{V}_H(T_m^{(1)}) = V_H(T_m^{(2)}) + o_P(1)~,$$
\[\hat{V}_Y(T_m^{(1)}) = V_Y(T_m^{(2)}) + o_P(1)~,\]
where
\[V_H(T) = \sum_k P(S(X) = k)\left(E[Y(1) - Y(0)|S(X) = k] - E[Y(1) - Y(0)]\right)^2~,\]
\[V_Y(T) = \sum_k P(S(X) = k)\left(\frac{\sigma^2_0(k)}{1 - \pi(k)} + \frac{\sigma^2_1(k)}{\pi(k)}\right)~.\]
Since $T_m^{(2)} \in \mathcal{T}^*_L$, it follows immediately that 
\[E_2\left[{\bf 1}\left\{\left|\hat{V}(T_m^{(1)}) - V^*\right| > \epsilon\right\}\right] \rightarrow 0~,\]
where $E_1[\cdot]$ and $E_2[\cdot]$ denote the expectations with respect to the first wave and second wave data, respectively. Applying the above result to $\hat{T}_m$, and arguing as in the proof of Theorem \ref{thm:main} using Lemmas \ref{lem:S_const}, \ref{lem:double_index}, the dominated convergence theorem and Fubini's theorem, it can be shown that 
\[P\left(\left|\hat{V}(\hat{T}_m) - V^*\right| > \epsilon\right) \rightarrow 0~,\]
as desired.
\end{proof}

\noindent{\bf Proof of Proposition \ref{prop:SBR_unif}}
\begin{proof}
By definition,
$$\frac{n_1(k)}{n} = \frac{\left\lfloor n(k)\pi(k)\right\rfloor}{n}~.$$
We bound the floor function from above and below:
$$\pi(k)\frac{n(k)}{n} - \frac{1}{n} < \frac{n_1(k)}{n} \le \pi(k)\frac{n(k)}{n}~.$$
We consider the upper bound (the lower bound proceeds identically). It suffices to show that
$$\sup_{T \in \mathcal{T}_L}\left|\frac{n(k;T)}{n} - p(k; T)\right| \xrightarrow{p} 0~.$$
Since the partitions are rectangles, for a fixed $k$ we get a VC class and hence by the Glivenko-Cantelli theorem the result follows.
\end{proof}

\noindent{\bf Proof of Proposition \ref{prop:emp_const}}
\begin{proof}
Recall that we define $\widetilde{V}_m(T)$ as follows:
\[\widetilde{V}_m(T) := \sum_{k=1}^K\frac{m(k;T)}{m}\left[\left(\hat{E}[Y(1) - Y(0)|S(X) = k] - \hat{E}[Y(1) - Y(0)]\right)^2 + \left(\frac{\hat{\sigma}^2_{0,S}(k)}{1 - \pi(k)} + \frac{\hat{\sigma}^2_{1,S}(k)}{\pi(k)}\right)\right]\]
\[ =  \sum_{k=1}^K\frac{m(k;T)}{m}\left[\left(\hat{\theta}_S(k) - \hat{\theta}\right)^2 + \left(\frac{\hat{\sigma}^2_{0,S}(k)}{1 - \pi(k)} + \frac{\hat{\sigma}^2_{1,S}(k)}{\pi(k)}\right)\right]~,\]
with
\[\hat{\theta}_S(k) = \hat{E}[Y(1) - Y(0)|S(X) = k]~,\]
\[\hat{\theta} = \hat{E}[Y(1)] - \hat{E}[Y(0)] ~.\]
First note that, for a given realization of the data, there exists an optimal choice of $\pi$ for every $S \in \mathcal{S}_L$ by continuity of $\widetilde{V}_m(T)$ in $\pi$ (which we'll call $\pi^*_S$), so our task is to choose $(S, \pi^*_S)$ to minimize $\widetilde{V}_m(T)$. Given this, note that for a given realization of the data, the empirical objective $\widetilde{V}_m(T)$ can take on only finitely many values, and hence a minimizer $\hat{T}_m$ exists. In what follows, let $\{\hat{T}_m\}_m$ be a sequence of empirical minimizers constructed from this procedure.

Now, let $T^*$ be a minimizer of $V(T)$ (which exists by Lemma \ref{lem:S_const}), and let $\hat{T}_m$ be a minimizer of $\widetilde{V}_m(T)$ (which exists by the argument in the first paragraph of the proof) then
\begin{align*}
V(\hat{T}_m) - V(T^*)
&=V(\hat{T}_m) - \widetilde{V}_m(\hat{T}_m) + \widetilde{V}_m(\hat{T}_m) - V(T^*) \\
&\le V(\hat{T}_m) - \widetilde{V}_m(\hat{T}_m) + \widetilde{V}_m(T^*) - V(T^*)\\
&\le 2\sup_{T \in \mathcal{T}_L}|\widetilde{V}_m(T) - V(T)| ~.
\end{align*}
So if we can show 
$$\sup_{T \in \mathcal{T}_L} |\widetilde{V}_m(T) - V(T)| \xrightarrow{p} 0~,$$
then we are done. We prove this in Lemma \ref{lem:emp_unif} below.
\end{proof}


\noindent{\bf Proof of Proposition \ref{prop:theta_robust}}
\begin{proof}
We show that for any deterministic sequence of trees $\{T_m\}_m$ in $\mathcal{T}_L$, 
\[\lim_{n,m \rightarrow \infty}E_2[\phi_n(T_m)] = \alpha~,\]
under $H_0$, where $E_2[\cdot]$ denotes the expectation with respect to the second wave of data. The result then follows by Fubini's theorem and the dominated convergence theorem. To that end, we argue that
\[\frac{\sqrt{n}(\hat{\theta}(T_m) - \theta)}{\sqrt{\hat{V}(T_m)}} \xrightarrow{d} N(0,1)~.\]
As in the proof of Lemma \ref{lem:main_1}, fix a strictly increasing indexing $(n_1, m_1) < \ldots < (n_\ell, m_\ell) <\ldots$. By the compactness of $\mathcal{T}_L$ (Lemmas \ref{lem:rho_complete}, \ref{lem:rho_tb}), $\{T_{m_\ell}\}$ contains a convergent subsequence (which by an abuse of notation we continue to index by $m_\ell$, with corresponding index $n_\ell$), so that $T_{m_\ell} \rightarrow T'$ for some $T' \in \mathcal{T}_L$. Then as in the proof of Lemma \ref{lem:main_1},
\[\sqrt{n_\ell}(\hat{\theta}(T_{m_\ell}) - \theta) \,{\buildrel d \over =}\, \Psi_{n_\ell}(T') + o_P(1)~,\]
and as in the proof of Theorem \ref{thm:var_const},
\[\hat{V}(T_{m_\ell}) = V(T') + o_P(1)~.\]
By the partial sum arguments in Lemma C.1. of \cite{bugni2017},
\[\Psi_{n_\ell}(T') \xrightarrow{d} N(0, V(T'))~,\]
so by Slutsky's theorem 
\[\frac{\sqrt{n_\ell}(\hat{\theta}(T_{m_\ell}) - \theta)}{\sqrt{\hat{V}(T_{m_\ell})}} \xrightarrow{d} N(0,1)~,\]
and hence the result follows by applying Lemma \ref{lem:double_index} (to the CDFs).
\end{proof}

\noindent{\bf Proof of Theorem \ref{thm:pool_est}}
\begin{proof}
Let $t_1, t_2 \in \mathbb{R}$ be arbitrary, then we will to show that
\[P\left(\sqrt{m}(\hat{\theta}_1 - \theta) \le t_1, \sqrt{n}(\hat{\theta}(\hat{T}_m) - \theta) \le t_2\right) \rightarrow \Phi_1(t_1)\Phi^*(t_2)~,\]
where $\Phi_1(\cdot)$ is the CDF of a $N(0, V_1)$ random variable, and $\Phi^*(\cdot)$ is the CDF of a $N(0, V^*)$ random variable.
The result will then follow by Assumption \ref{ass:m/N} and Slutsky's theorem.
As in the proof of Theorem \ref{thm:main}, let $E_1[\cdot]$ and $E_2[\cdot]$ denote the expectations with respect to the first and second wave data, respectively, then by Fubini's theorem,
\[P\left(\sqrt{m}(\hat{\theta}_1 - \theta) \le t_1, \sqrt{n}(\hat{\theta}(\hat{T}_m) - \theta) \le t_2\right) = E_1\left[E_2\left[{\bf 1}\{\sqrt{m}(\hat{\theta}_1 - \theta) \le t_1\}{\bf 1}\{\sqrt{n}(\hat{\theta}(\hat{T}_m) - \theta) \le t_2\}\right]\right]~.\]
Adding and subtracting $P(\sqrt{m}(\hat{\theta}_1 - \theta) \le t_1)\Phi^*(t_2)$ gives, after some additional algebra,
\begin{align*}
E_1\left[E_2\left[{\bf 1}\{\sqrt{m}(\hat{\theta}_1 - \theta) \le t_1\}{\bf 1}\{\sqrt{n}(\hat{\theta}(\hat{T}_m) - \theta) \le t_2\}\right]\right] &= E_1\left[\left(E_2\left[{\bf 1}\{\sqrt{n}(\hat{\theta}(\hat{T}_m) - \theta) \le t_2\}\right] - \Phi^*(t_2)\right){\bf 1}\{\sqrt{m}(\hat{\theta}_1 - \theta) \le t_1\}\right]  \\
&+ P(\sqrt{m}(\hat{\theta}_1 - \theta) \le t_1)\Phi^*(t_2)~.
\end{align*}

By Assumption \ref{ass:pilot_normal}, we have that
\[P(\sqrt{m}(\hat{\theta}_1 - \theta) \le t_1)\Phi^*(t_2) \rightarrow \Phi_1(t_1)\Phi^*(t_2)~.\]
It remains to show that 
\[E_1\left[\left(E_2\left[{\bf 1}\{\sqrt{n}(\hat{\theta}(\hat{T}_m) - \theta) \le t_2\}\right] - \Phi^*(t_2)\right){\bf 1}\{\sqrt{m}(\hat{\theta}_1 - \theta) \le t_1\}\right] \rightarrow 0~.\]
By the triangle inequality,
\[\left|E_1\left[\left(E_2\left[{\bf 1}\{\sqrt{n}(\hat{\theta}(\hat{T}_m) - \theta) \le t_2\}\right] - \Phi^*(t_2)\right){\bf 1}\{\sqrt{m}(\hat{\theta}_1 - \theta) \le t_1\}\right]\right| \le E_1\left|E_2\left[{\bf 1}\{\sqrt{n}(\hat{\theta}(\hat{T}_m) - \theta) \le t_2\}\right] - \Phi^*(t_2)\right|~.\]
This last term converges to zero as in the proof of Theorem \ref{thm:main} and hence we're done.
\end{proof}

\begin{lemma}\label{lem:emp_unif}
Maintain the assumptions of Proposition \ref{prop:emp_const}. Then
\[\sup_{T \in \mathcal{T}_L} |\widetilde{V}_m(T) - V(T)| \xrightarrow{p} 0~,\]
as $m \rightarrow \infty$.
\end{lemma}
\begin{proof}
Re-write the population-level variance $V(T)$ as follows:
$$V(T) = E[\nu_{T}(X)]~,$$
where
$$\nu_T(x) = \left[\frac{\sigma^2_{1,S}(x)}{\pi(S(x))} - \frac{\sigma^2_{0,S}(x)}{1 - \pi(S(x))} + (\theta_S(x) - \theta)^2\right]~,$$
$$\sigma^2_{a,S}(x) = Var(Y(a)|S(X) = S(x))~,$$
$$\theta_S(x) = E[Y(1) - Y(0)|S(X) = S(x)]~.$$
Given this notation we also introduce the following intermediate quantity:
$$V_m(T) := \frac{1}{m}\sum_{j=1}^m\nu_{T}(X_j)~.$$
By the triangle inequality:
$$\sup_{T\in\mathcal{T}_L}|\widetilde{V}_m(T) - V(T)| \le\sup_{T\in\mathcal{T}_L} |\widetilde{V}_m(T) - V_m(T)| + \sup_{T\in\mathcal{T}_L}|V_m(T) - V(T)|~,$$
so we study each of these in turn. Let us look at the second term on the right hand side. This converges almost surely to zero by the Glivenko-Cantelli theorem, since the sequence $\{X_j\}_{i=1}^m$ is i.i.d and the class of functions $\{\nu_T(\cdot): T \in \mathcal{T}\}$ is Glivenko-Cantelli (this can be seen by the fact that $\nu_T(\cdot)$ can be constructed through appropriate sums, products, differences and quotients of various types of VC-subgraph functions, and by invoking Assumption \ref{ass:prop/cell_size} to avoid potential degeneracies through division). Hence it remains to show that the first term converges to zero in probability. 

Re-writing:
$$\widetilde{V}_m(T) = \sum_{k=1}^K\left[\left(\frac{1}{m}\sum_{j=1}^m{\bf 1}\{S(X_j)=k\}\right)\left(\frac{\hat{\sigma}^2_{1,S}(k)}{\pi(k)} - \frac{\hat{\sigma}^2_{0,S}(k)}{1 - \pi(k)} + (\hat{\theta}_S(k) - \hat{\theta})^2\right)\right]~,$$
and 
$$V_m(T) = \sum_{k=1}^K\left[\left(\frac{1}{m}\sum_{j=1}^m{\bf 1}\{S(X_j)=k\}\right)\left(\frac{\sigma^2_{1,S}(k)}{\pi(k)} - \frac{\sigma^2_{0,S}(k)}{1 - \pi(k)} + (\theta_S(k) - \theta)^2\right)\right]~,$$
where, through an abuse of notation, we define $\sigma^2_{a,S}(k) := Var(Y(a)|S(X)=k)$ etc. By the triangle inequality it suffices to consider each $k \in [K]$ individually. Moreover, since the expression $\frac{1}{m}\sum_{j=1}^m{\bf 1}\{S(X_j)=k\}$ is bounded, we can factor it out and ignore it in what follows.
It can be shown by repeated applications of the triangle inequality, Assumptions \ref{ass:bounded}, \ref{ass:prop/cell_size}, and Lemma \ref{lem:GC_couple}, that
$$\sup_{T \in \mathcal{T}_L}\left|\left(\frac{\hat{\sigma}^2_{1,S}(k)}{\pi(k)} - \frac{\hat{\sigma}^2_{0,S}(k)}{1 - \pi(k)} + (\hat{\theta}_S(k) - \hat{\theta})^2\right) - \left(\frac{\sigma^2_{1,S}(k)}{\pi(k)} - \frac{\sigma^2_{0,S}(k)}{1 - \pi(k)} + (\theta_S(k) - \theta)^2\right)\right| \xrightarrow{p} 0~.$$
Hence, we see that our result follows.
\end{proof}

\begin{lemma}\label{lem:GC_couple}
Maintain the assumptions of Proposition \ref{prop:emp_const}, then for all $a$ and $k$:
\[\sup_{T \in \mathcal{T}_L}\left|\hat{\theta} - \theta\right| \xrightarrow{p} 0~,\]
\[\sup_{T \in \mathcal{T}_L}\left|\hat{\theta}_S(k) - \theta_S(k)\right| \xrightarrow{p} 0~,\]
\[\sup_{T \in \mathcal{T}_L}\left|\hat{\sigma}^2_{a,S}(k) - \sigma^2_{a,S}(k)\right| \xrightarrow{p} 0~,\]
where $\hat{\theta}$, $\hat{\theta}_S(k)$ and $\hat{\sigma}^2_{a,S}(k)$ are defined as in the proof of Proposition \ref{prop:emp_const}.
\end{lemma}
\begin{proof}
We show that, for $f(\cdot)$ a continuous function,
\begin{equation}\label{eq:unif_conv}
\sup_{T \in \mathcal{T}_L}\left|\frac{1}{m_a(k; T)}\sum_{j = 1}^mf(Y_j(a)){\bf 1}\{A_j = a\}{\bf 1}\{S(X_j) = k\} - E[f(Y(a))|S(X) = k]\right| \xrightarrow{p} 0~.
\end{equation}
First note that \emph{if} $\{(Y_j(a),X_j, A_j)\}_j$ were an i.i.d sequence of random variables, then by the Glivenko-Cantelli Theorem
\[\sup_{T \in \mathcal{T}_L}\left|\frac{1}{m}\sum_{j=1}^mf(Y_j(a)){\bf 1}\{A_j = a\}{\bf 1}\{S(X_j) = k\} - E\left[f(Y(a)){\bf 1}\{A = a\}{\bf 1}\{S(X) = k\}\right]\right| \xrightarrow{p} 0~,\]
\[\sup_{T \in \mathcal{T}_L}\left|\frac{m_a(k;T)}{m} - E\left[{\bf 1}\{A = a\}{\bf 1}\{S(X) = k\}\right]\right| \xrightarrow{p} 0 ~.\]
Using Assumptions \ref{ass:bounded} (boundedness of $Y(a)$) and \ref{ass:prop/cell_size}, it follows that
\begin{equation}\label{eq:frac_conv} 
\sup_{T \in \mathcal{T}_L}\left|\frac{1}{m_a(k; T)}\sum_{j = 1}^mf(Y_j(a)){\bf 1}\{A_j = a\}{\bf 1}\{S(X_j) = k\} - \frac{E\left[f(Y(a)){\bf 1}\{A = a\}{\bf 1}\{S(X) = k\}\right]}{E\left[{\bf 1}\{A = a\}{\bf 1}\{S(X) = k\}\right]}\right| \xrightarrow{p} 0~.
\end{equation}
Recall that we denote the stratification function used to generate $\{{A}_j\}_j$ by $\zeta(\cdot)$, and let  $\zeta_j = \zeta(X_j)$. If we further assume that  $E[{\bf 1}\{A = a\}|\zeta = z] = \pi$, then (\ref{eq:unif_conv}) follows since, using the fact that $(Y(1),Y(0),X) \perp A | \zeta$ and the law of iterated expectations:
\[\frac{E\left[f(Y(a)){\bf 1}\{A = a\}{\bf 1}\{S(X) = k\}\right]}{E\left[{\bf 1}\{A = a\}{\bf 1}\{S(X) = k\}\right]} = \frac{\pi E\left[f(Y(a)){\bf 1}\{S(X) = k\}\right]}{\pi E\left[{\bf 1}\{S(X) = k\}\right]} = E[f(Y(a))|S(X) = k]~. \]
However, we \emph{do not} assume that $\{A_j\}_j$  are i.i.d. To proceed, we couple $\{(Y_j(a), X_j, A_j)\}_{j=1}^m$ to a new i.i.d sequence of random variables $\{(\tilde{Y}_{j,m}, \tilde{X}_{j,m}, \tilde{A}_{j,m})\}_{j=1}^m$. Following the constructions in Theorem 4.5 of \cite{bugni2015} and Section 5.3 in \cite{chung2013}, we generate $\{(\tilde{Y}_{j,m}, \tilde{X}_{j,m}, \tilde{A}_{j,m})\}_{j=1}^m$ via the following procedure. Let $\mathcal{J} = \{1, ..., m\}$, $R_m = 0$. For $j = 1, ..., m$:

%

\begin{itemize}
\item Draw a random variable $\tilde{\zeta}_{j,m} \in \{0, 1, \ldots, Z\}$ independent of everything else such that $P(\tilde{\zeta}_{j,m} = z) = P(\zeta(X) = z)$.
\item Draw a random variable $\tilde{A}_{j,m} \in \{0, 1\}$ independent of everything else, such that $P(\tilde{A}_{j,m} = 1) =  \pi$. 
\item If there exists some $j' \in \mathcal{J}$ such that $\zeta_{j'} = \tilde{\zeta}_{j,m}$ and $A_{j'} = \tilde{A}_{j,m}$, set $\tilde{Y}_{j,m} = Y_{j'}(a)$, $\tilde{X}_{j,m} = X_{j'}$, and $\mathcal{J} := \mathcal{J} \setminus \{j'\}$.
Otherwise, if no such $j'$ exists, draw a realization from the distribution of $(Y(a), X) | \zeta = \tilde{\zeta}_{j,m}$, and set $\tilde{Y}_{j,m} = Y(a)$, $\tilde{X}_{j,m} = X$, and $R_m := R_m + 1$.
\end{itemize}


By Assumption \ref{ass:bounded} and the continuity of $f(\cdot)$, $f(\cdot)$ is bounded on the support of $Y(a)$ (and hence $\tilde{Y}_{j,m}$) by some finite $M'$. Hence by the construction of $\{(\tilde{Y}_{j,m}, \tilde{X}_{j,m}, \tilde{A}_{j,m})\}_{j=1}^m$ and the uniform Glivenko-Cantelli Theorem 2.8.1 in \cite{van1996}:\footnotemark[4]
\[\sup_{T \in \mathcal{T}_L}\left|\frac{1}{m}\sum_{j = 1}^mf(\tilde{Y}_{j,m}){\bf 1}\{\tilde{A}_{j,m} = a\}{\bf 1}\{S(\tilde{X}_{j,m}) = k\} - \pi E[f(Y(a)){\bf 1}\{S(X) = k\}]\right| \xrightarrow{p} 0~,\]
\[\sup_{T \in \mathcal{T}_L}\left|\frac{\tilde{m}_a(k;T)}{m} -\pi E[{\bf 1}\{S(X) = k\}]\right| \xrightarrow{p} 0~,\]
\footnotetext[4]{Here we require a uniform Glivenko-Cantelli result since we constructed $\{(\tilde{Y}_{j,m}, \tilde{X}_{j,m}, \tilde{A}_{j,m})\}_{j=1}^m$ as a triangular array. Since our class of functions is VC subgraph we do obtain such a result: see Chapter 2.8 of \cite{van1996} for details.}
where $\tilde{m}_a(k;T) := \sum_{j=1}^m{\bf 1}\{\tilde{A}_{j,m} = a\}{\bf 1}\{S(\tilde{X}_{j,m}) = k\}$. Moreover from the construction of the coupling we have that $R_m$ upper-bounds the number of observations in $\{(\tilde{Y}_{j,m}, \tilde{X}_{j,m}, \tilde{A}_{j,m})\}_{j=1}^m$ which are not identically equal to some element in $\{(Y_j(a), X_j, A_j)\}_{j=1}^m$. Hence by the triangle inequality,
\[\sup_{T \in \mathcal{T}_L}\left|\frac{1}{m}\sum_{j = 1}^mf(Y_j(a)){\bf 1}\{A_j = a\}{\bf 1}\{S(X_j) = k\} - \frac{1}{m}\sum_{j=1}^mf(\tilde{Y}_{j,m}){\bf 1}\{\tilde{A}_{j,m} = a\}{\bf 1}\{S(\tilde{X}_{j,m}) = k\}\right| \le 2M'\frac{R_m}{m}~,\]
\[\sup_{T \in \mathcal{T}_L}\left|\frac{\tilde{m}_a(k;T)}{m} - \frac{m_a(k;T)}{m} \right| \le 2\frac{R_m}{m}~.\]
By Lemma \ref{lem:Rm_converge}, $R_m/m \xrightarrow{p} 0$, and hence it follows that 
\[\sup_{T \in \mathcal{T}_L}\left|\frac{1}{m}\sum_{j = 1}^mf(Y_j(a)){\bf 1}\{A_j = a\}{\bf 1}\{S(X_j) = k\} -  \pi E[f(Y(a)){\bf 1}\{S(X) = k\}]\right| \xrightarrow{p} 0~,\]
\[\sup_{T \in \mathcal{T}_L}\left|\frac{m_a(k;T)}{m} - \pi E[{\bf 1}\{S(X) = k\}] \right| \xrightarrow{p} 0~.\]
By invoking Assumptions \ref{ass:bounded} and \ref{ass:prop/cell_size}, we thus obtain that
\[\sup_{T \in \mathcal{T}_L}\left|\frac{1}{m_a(k;T)}\sum_{j=1}^m f(Y_j(a)){\bf 1}\{A_j = a\}{\bf 1}\{S(X_j) = k\} - E[f(Y(a))|S(X) = k]\right| \xrightarrow{p} 0 ~,\]
as desired.
\end{proof}

\begin{lemma}\label{lem:Rm_converge}
Let $R_m$ be defined as in the proof of Lemma \ref{lem:GC_couple}. Maintain the assumptions of Proposition \ref{prop:emp_const}. Then
\[\frac{R_m}{m} \xrightarrow{p} 0~.\]
\end{lemma}
\begin{proof}
Let $M_a(z) := \sum_{j=1}^m{\bf 1}\{\tilde{\zeta}_{j,m} = z, \tilde{A}_{j,m} = a\}$ for $\tilde{\zeta}_{j,m}, \tilde{A}_{j,m}$ as in the proof of Lemma \ref{lem:GC_couple}. Since $R_m$ increments by $1$ everytime there does not exist a suitable index $j'$ such that $\tilde{\zeta}_{j,m} = \zeta_{j'}$ and  $\tilde{A}_{j,m} = A_{j'}$, an upper bound for $R_m$ is given by
\[R_m \le \sum_{z = 1}^Z\sum_{a \in \{0, 1\}} \max\{m_a(z;\zeta) - M_a(z), 0\} ~.\]
Hence to prove our result it suffices to show that for all $z$ and $a$, 
\[\frac{m_a(z;\zeta) - M_a(z)}{m} \xrightarrow{p} 0~,\]
By the law of large numbers,
\[\frac{M_1(z)}{m} \xrightarrow{p} \pi P(\zeta(X) = z)~.\]
By our assumptions on the assignment mechanism an the law of large numbers,
\[\frac{m_1(z)}{m} \xrightarrow{p} \pi P(\zeta(X) = z)~.\]
Hence the result follows.
\end{proof}

\begin{lemma}\label{lem:F_limit}
Let $\hat{F}$, $\hat{F}_a$, $F$ and $F_a$ be defined as in the proof of Lemma \ref{lem:main_2}. Let $T^{(1)}_m$, $T^{(2)}_m$ be defined as in the statement of Lemma \ref{lem:main_2}. Given the Assumptions of Theorem \ref{thm:main}, we have that, for $k = 1,..., K$,
$$|\hat{F}_a(k; T^{(1)}_m) - F_a(k; T^{(2)}_m)| \xrightarrow{p} 0~,$$
and
$$|\hat{F}(k; T^{(1)}_m) - F(k; T^{(2)}_m)| \xrightarrow{p} 0~.$$
\end{lemma}
\begin{proof}
We prove the first statement for $a = 1$, and the rest of the results follow similarly. We want to show that
$$\left|\frac{1}{n}\sum_{i=1}^n{\bf 1}\{S_i(T^{(1)}_m)= k, A_i(T^{(1)}_m) = 0\} - (1-{\pi}(k; T^{(2)}_m))p(k;T^{(2)}_m)\right| \xrightarrow{p} 0~.$$
By the triangle inequality, we bound this above by
$$\left|\frac{1}{n}\sum_{i=1}^n{\bf 1}\{S_i(T^{(1)}_m)= k, A_i(T^{(1)}_m) = 0\} - (1-\pi(k; T^{(1)}_m))p(k;T^{(1)}_m)\right| +$$ 
$$+ \left|(1-\pi(k; T^{(1)}_m))p(k;T^{(1)}_m) - (1 - \pi(k; T^{(2)}_m))p(k;T^{(2)}_m)\right|~.$$
The first line of the above expression converges to zero by Assumption \ref{ass:unif_treat} and the Glivenko-Cantelli Theorem. Next consider the second line: by assumption, we have that $|p(k;T^{(1)}_m) - p(k;T^{(2)}_m)| \rightarrow 0$ and $|\pi(k; T^{(1)}_m) - \pi(k; T^{(2)}_m)| \rightarrow 0$ and hence the second line converges to zero (since all of these quantities are bounded).
\end{proof}

\begin{lemma}\label{lem:N_limit}
Let $T^{(1)}_m$, $T^{(2)}_m$ be defined as in the statement of Lemma \ref{lem:main_2}. Given the Assumptions of Theorem \ref{thm:main}, we have that, for $k = 1,... ,K$,
$$\frac{n(k;T^{(1)}_m)}{n_a(k;T^{(1)}_m)} = \frac{1}{\pi(k; T^{(2)}_m)} + o_P(1)~.$$
\end{lemma}
\begin{proof}
This follows from Assumptions \ref{ass:prop/cell_size}, \ref{ass:unif_treat}, and the Glivenko-Cantelli Theorem.
\end{proof}

\begin{lemma}\label{lem:G_donsker}
Given Assumption \ref{ass:bounded}, the class of functions $\mathcal{G}$ defined as
$$\mathcal{G} := \{G^{k}_a(\cdot\hspace{1mm}; T): T \in \mathcal{T}\}~,$$
for a given $a$ and $k$ is a Donsker class.
\end{lemma}
\begin{proof}
This follows from the discussion of classes of monotone uniformly bounded functions in \cite{van1996NEW}.
\end{proof}

\begin{lemma}\label{lem:G_limit}
Let $T^{(1)}_m$, $T^{(2)}_m$ be defined as in the statement of Lemma \ref{lem:main_2}. Then we have that
$$||G^{k}_{a}(\cdot \hspace{1mm}; T^{(1)}_m) - G^{k}_{a}(\cdot \hspace{1mm}; T^{(2)}_m)|| \rightarrow 0~.$$
\end{lemma}
\begin{proof}
Again we argue along subsequences. For every subsequence we show there exists a further subsequence (indexed by $m_\ell$) such that the result holds. By compactness (Lemmas \ref{lem:rho_complete}, \ref{lem:rho_tb}), there exists a tree $T^* \in \mathcal{T}^*_L$ such that $T^{(2)}_{m_\ell} \rightarrow T^*$, and hence by assumption we also have that $T^{(1)}_{m_\ell} \rightarrow T^*$. By Lemma \ref{lem:G_limit}
\[Z_a^k(t; T^{(j)}_{m_\ell}) \rightarrow Z_a^k(t; T^*)~,\]
for $j = 1, 2$, at every point of continuity $t$ of $Z_a^k(\cdot; T^*)$, where $Z_a^k(\cdot ; T)$ is the CDF of the distribution of $(Y(a) - E[Y(a)|S(X)]) \big | S(X) = k$.  By Lemma 21.2 in \cite{van1998},
\[G_a^k(p; T^{(j)}_{m_\ell}) \rightarrow G_a^k(p; T^*)~,\]
for $j = 1,2$, at every point of continuity $p$ of $G_a^k(p; T^*)$, and hence
\[\left|G_a^k(p; T^{(1)}_{m_\ell}) - G_a^k(p; T^{(2)}_{m_\ell})\right| \rightarrow 0~,\]
at every point of continuity $p$ of $G_a^k(p; T^*)$. Since the quantile functions $G_a^k(p; T)$ are bounded by Assumption \ref{ass:bounded} and $G_a^k(p; T^*)$ has only countably many discontinuities (because it is non-decreasing), it follows by the dominated convergence theorem that
\[||G_a^k(\cdot; T^{(1)}_{m_\ell}) - G_a^k(\cdot; T^{(2)}_{m_\ell})|| \rightarrow 0~,\]
and hence the result follows.
\end{proof}

\begin{lemma}\label{lem:G_limit}
Let $\{T_m\}_m$ be a sequence of trees converging to some tree $T \in \mathcal{T}_L$. Given the Assumptions of Theorem \ref{thm:main}, we have that
\[Z_a^k(t; T_m) \rightarrow Z_a^k(t; T)~,\]
at every point of continuity $t$ of $Z_a^k(\cdot; T)$, where $Z_a^k(\cdot; T)$ is the CDF of the distribution of $(Y(a) - E[Y(a)|S(X)]) \big | S(X) = k$.
\end{lemma}
\begin{proof}
Re-writing, we have that
\[Z^{k}_a(t; T) = \frac{E[{\bf 1}\{Y(a) \le t + E(Y(a)|S(X)=k)\}{\bf 1}\{S(X)= k\}]}{P(S(X) = k)}~,\]
Hence by the triangle inequality, Assumption \ref{ass:prop/cell_size} and a little bit of algebra, we get that
\[|Z^{k}_a(t; T_m) - Z^{k}_a(t; T)| \le \frac{1}{\delta}\left|R_{m1} - R_{m2}\right| + \frac{1}{\delta^2}\left|R_{m3}\right|~,\]
where
\[R_{m1} := E[{\bf 1}\{Y(a) - E(Y(a)|S_m(X) = k) \le t\}{\bf 1}\{S_m(X)= k\}]~,\]
\[R_{m2} := E[{\bf 1}\{Y(a) - E(Y(a)|S(X) = k)\le t\}{\bf 1}\{S(X)= k\}]~,\]
\[R_{m3} := P(S_m(X)=k) - P(S(X)=k)~.\]
$|R_{m3}|$ goes to zero since it is bounded above by $P(S^{-1}_m(k) \Delta S^{-1}(k))$ which converges to zero by the definition of our metric (Definition  \ref{def:rho1}), where $\Delta$ denotes the symmetric difference.  It remains to show that $|R_{m1} - R_{m2}|$ converges to zero. Again by the triangle inequality, 
\[|R_{m1} - R_{m2}| \le |R_{m1} - R_{m4}| + |R_{m4} - R_{m2}|~,\]
where
\[R_{m4} =  E[{\bf 1}\{Y(a) - E(Y(a)|S_m(X)=k) \le t \}{\bf 1}\{S(X) = k\}]~.\]
By another application of the triangle inequality,
\[|R_{m1} - R_{m4}| \le E\left|{\bf 1}\{S_m(X) = k\} - {\bf 1}\{S(X) = k\}\right|~,\]
and this bound is equal to $P(S^{-1}_m(k) \Delta S^{-1}(k))$ which converges to zero by the definition of our metric. It remains to show that $|R_{m4} - R_{m2}|$ converges to zero. To show this, we begin by noting that
\[R_{m4} = E[{\bf 1}\{Y(a) - E(Y(a)|S(X) = k) \le t + E(Y(a)|S_m(X) = k) - E(Y(a)|S(X) = k)\}{\bf 1}\{S(X) = k\}]~.\]
Let $t_m =  t + E(Y(a)|S_m(X) = k) - E(Y(a)|S(X) = k)$, then we see that $R_{m4}$ is the numerator in the expression for $Z_a^k(t_m; T)$ from above. Similarly, $R_{m2}$ is the numerator in our expression for $Z_a^k(t; T)$. By similar arguments to what we have used in this proof, it can be shown that
\[E(Y(a)|S_m(X)=k) \rightarrow E(Y(a)|S(X) = k)~,\] 
(in fact we derive this in the proof of Lemma \ref{lem:V_cont}) so that $t_m \rightarrow t$. Since $t$ is a point of continuity of $Z_a^k(\cdot; T)$, it follows that 
\[Z_a^k(t_m; T) \rightarrow Z_a^k(t; T)~,\]
and from this it follows that $|R_{m4} - R_{m2}|$ converges to zero, so we are done.
\end{proof}

\begin{lemma}\label{lem:double_index}
Let $\{x_{n,m}\}$ be a doubly-indexed sequence of real numbers. If for any strictly increasing indexing $(n_1, m_1) < (n_2, m_2) < ... <  (n_{\ell}, m_{\ell}) < ...$ (where the inequality is to be interpreted componentwise) the sequence $\{x_{n_{\ell}, m_{\ell}}\}$ contains a convergent subsequence which converges to $x$, then $x_{n,m} \rightarrow x$ as $n, m \rightarrow \infty$.
\end{lemma}
\begin{proof}
Suppose not, then there exists some $\epsilon > 0$ such that for any $M \in \mathbb{N}$, we can find $n', m' > M$ such that $|x_{n', m'} - x| > \epsilon$. We use this fact to construct the following sequence: first pick $n_1, m_1 > 1$ such that $|x_{n_1, m_1} - x| > \epsilon$. Next pick $n_2, m_2 > \max(n_1, m_1)$ such that $|x_{n_2, m_2} - x| > \epsilon$. Continue to pick $n_{\ell + 1}, m_{\ell + 1} > \max(n_\ell, m_\ell)$ such that $|x_{n_{\ell+1}, m_{\ell + 1}} - x| > \epsilon$. The resulting sequence $\{x_{n_{\ell}, m _\ell}\}$ satisfies the conditions of the lemma but contains no subsequence converging to $x$. Hence the result follows by contradiction.
\end{proof}



\section{A Theory of Convergence for Stratification Trees}\label{sec:rho}
\begin{remark}\label{rem:Xcont}
For the remainder of this section suppose $X$ is continuously distributed. Modifying the results to include discrete covariates with finite support is straightforward.
\end{remark}
We will define a metric $\rho$ on the space $\mathcal{T}_L$ and study its properties. To define $\rho$, we write it as a product metric between a metric $\rho_1$ on $\mathcal{S}_L$, which we define below, and $\rho_2$ the Euclidean metric on $[0,1]^{K}$. Recall from Remark \ref{rem:quotient} that any permutation of the elements in $[K]$ simply results in a re-labeling of the partition induced by $S(\cdot)$. For this reason we explicitly define the labeling of a tree partition that we will use, which we call the \emph{canonical labeling}:
\begin{definition}\label{def:canon}(The Canonical Labeling)
\begin{itemize}
\item Given a tree partition $\{\mathcal{X}_D, \mathcal{X}_U\}$ of depth one on $\mathcal{X}$, we assign a label of $1$ to $\mathcal{X}_D$ and a label of $2$ to $\mathcal{X}_U$.
\item Given a tree partition $\{\mathcal{X}^{(L-1)}_D, \mathcal{X}^{(L-1)}_U\}$ of depth $L > 1$ on $\mathcal{X}$, we label $\mathcal{X}^{(L-1)}_D$ as a tree partition of depth $L-1$ using the labels $\{1, 2, ..., K/2\}$, and use the remaining labels $\{K/2 +1, ..., K\}$ to label $\mathcal{X}^{(L-1)}_U$ as a tree partition of depth $L-1$.
\item If it is ever the case that a tree partition of depth $L$ can be constructed in two different ways, we specify the partition unambiguously as follows: if the partition can be written as $\{\mathcal{X}_D^{(L-1)}, \mathcal{X}_U^{(L-1)}\}$ with cut $(j, \gamma)$ and $\{\mathcal{X}_D^{'(L-1)},\mathcal{X}_U^{'(L-1)}\}$ with cut $(j', \gamma')$, then we select whichever of these has the smallest pair $(j,\gamma)$ where our ordering is lexicographic. If the cuts $(j, \gamma)$ are equal then we continue this recursively on the subtrees, beginning with the left subtree, until a distinction can be made.
\end{itemize}
\end{definition}
In words, the canonical labeling labels the leaves from ``left-to-right" when the tree is depicted in a tree representation (and the third bullet point is used to break ties whenever multiple such representations are possible). All of our previous examples have been canonically labeled (see Examples \ref{ex:tree_part1}, \ref{ex:tree_part2}). From now on, given some $S \in \mathcal{S}_L$, we will use the the version of $S$ that has been canonically labeled. Let $P_X$ be the measure induced by the distribution of $X$ on $\mathcal{X}$. We are now ready to define our metric $\rho_1(\cdot,\cdot)$ on $\mathcal{S}_L$ as follows:

\begin{definition}\label{def:rho1}
For $S_1, S_2 \in \mathcal{S}_L$, 
$$\rho_1(S_1, S_2) := \sum_{k=1}^{2^L}P_X(S^{-1}_1(k) \Delta S^{-1}_2(k))~.$$
\end{definition}

Where $A \Delta B := A\setminus B \cup B \setminus A$ denotes the symmetric difference of $A$ and $B$ and $S^{-1}(k) := \{x \in \mathcal{X}: S(x) = k\}$. In words, $\rho_1$ computes the Frechet-Nikodym set distance between each component of the partitions $S_1$ and $S_2$ once these have been canonically labeled. Our metric on partitions most closely resembles the metric defined in \cite{leonardi2002}. In that paper, the metric is defined over the set of all countable (measurable) partitions, which complicates the definition relative to what we present here.

That $\rho_1$ is a metric follows from the properties of symmetric differences and Assumption \ref{ass:bounded}. We show under appropriate assumptions that $(\mathcal{S}_L, \rho_1)$ is a complete metric space in Lemma \ref{lem:rho_complete}, and that $(\mathcal{S}_L, \rho_1)$ is totally bounded in Lemma \ref{lem:rho_tb}. Hence $(\mathcal{S}_L, \rho_1)$ is a compact metric space under appropriate assumptions. Combined with the fact that $([0,1]^{2^L},\rho_2)$ is a compact metric space, it follows that $(\mathcal{T}_L, \rho)$ is a compact metric space.

Next we show that $V(\cdot)$ is continuous in our new metric.

\begin{lemma}\label{lem:V_cont}
Given Assumption \ref{ass:bounded}, $V(\cdot)$ is a continuous function in $\rho$.
\end{lemma}
\begin{proof}
We want to show that for a sequence $T_n \rightarrow T$, we have $V(T_n) \rightarrow V(T)$. By definition, $T_n \rightarrow T$ implies $S_n \rightarrow S$ and $\pi_n \rightarrow \pi$ where $T_n = (S_n, \pi_n)$, $T = (S, \pi)$. By the definition of the symmetric difference,
$$|P(S_n(X) = k) - P(S(X) = k)| \le P_X(S^{-1}_n(k) \Delta S^{-1}(k))~,$$
and hence $P(S_n(X) = k) \rightarrow P(S(X) = k)$. It remains to show that $E[f(Y(a))|S_n(X)=k] \rightarrow E[f(Y(a))|S(X)=k]$ for $f(\cdot)$ a continuous function. Re-writing:
$$E[f(Y(a))|S_n(X)=k] = \frac{E[f(Y(a)){\bf 1}\{S_n(X) = k\}]}{P(S_n(X)=k)}~.$$
The denominator converges by the above inequality, for the numerator:
\[\left|E[f(Y(a)){\bf 1}\{S_n(X) = k\}] - E[f(Y(a)){\bf 1}\{S(X) = k\}]\right| \le M'E\left|{\bf 1}\{S_n(X)=k\} - {\bf 1}\{S(X) = k\}\right|~,\]
for some $M' < \infty$, since $f(\cdot)$ is continuous and $Y(a)$ has bounded support. Finally 
\[E\left|{\bf 1}\{S_n(X)=k\} - {\bf 1}\{S(X) = k\}\right| = P_X(S_n^{-1}(k) \Delta S^{-1}(k))~,\]
which converges to zero by definition.
\end{proof}

\begin{lemma}\label{lem:rho_complete}
Given Assumptions \ref{ass:bounded} and \ref{ass:prop/cell_size}, $(\mathcal{S}_L, \rho_1)$ is a complete metric space.
\end{lemma}
\begin{proof}
To simplify the notation, in what follows we normalize $\mathcal{X} = [0,1]^d$. Let $\{S_n\}_n$ be a Cauchy sequence in $\mathcal{S}_L$. We want to show that this sequence converges to a limit $S \in \mathcal{S}_L$. To prove this, we proceed by induction on the depth of the tree. In what follows we use the notation introduced in Definitions \ref{def:tree_part1} and \ref{def:tree_partL}. 

For the base case $L = 1$, we show that if we have a Cauchy sequence of depth one tree partitions on a generic cube $\Gamma_n = \bigtimes_{j = 1}^d [p_{jn}, q_{jn}] \subseteq \mathcal{X}$, given by $\{\Gamma_D(j_n, \gamma_n), \Gamma_U(j_n, \gamma_n)\}$ (recalling the notation for a depth one tree partition given in Definition \ref{def:tree_part1}), then the sequences $\{p_{jn}\}_n$ and $\{q_{jn}\}_n$ converge, and the partition $\{\Gamma_D(j_n, \gamma_n), \Gamma_U(j_n, \gamma_n)\}$ converges to a partition $\{\Gamma_D(j', \gamma'), \Gamma_U(j', \gamma')\}$ on the cube given by $\Gamma =  \bigtimes_{j = 1}^d [\lim_n p_{jn}, \lim_n q_{jn}]$. 

To that end, first we will show that if a sequence of cubes $\{\Gamma_n\}_n = \bigtimes_{j=1}^d [p_{jn},q_{jn}]$ such that $P_X(\Gamma_n) > \delta$ for $\delta > 0$ is Cauchy in the set-difference metric, then the sequences $\{p_{jn}\}_n$, $\{q_{jn}\}_n$ converge, and the sequence of cubes converges to $\Gamma = \bigtimes_{j = 1}^d[\lim_n p_{jn}, \lim_n q_{jn}]$. We show this by arguing that the corresponding sequences $\{p_{jn}\}_n$ and $\{q_{jn}\}_n$ are all Cauchy as sequences in $\mathbb{R}$ and hence convergent. First note that if $\{\Gamma_n\}_n$ is Cauchy with respect to the metric induced by $P_X$, then it is Cauchy with respect to the metric induced by Lebesgue measure $\lambda$ on $\mathcal{X}$, since by Assumption \ref{ass:bounded}, for any measurable set $A$,
$$P_X(A) = \int_{A} f_X d\lambda \ge c\lambda(A)~,$$ 
 for some $c > 0$. Moreover by Assumption \ref{ass:bounded} (bounded density), our assumption that $P_X(\Gamma_n) > \delta$, and the fact that $\mathcal{X} = [0,1]^d$, each $\Gamma_n$ has a minimal side-length. It follows that if $\{\Gamma_n\}_n$ is Cauchy w.r.t to the metric induced by $\lambda$, then each sequence of intervals $\{[p_{jn}, q_{jn}]\}_n$ for $j = 1 ..., d$ is Cauchy w.r.t to the metric induced by Lebesgue measure on $[0,1]$ (which we denote by $\lambda_1$). By the definition of symmetric difference, when $[p_{jn}, q_{jn}] \cap [p_{jn'}, q_{jn'}] \ne \emptyset$ for $n \ne n'$, 
$$\lambda_1([p_{jn}, q_{jn}]\Delta[p_{jn'}, q_{jn'}]) = |q_{jn'} - q_{jn}| + |p_{jn'} - p_{jn}|,$$
and hence it follows that the sequences $\{p_{jn}\}_n$ and $\{q_{jn}\}_n$ are Cauchy as sequences in $\mathbb{R}$, and thus convergent. It follows that $\{[p_{jn}, q_{jn}]\}_n$ converges to $[\lim p_{jn}, \lim q_{jn}]$ in the metric induced by $\lambda_1$, and hence again since $\mathcal{X} = [0,1]^2$ it follows that $\Gamma_n$ converges to $\Gamma$ in the metric induced by $\lambda$. Finally by Assumption \ref{ass:bounded} (bounded density), we obtain that $\Gamma_n$ converges to $\Gamma$ in the metric induced by $P_X$, as desired.

Returning to the base case $L = 1$, the components of our partition can be written as $\Gamma_D(j_n, \gamma_n) = \bigtimes_{j = 1}^d [p_{jn}, q'_{jn}] , \Gamma_U(j_n, \gamma_n) = \bigtimes_{j = 1}^d [p'_{jn}, q_{jn}]$, where
\begin{equation*}
    p'_{jn}=
    \begin{cases}
      p_{jn}, & \text{if}\ j \ne j_n \\
      \gamma_n, & \text{if}\ j = j_n~,
    \end{cases}
  \end{equation*}
 and similarly for $q'_{jn}$ (to be fully precise, the component of each product corresponding to $j = j_n$ is not a closed interval, but from the assumption that $X$ is continuously distributed this distinction is irrelevant and so we ignore it). By the definition of $\rho_1(\cdot, \cdot)$ these sequences of cubes are Cauchy in the set-difference metric, and hence the argument above shows that the sequences  $\{p_{jn}\}_n$, $\{p'_{jn}\}_n$, $\{q_{jn}\}_n$ and $\{q'_{jn}\}_n$ all converge. We claim that this implies that the sequence $\{j_n\}_n$ is eventually constant (call it $j'$), and the sequence $\{\gamma_n\}_n$ converges as well (call the limit $\gamma'$). To see this, suppose that the sequence $\{j_n\}_n$ is eventually constant. Then $p'_{jn} = \gamma_n$ for some $j$ (for $n$ sufficiently large) and hence the sequence $\{\gamma_n\}_n$ converges. Now suppose that the sequence $\{j_n\}_n$ oscillates forever (WLOG say between $j = 1$ and $j = 2$). Then there exists a sub-indexing $\{n_\ell\}$ such that $p'_{1n_\ell} = p_{1n_\ell}$ and $q'_{1n_\ell} = q_{1n_\ell}$. Since all these sequence are convergent, this implies that $\{p'_{1n}\}_n$ and $\{p_{1n}\}_n$ have the same limit (call it $p$), and similarly for $\{q'_{1n}\}_n$ and $\{q_{1n}\}_n$ (call it $q$). There must also exist another sub-indexing (call it $\{n_k\}$) such that $p'_{1n_k} = \gamma_{n_k} = q'_{1n_k}$ and hence it must be the case that $p = q$. However, by Assumption \ref{ass:bounded} and \ref{ass:prop/cell_size}, the sequences $\{p_{1n}\}_n$ and $\{q_{1n}\}_n$ cannot come arbitrarily close to each other, and hence $p < q$, and thus we obtain a contradiction. From this we can conclude that $\Gamma_D(j_n, \gamma_n)$ and $\Gamma_U(j_n, \gamma_n)$ converge to $\Gamma_D(j', \gamma'), \Gamma_U(j', \gamma')$, on the cube $\Gamma =  \bigtimes_{j = 1}^d [\lim_n p_{jn}, \lim_n q_{jn}]$, which establishes the base case.
 
Now for the induction step, suppose it is the case that a Cauchy sequence of depth $L$ tree partitions given by $\{\Gamma_D^{(L -1)}(j_n, \gamma_n), \Gamma_U^{(L-1)}(j_n, \gamma_n)\}$ on the cubes $\Gamma_n = \bigtimes_{j=1}^d [p_{jn},q_{jn}] \subseteq \mathcal{X}$ converges to a depth $L$ tree partition $\{\Gamma_D^{(L-1)}(j', \gamma'), \Gamma_U^{(L-1)}(j', \gamma')\}$ on $\Gamma = \bigtimes_{j=1}^d [\lim_n p_{jn},\lim_n q_{jn}]$.  

Next consider a Cauchy sequence of depth $L + 1$ tree partitions given by $\{\Gamma_D^{(L)}(j_n, \gamma_n), \Gamma_U^{(L)}(j_n, \gamma_n)\}$ on the cubes $\Gamma_n = \bigtimes_{j=1}^d [p_{jn}, q_{jn}] = \{\Gamma_D(j_n, \gamma_n), \Gamma_U(j_n, \gamma_n)\}$. By the definition of $\rho_1(\cdot, \cdot)$, it is immediate that $\Gamma_D^{(L)}(j_n, \gamma_n)$ and $\Gamma_U^{(L)}(j_n, \gamma_n)$ are Cauchy sequences of depth $L$ tree partitions on the cubes $\Gamma_D(j_n, \gamma_n)$ and $\Gamma_U(j_n, \gamma_n)$, respectively. Write $\Gamma_D(j_n, \gamma_n) =  \bigtimes_{j = 1}^d [p_{jn}, q'_{jn}]$ and $\Gamma_U(j_n, \gamma_n) = \bigtimes_{j = 1}^d [p'_{jn}, q_{jn}]$, where we use the notation introduced in the base case. By the induction hypothesis it follows that the sequences $\{p_{jn}\}_n$, $\{p'_{jn}\}_n$, $\{q_{jn}\}_n$ and $\{q'_{jn}\}_n$ converge. This further implies (using the same argument as in the base case) that the sequence $\{j_n\}_n$ is eventually constant (call it $j'$), and the sequence $\{\gamma_n\}_n$ converges as well (call the limit $\gamma'$). Once again invoking the induction hypothesis we can conclude that $\Gamma_D^{(L)}(j_n, \gamma_n)$ and $\Gamma_U^{(L)}(j_n, \gamma_n)$ converge to depth $L$ tree partitions $\Gamma_D^{(L)}(j', \gamma'), \Gamma_U^{(L)}(j', \gamma')$, and these two partitions form a depth $L + 1$ tree partition on $\Gamma = \bigtimes_{j = 1}^d [\lim_n p_{jn}, \lim_n q_{jn}]$. This partition is by construction our desired limit, and hence the result follows by induction. 
\end{proof}

\begin{lemma}\label{lem:rho_tb}
Given Assumption \ref{ass:bounded} $(\mathcal{S}_L, \rho_1)$ is a totally bounded metric space.
\end{lemma}
\begin{proof}
To simplify the notation, in what follows we normalize $\mathcal{X} = [0,1]^d$. Given any measurable set $A$, we have by Assumption \ref{ass:bounded} that
$$P_X(A) = \int_A f_X d\lambda \le C\lambda(A)~,$$
where $\lambda$ is Lebesgue measure, for some constant $C > 0$. The result now follows immediately by constructing the following $\epsilon$-cover: at each depth $L$, consider the set of all trees that can be constructed from the set of splits $\{\frac{\epsilon}{C(2^{2L})}, \frac{2\epsilon}{C(2^{2L})}, ..., 1\}$. By construction any tree in $\mathcal{S}_L$ is at most $\epsilon$ away from some tree in this set.
\end{proof}

\begin{lemma}\label{lem:S_const}
Given Assumptions \ref{ass:bounded}, \ref{ass:prop/cell_size}, and \ref{ass:tree_estimate}. Then the set $\mathcal{T}^*_L$ of maximizers of $V(\cdot)$ exists, and 
$$\inf_{T^*\in\mathcal{T}^*_L} \rho(\hat{T}_m,T^*) \xrightarrow{p} 0~,$$
as $m \rightarrow \infty$.
Furthermore, there exists a sequence of $\sigma\{(W_i)_{i=1}^m\}/\mathcal{B}(\mathcal{T}_L)$-measurable trees $\bar{T}_m \in \mathcal{T}^*_L$ such that
$$\rho(\hat{T}_m,\bar{T}_m) \xrightarrow{p} 0~.$$
\end{lemma}
\begin{proof}
First note that, since $(\mathcal{T}_L,\rho)$ is a compact metric space and $V(\cdot)$ is continuous, we have that $\mathcal{T}^*_L$ exists and is itself compact. Next, fix a subsequence $\{\hat{T}_{m_\ell}\}$ of $\{\hat{T}_m\}$. By Assumption \ref{ass:tree_estimate}, there exists a sub-sub sequence of $\{\hat{T}_{m_\ell}\}$ (which by an abuse of notation we continue to index by $m_\ell$) such that  
\[\left|V(\hat{T}_{m_\ell}) - V^*\right| \xrightarrow{a.s} 0~.\]
Fix an $\epsilon > 0$, and let $$\mathcal{T}^\epsilon_L := \{T \in \mathcal{T}_L: \inf_{T^* \in \mathcal{T}_L^*} \rho(T,T^*) > \epsilon\}~,$$ then it is the case that
$$\inf_{T \in \mathcal{T}^\epsilon_L} V(T) > V^*~.$$
To see why, suppose not and consider a sequence $\{T_k\} \in \mathcal{T}^\epsilon_L$ such that $V(T_k) \rightarrow V^*$. Now by the compactness of $\mathcal{T}_L$, there exists a convergent subsequence $\{T_{k_j}\}$ of $\{T_k\}$, i.e. $T_{k_j} \rightarrow T'$ for some $T' \in \mathcal{T}_L$. By continuity, it is the case that $V(T_{k_j}) \rightarrow V(T')$ and by assumption we have that $V(T_{k_j}) \rightarrow V^*$, so we see that $T' \in \mathcal{T}^*_L$ but this is a contradiction since it implies that there exist trees in $\mathcal{T}_L^{\epsilon}$ which are arbitrarily close to trees in $\mathcal{T}_L^*$.

Hence, for every $\epsilon > 0$, there exists some $\eta > 0$ such that
$$V(T) > V^* + \eta~,$$
for every $T \in \mathcal{T}^\epsilon_L$. Let $\omega$ be any point in the sample space for which we have that $V(\hat{T}_{m_\ell}(\omega)) \rightarrow V^*$, then it must be the case that $\hat{T}_{m_\ell}(\omega) \notin \mathcal{T}^\epsilon_L$ for $m_\ell$ sufficiently large, and hence
$$\inf_{T^*\in\mathcal{T}^*_L} \rho(\hat{T}_{m_\ell},T^*) \xrightarrow{a.s.} 0~.$$ 
Since every subsequence of $\{\hat{T}_m\}$ contains a further subsequence for which the above convergence holds, it follows that 
\[\inf_{T^*\in\mathcal{T}^*_L} \rho(\hat{T}_m,T^*) \xrightarrow{p} 0~.\]
To make our final conclusion, it suffices to note that $\rho(\cdot,\cdot)$ is itself a continuous function and so by the compactness of $\mathcal{T}^*_L$, there exists some sequence of trees $\bar{T}_m$ such that
$$\inf_{T^*\in\mathcal{T}^*_L} \rho(\hat{T}_m,T^*) = \rho(\hat{T}_m,\bar{T}_m)~.$$
Furthermore, by the continuity of $\rho$, the measurability of $\hat{T}_m$, and the compactness of $\mathcal{T}^*_L$, we can ensure the measurability of $\bar{T}_m$ by invoking a measurable selection theorem (see Theorem 18.19 in  \cite{aliprantis1986}).
\end{proof}
\section{Supplementary Results}\label{sec:appendixB}
\subsection{Supplementary Example}\label{sec:intro_example}
In this section we present a result which complements the discussion in the introduction on how stratification can reduce the variance of the difference-in-means estimator. Using the notation from Section \ref{sec:notation}, let $\{Y_i(1), Y_i(0), X_i\}_{i=1}^n$ be i.i.d and let $Y$ be the observed outcome. Let $S: \mathcal{X} \rightarrow [K]$ be a stratification function.  Consider treatments $\{A_i\}_{i=1}^n$ which are assigned via stratified block randomization using $S$, with a target proportion of $0.5$ in each stratum (see Example \ref{ex:sbr_treat} for a definition). Finally, let
$$\hat{\theta} = \frac{1}{n_1}\sum_{i=1}^{n}Y_iA_i - \frac{1}{n - n_1}\sum_{i=1}^nY_i(1-A_i)~,$$
where $n_1 = \sum_{i=1}^n {\bf 1}\{A_i = 1\}$. It can be shown using Theorem 4.1 of \cite{bugni2015} that 
$$\sqrt{n}(\hat{\theta} - \theta) \xrightarrow{d} N(0, V)~,$$
with $V = V_Y - V_S$, where $V_Y$ does not depend on $S$ and 
$$V_S := E\left[\left(E[Y(1)|S(X)] + E[Y(0)|S(X)]\right)^2\right]~.$$
In contrast, if treatment is assigned without any stratification, then
$$\sqrt{n}(\hat{\theta} - \theta) \xrightarrow{d} N(0, V')~,$$
with $V' = V_Y - E[Y(1) + Y(0)]^2$.
It follows by Jensen's inequality that $V_S > E[Y(1) + Y(0)]^2$ as long as $E[Y(1) + Y(0) | S(X) = k]$ is not constant for all $k$. Hence we see that stratification lowers the asymptotic variance of the difference in means estimator as long as the outcomes are related to the covariates as described above.

\subsection{``Smart" Pooling}\label{sec:smart_pool}
In this section we propose a ``smart" pooling estimation strategy when the pilot data $(Y_j, A_j, X_j)_{j=1}^m$ comes from a randomized experiment conducted via simple random assignment (Example \ref{ex:iid_treat}) without stratification. Formally, let the assignment procedure be such that $\{A_j\}_{j=1}^m$ is an i.i.d sequence, independent of everything else, such that $P(A_j = 1) = \pi$ for some $\pi \in (0, 1)$. Fix a stratification tree $T \in \mathcal{T}_L$ and consider second-wave data $\{(Y_i, A_i(T), X_i)\}_{i=1}^n$ which satisfies the assumptions of Theorem \ref{thm:main}. Suppose we estimate the ATE by computing $\hat{\theta}(T)$ on \emph{all} of the data simultaneously; call the resulting estimator $\hat{\theta}_{SW}(T)$. Then it can be shown that
\[\sqrt{N}(\hat{\theta}_{SW}(T) - \theta) \xrightarrow{d} N(0, V_\lambda(T))~,\]
where 
\[V_\lambda(T) = \sum_{k=1}^KP(S(X) = k)\left[\left(E[Y(1) - Y(0)|S(X) = k] - E[Y(1) - Y(0)]\right)^2 + \left(\frac{\sigma^2_0(k)}{1 - \pi_\lambda(k)} + \frac{\sigma^2_1(k)}{\pi_\lambda(k)}\right)\right]~,\]
with
\[\pi_\lambda(k) = \lambda\pi + (1 - \lambda)\pi(k)~,\]
and as in Section \ref{sec:pool} we assume that $m/N \rightarrow \lambda$. 
We note that this expression coincides with the asymptotic variance for the pooled IPW estimator considered in \cite{hahn2011}. With this expression in hand, we could construct an empirical analog of $V_\lambda(T)$ and proceed as we have done before. However, to our knowledge the techniques developed in this paper are not sufficient to establish the formal properties of this procedure. In particular, it seems that many of the steps which involve first conditioning on the pilot data (which we use to deal with the non-uniqueness of the optimal tree) can no longer be applied. For this reason, we leave the study of the formal properties of this procedure to future work. 

\subsection{Large Sample Behavior of Cross-Validation}\label{sec:CV_large}
In this section we derive a result about the large sample behavior of $\hat{L}_m$. We show that if each $\hat{T}^{(L, d)}_m$ for $d = 1,2$ is estimated by minimizing the empirical variance over $\mathcal{T}_L$ as described in Section \ref{sec:mainres}, and the pilot experiment is conducted using SBR, then under appropriate assumptions, $\hat{L}$ equals $\bar{L}$ with probability approaching $1$ as $m \rightarrow \infty$:
\begin{proposition}\label{prop:CV_const}
Suppose the pilot data come from a RCT performed using stratified block randomization with assignment fraction $\pi \in (0, 1)$ across all strata, and $\hat{T}^{(L,d)}_m$ for $d = 1,2$ are estimated by minimizing the empirical variance. Suppose $V^*_{\bar{L}} < V^*_L$
for all $L < \bar{L}$. Then under Assumptions \ref{ass:bounded}, \ref{ass:prop/cell_size}, and \ref{ass:cell_grid}, 
\[P(\hat{L} = \bar{L}) \rightarrow 1~,\]
as $m \rightarrow \infty$.
\end{proposition}

Again we note that Proposition \ref{prop:CV_const} \emph{does not} help us conclude that $\hat{T}_m^{CV}$ should perform any better than $\hat{T}_{\bar{L}}$ in \emph{finite} samples. However it does help shed light on the large sample behavior we observe for $\hat{T}_m^{CV}$ in the simulations of Section \ref{sec:simulations}.  

\noindent{\bf Proof of Proposition \ref{prop:CV_const}}
\begin{proof}
First we show that 
\begin{equation}\label{eq:CV_conv}
\sup_{L}\left|\widetilde{V}^{CV}_L - V^*_L\right| \xrightarrow{p} 0~.
\end{equation}
To do this we show 
\[\left|\widetilde{V}^{(1)}(\hat{T}^{(2)}_L) - V^*_L\right| \xrightarrow{p} 0~,\]
and also 
\[\left|\widetilde{V}^{(2)}(\hat{T}^{(1)}_L) - V^*_L\right| \xrightarrow{p} 0~,\]
for every $L$ (where here we exploit the fact that $L$ can take on only finitely many values). We prove the first claim; the second follows identically.
By the triangle inequality,
$$\left|\widetilde{V}^{(1)}(\hat{T}^{(2)}_L) - V^*_L\right| \le \left|\widetilde{V}^{(1)}(\hat{T}^{(2)}_L) - \widetilde{V}^{(2)}(\hat{T}^{(2)}_L)\right| + \left|\widetilde{V}^{(2)}(\hat{T}^{(2)}_L) - V^*_L\right|~.$$
Consider the first term on the RHS. This is bounded above by
$$\sup_{T} \left|\widetilde{V}^{(1)}(T) - \widetilde{V}^{(2)}(T)\right|~,$$
and another application of the triangle inequality yields
$$\sup_{T} \left|\widetilde{V}^{(1)}(T) - \widetilde{V}^{(2)}(T)\right| \le \sup_{T} \left|\widetilde{V}^{(1)}(T) - V(T)\right| + \sup_{T} \left|\widetilde{V}^{(2)}(T) - V(T)\right|~,$$
with both terms converging to 0 in probability by Lemmas \ref{lem:CV_pilot} and Lemma \ref{lem:emp_unif}. Next consider the second term on the RHS. Applying the triangle inequality,
$$\left|\widetilde{V}^{(2)}(\hat{T}^{(2)}_L) - V^*_L\right| \le \left|\widetilde{V}^{(2)}(\hat{T}^{(2)}_L) - V(\hat{T}^{(2)}_L)\right| + \left|V(\hat{T}^{(2)}_L) - V^*_L\right|~.$$
Here the first term converges in probability to zero by the same argument we have made above, and the second term converges in probability to zero by Lemma \ref{lem:CV_pilot} and Proposition \ref{prop:emp_const}.

Now that we have established (\ref{eq:CV_conv}), we prove the statement of the proposition. To that end, note that 
\[\widetilde{V}^{CV}_{\hat{L}} \le \widetilde{V}^{CV}_{\bar{L}} = V^*_{\bar{L}} + o_P(1)~,\]
where the inequality follows by definition and the equality by the fact that $\widetilde{V}^{CV}_{\bar{L}} \xrightarrow{p} V^*_{\bar{L}}$. It is thus the case that 
\[V^*_{\hat{L}} - V^*_{\bar{L}} \le V^*_{\hat{L}} - \widetilde{V}^{CV}_{\hat{L}} - o_P(1) \le \sup_L\left|V^*_{L} - \widetilde{V}^{CV}_L\right| - o_P(1) \xrightarrow{p} 0~,\]
 by (\ref{eq:CV_conv}). Since $V^*_{\bar{L}} < V^*_{L}$ for all $L \ne \bar{L}$, there exists an $\eta > 0$ such that 
 \[V^*_{\bar{L}} < V^*_{L} - \eta~,\]
 for all $L \ne \bar{L}$. Hence
 \[P(\hat{L} \ne \bar{L}) \le P(V^*_{\bar{L}} < V^*_{\hat{L}} - \eta) = P(V^*_{\hat{L}} - V^*_{\bar{L}} > \eta) \rightarrow 0~,\]
 so that 
 \[P(\hat{L} = \bar{L}) \rightarrow 1~,\]
 as desired.
\end{proof}

\begin{lemma}\label{lem:CV_pilot}
Let $\mathcal{D}_\ell$, $\ell = 1, 2,$ denote the two folds of the pilot sample for the cross-validation procedure described in Section \ref{sec:extensions}. Using the notation of Proposition \ref{prop:emp_const}, suppose the assignment procedure in the pilot is given by SBR with equal assignment proportion $\pi$ in each stratum. Let $m^{(\ell)}(z; \zeta)$ and $m_1^{(\ell)}(z;\zeta)$  denote the total number of observations in fold $\ell$ and stratum $z$, and the total number of treated observations in fold $\ell$ and stratum $z$, respectively. Then
\[\frac{m_1^{(\ell)}(z;\zeta)}{m^{(\ell)}(z;\zeta)} \xrightarrow{p} \pi~,\]
for $\ell = 1, 2$ and all $z$. 
\end{lemma}
\begin{proof}
We prove the result for $\ell = 1$, since the  $\ell = 2$ case follows identically. In what follows, we implicitly condition throughout on the event that $m_a(z;\zeta) \ge 2$ in order to guarantee that we do not divide by zero. It can be shown that this event occurs with probability approaching 1. Let $C_j \in \{0, 1\}$ be a binary variable which denotes whether or not observation $j$ is in fold $1$. Then our goal is to show that 
\[\frac{m_1^{(\ell)}(z;\zeta)}{m^{(\ell)}(z;\zeta)} = \frac{\sum_{j = 1}^mA_j{\bf 1}\{\zeta_j = z\}C_j}{\sum_{j=1}^m{\bf 1}\{\zeta_j = z\}C_j} \xrightarrow{p} \pi~.\]
We proceed by arguing conditionally on $\{C_j\}_{j =1}^m$ and $\{\zeta_j\}_{j = 1}^m$. Note that conditionally on $\{C_j\}_{j =1}^m$ and $\{\zeta_j\}_{j = 1}^m$, SBR is effectively drawing a sample of size $m_a(z;\zeta)$ without replacement from the set $\{C_j: \zeta_j = z\}$ for every $z$. We can thus apply standard results from the literature on survey sampling from a finite population. By Theorem 2.1 in \cite{cochran2007}, 
\[E\left[\frac{1}{m_a(z; \zeta)}\sum_{j=1}^mA_j{\bf 1}\{\zeta_j = z\}C_j \Big| \{C_j\}_{j =1}^m, \{\zeta_j\}_{j = 1}^m\right] = \frac{1}{m(z; \zeta)}\sum_{j=1}^m{\bf 1}\{\zeta_j = z\}C_j~,\]
hence
\[E\left[\frac{\sum_{j = 1}^mA_j{\bf 1}\{\zeta_j = z\}C_j}{\sum_{j=1}^m{\bf 1}\{\zeta_j = z\}C_j} \Big| \{C_j\}_{j =1}^m, \{\zeta_j\}_{j = 1}^m \right] =E\left[\frac{m_a(z; \zeta)}{\sum_{j=1}^m{\bf 1}\{\zeta_j = z\}C_j}\frac{1}{m_a(z; \zeta)}\sum_{j=1}^mA_j{\bf 1}\{\zeta_j = z\}C_j \Big| \{C_j\}_{j =1}^m, \{\zeta_j\}_{j = 1}^m\right]  \]
\[ = \frac{\lfloor \pi m(z;\zeta) \rfloor}{m(z;\zeta)}~,\]
where, importantly, we have used the fact that $m_a(z;\zeta) = \lfloor \pi m(z;\zeta) \rfloor$. Let $\xi_m =  \frac{\lfloor \pi m(z;\zeta) \rfloor}{m(z;\zeta)} - \pi$. By the law of iterated expectations
\[E\left[\left|\xi_m\right|\right] = E\left[\left|\xi_m\right| \Big| \left|\xi_m\right|>\frac{\epsilon}{2}\right]P\left(\left|\xi_m\right| > \frac{\epsilon}{2}\right) +E\left[\left|\xi_m\right| \Big| \left|\xi_m\right|\le \frac{\epsilon}{2}\right]P\left(\left|\xi_m\right| \le \frac{\epsilon}{2}\right)~.\]  Conditional on the event that $m_a(z;\zeta) \ge 2$, $\xi_m \le 2\pi$ for all $m$, hence we have that
\[E\left[\left|\xi_m\right|\right] \le 2\pi P\left(\left|\xi_m\right| > \frac{\epsilon}{2}\right) + \frac{\epsilon}{2}~.\] 
Since $\xi_m \xrightarrow{p} 0$, it follows that $E\left[\left|\xi_m\right|\right] \rightarrow 0$ and hence
\[E\left[\frac{\sum_{j = 1}^mA_j{\bf 1}\{\zeta_j = z\}C_j}{\sum_{j=1}^m{\bf 1}\{\zeta_j = z\}C_j}\right] \rightarrow \pi~.\]
It thus suffices to show that the variance converges to zero. To that end, by the law of total variance,
\[Var\left(\frac{\sum_{j = 1}^mA_j{\bf 1}\{\zeta_j = z\}C_j}{\sum_{j=1}^m{\bf 1}\{\zeta_j = z\}C_j}\right) = E\left[Var\left(\frac{\sum_{j = 1}^mA_j{\bf 1}\{\zeta_j = z\}C_j}{\sum_{j=1}^m{\bf 1}\{\zeta_j = z\}C_j} \Big| \{C_j\}_{j =1}^m, \{\zeta_j\}_{j = 1}^m\right)\right] + Var\left( \frac{\lfloor \pi m(z;\zeta) \rfloor}{m(z;\zeta)}\right)~.\]
The second term on the RHS converges to zero by the same argument as we used for the expectation above. By Theorem 2.2 in \cite{cochran2007},
\[Var\left(\frac{1}{m_a(z; \zeta)}\sum_{j=1}^mA_j{\bf 1}\{\zeta_j = z\}C_j \Big| \{C_j\}_{j =1}^m, \{\zeta_j\}_{j = 1}^m\right) = \frac{\sum_{j=1}^m(C_j - \bar{C}_z)^2{\bf 1}\{\zeta_j = z\}}{m(z;
\zeta)(m(z;\zeta)-1)}\frac{m(z;\zeta)}{m_a(z;\zeta)}\left(1 - \frac{m_a(z;\zeta)}{m(z;\zeta)}\right)~,\]
where $\bar{C}_z = \frac{1}{m(z;\zeta)}\sum_{j=1}^mC_j{\bf 1}\{\zeta_j = z\}$. Hence
\[Var\left(\frac{\sum_{j = 1}^mA_j{\bf 1}\{\zeta_j = z\}C_j}{\sum_{j=1}^m{\bf 1}\{\zeta_j = z\}C_j}  \Big| \{C_j\}_{j =1}^m, \{\zeta_j\}_{j = 1}^m\right) = Var\left(\frac{m_a(z; \zeta)}{\sum_{j=1}^m{\bf 1}\{\zeta_j = z\}C_j}\frac{1}{m_a(z; \zeta)}\sum_{j=1}^mA_j{\bf 1}\{\zeta_j = z\}C_j  \Big| \{C_j\}_{j =1}^m, \{\zeta_j\}_{j = 1}^m\right)\]
\[ = \frac{\sum_{j=1}^m(C_j - \bar{C}_z)^2{\bf 1}\{\zeta_j = z\}}{m(z;
\zeta)(m(z;\zeta)-1)\bar{C}_z^2}\left(\frac{m_a(z;\zeta)}{m(z;\zeta)}\right)\left(1 - \frac{m_a(z;\zeta)}{m(z;\zeta)}\right)~.\]
First we show this is bounded (conditional on the event that $m_a(z;\zeta) \ge 2$). The second and third terms are bounded by $1$. For the first term, note that 
\[\frac{\sum_{j=1}^m(C_j - \bar{C}_z)^2{\bf 1}\{\zeta_j = z\}}{m(z;
\zeta)(m(z;\zeta)-1)\bar{C}_z^2} \le \frac{\sum_{j=1}^mC_j^2{\bf 1}\{\zeta_j = z\}}{m(z;
\zeta)(m(z;\zeta)-1)\bar{C}_z^2} = \frac{m(z;\zeta)}{m(z;\zeta) - 1}\frac{\sum_{j=1}^mC_j{\bf 1}\{\zeta_j = z\}}{\left(\sum_{j=1}^mC_j{\bf 1}\{\zeta_j = z\}\right)^2} \le 2~.\]
Next, we note that first term converges in probability to zero, and the next two terms are bounded by $1$. It thus follows from the same argument as we used for the expectation above that 
\[E\left[Var\left(\frac{\sum_{j = 1}^mA_j{\bf 1}\{\zeta_j = z\}C_j}{\sum_{j=1}^m{\bf 1}\{\zeta_j = z\}C_j} \Big|  \{C_j\}_{j =1}^m, \{\zeta_j\}_{j = 1}^m \right)\right] \rightarrow 0~,\]
and hence that
\[Var\left(\frac{\sum_{j = 1}^mA_j{\bf 1}\{\zeta_j = z\}C_j}{\sum_{j=1}^m{\bf 1}\{\zeta_j = z\}C_j}\right)\rightarrow 0~,\]
as desired.
\end{proof}

\subsection{Supplemental Simulation Exercises}\label{sec:supp_sim}

\subsubsection{Changing the Number of CV Folds}
In this section we repeat the simulation exercise of Section \ref{sec:simulations} with alternative choices for the number of folds in the $B$-fold cross validation procedure. Tables \ref{tab:model1CV}, \ref{tab:model2CV}, \ref{tab:model3CV} present the results. We find that, at least for these designs, $2$-fold cross validation seems to outperform both $3$ and $5$-fold cross validation for most specifications.

\begin{table}[p]
  \centering
    \begin{tabular}{ccccccc}
       \toprule
    \multicolumn{2}{c}{Sample Size} & \multirow{2}[4]{*}{Randomization Procedure} & \multicolumn{4}{c}{Criteria} \\
\cmidrule{1-2}\cmidrule{4-7}    Pilot & Main  &       & Coverage &  $\%\Delta$Length & Power & $\%\Delta$RMSE \\
 \midrule
    \multirow{6}[1]{*}{100} & \multirow{6}[1]{*}{4900} & No Stratification & 95.3  & 0.0   & 77.2  & 0.0 \\
          &       & Strat. Tree & 94.9  & -1.2  & 78.9  & -0.4 \\
          &       & CV Tree 2 & 95.0  & -9.6  & 85.0  & -9.2 \\
          &       & CV Tree 3 & 94.9  & -8.6  & 84.1  & -8.2 \\
          &       & CV Tree 5 & 94.9  & -7.3  & 83.4  & -7.7 \\
          &       & Optimal Tree & 94.9  & -18.7 & 91.3  & -18.6 \\
 \midrule
    \multirow{6}[0]{*}{500} & \multirow{6}[0]{*}{4500} & No Stratification & 95.1  & 0.0   & 78.1  & 0.0 \\
          &       & Strat. Tree & 94.8  & -14.3 & 89.2  & -12.3 \\
          &       & CV Tree 2 & 94.3  & -14.1 & 87.9  & -11.3 \\
          &       & CV Tree 3 & 94.9  & -13.1 & 88.1  & -12.2 \\
          &       & CV Tree 5 & 94.4  & -12.2 & 87.3  & -10.2 \\
          &       & Optimal Tree & 94.5  & -17.6 & 91.1  & -15.5 \\
 \midrule
    \multirow{6}[1]{*}{1500} & \multirow{6}[1]{*}{3500} & No Stratification & 94.3  & 0.0   & 77.0  & 0.0 \\
          &       & Strat. Tree & 95.0  & -14.3 & 88.4  & -14.8 \\
          &       & CV Tree 2 & 94.9  & -13.8 & 88.6  & -15.3 \\
          &       & CV Tree 3 & 94.5  & -13.4 & 87.8  & -13.6 \\
          &       & CV Tree 5 & 95.2  & -12.9 & 87.1  & -14.1 \\
          &       & Optimal Tree & 95.0  & -15.1 & 89.4  & -16.2 \\
    \bottomrule

    \end{tabular}%
   \caption{Simulation Results for Model 1}
  \label{tab:model1CV}%
\end{table}%

\begin{table}[p]
  \centering
    \begin{tabular}{ccccccc}
    \toprule
    \multicolumn{2}{c}{Sample Size} & \multirow{2}[4]{*}{Randomization Procedure} & \multicolumn{4}{c}{Criteria} \\
\cmidrule{1-2}\cmidrule{4-7}    Pilot & Main  &       & Coverage &  $\%\Delta$Length & Power & $\%\Delta$RMSE \\
    \midrule
    \multirow{6}[1]{*}{100} & \multirow{6}[1]{*}{4900} & No Stratification & 95.2  & 0.0   & 56.4  & 0.0 \\
          &       & Strat. Tree & 94.8  & 8.0   & 50.6  & 10.2 \\
          &       & CV Tree 2 & 94.9  & -7.8  & 63.7  & -5.8 \\
          &       & CV Tree 3 & 94.5  & -5.7  & 61.7  & -3.3 \\
          &       & CV Tree 5 & 94.8  & -1.7  & 57.4  & 0.6 \\
          &       & Optimal Tree & 94.8  & -19.1 & 74.5  & -17.6 \\
    \midrule
    \multirow{6}[0]{*}{500} & \multirow{6}[0]{*}{4500} & No Stratification & 94.7  & 0.0   & 56.4  & 0.0 \\
          &       & Strat. Tree & 94.6  & -12.8 & 67.4  & -12.2 \\
          &       & CV Tree 2 & 94.6  & -14.0 & 69.3  & -13.9 \\
          &       & CV Tree 3 & 94.5  & -13.7 & 68.5  & -14.0 \\
          &       & CV Tree 5 & 94.8  & -13.5 & 69.3  & -14.2 \\
          &       & Optimal Tree & 95.0  & -17.5 & 73.1  & -18.0 \\
    \midrule
    \multirow{6}[1]{*}{1500} & \multirow{6}[1]{*}{3500} & No Stratification & 95.0  & 0.0   & 55.5  & 0.0 \\
          &       & Strat. Tree & 94.9  & -12.8 & 68.0  & -12.2 \\
          &       & CV Tree 2 & 94.9  & -12.5 & 67.7  & -12.0 \\
          &       & CV Tree 3 & 94.9  & -12.0 & 67.5  & -11.8 \\
          &       & CV Tree 5 & 95.4  & -11.5 & 67.0  & -11.8 \\
          &       & Optimal Tree & 95.0  & -13.9 & 69.0  & -13.1 \\
    \bottomrule
    \end{tabular}%
\caption{Simulation Results for Model 2}
  \label{tab:model2CV}%
\end{table}%

\begin{table}[htbp]
  \centering
    \begin{tabular}{ccccccc}
    \toprule
    \multicolumn{2}{c}{Sample Size} & \multirow{2}[4]{*}{Stratification Method} & \multicolumn{4}{c}{Criteria} \\
\cmidrule{1-2}\cmidrule{4-7}    Pilot & Main  &       & Coverage &  $\%\Delta$Length & Power & $\%\Delta$RMSE \\
    \midrule
    \multirow{6}[1]{*}{100} & \multirow{6}[1]{*}{4900} & No Stratification & 95.4  & 0.0   & 30.6  & 0.0 \\
          &       & Strat. Tree & 94.7  & 15.5  & 23.9  & 17.8 \\
          &       & CV Tree 2 & 95.4  & -1.4  & 30.9  & -0.8 \\
          &       & CV Tree 3 & 95.2  & -1.5  & 30.7  & -1.9 \\
          &       & CV Tree 5 & 95.0  & -1.7  & 31.7  & -0.7 \\
          &       & Optimal Tree & 95.0  & -7.1  & 34.1  & -8.2 \\
    \midrule
    \multirow{6}[0]{*}{500} & \multirow{6}[0]{*}{4500} & No Stratification & 95.0  & 0.0   & 30.5  & 0.0 \\
          &       & Strat. Tree & 94.8  & -2.2  & 31.3  & -2.2 \\
          &       & CV Tree 2 & 95.0  & -2.8  & 31.0  & -3.6 \\
          &       & CV Tree 3 & 95.2  & -2.7  & 32.0  & -2.3 \\
          &       & CV Tree 5 & 94.9  & -2.5  & 31.9  & -2.7 \\
          &       & Optimal Tree & 94.5  & -6.7  & 34.4  & -5.0 \\
    \midrule
    \multirow{6}[1]{*}{1500} & \multirow{6}[1]{*}{3500} & No Stratification & 94.9  & 0.0   & 31.7  & 0.0 \\
          &       & Strat. Tree & 94.7  & -4.5  & 33.2  & -3.9 \\
          &       & CV Tree 2 & 94.7  & -3.9  & 33.4  & -3.5 \\
          &       & CV Tree 3 & 94.2  & -3.3  & 33.3  & -2.1 \\
          &       & CV Tree 5 & 94.6  & -3.0  & 32.7  & -2.5 \\
          &       & Optimal Tree & 94.3  & -5.6  & 34.3  & -4.3 \\
    \bottomrule
    \end{tabular}%
     \caption{Simulation Results for Model 3}
  \label{tab:model3CV}%
\end{table}%

\subsubsection{Changing the Tree Depth}
In this section we repeat the simulation exercise of Section \ref{sec:simulations}, with alternative choices for the maximum tree depth (we do not report the results for the optimal infeasible tree for the case $L = 5$ since we were not able to reliably compute this tree on a large auxiliary sample). In all cases we consider a pilot/main sample split of 500/4500. Tables \ref{tab:model1Var}, \ref{tab:model2Var}, \ref{tab:model3Var} present the results. We find that, at least for these designs, there are decreasing returns to increasing maximum tree depth. However, in all cases we find that the CV tree outperforms the ad-hoc designs. The fifth column of the table reports the average (estimated) excess variance for the stratification tree and CV tree, as defined in Section \ref{sec:extensions}. As expected, the average excess variance increases with tree depth for the stratification tree, but the CV tree, by trading off estimation error for approximation error, is able to achieve a smaller excess variance for large tree depths.

\begin{table}[htbp]
  \centering
    \begin{tabular}{ccccccr}
    \toprule
    \multirow{2}[4]{*}{Depth} & \multirow{2}[4]{*}{Stratification Method} & \multicolumn{5}{c}{Criteria} \\
\cmidrule{3-7}          &       & Coverage & $\%\Delta$Length & Power & $\%\Delta$RMSE  & \multicolumn{1}{c}{Excess Variance} \\
    \midrule
    \multirow{6}[2]{*}{1} & No Stratification & 94.8  & 0.0   & 77.8  & 0.0   &  \\
          & Ad-Hoc & 95.5  & -3.6  & 80.1  & -5.8  &  \\
          & Ad-Hoc Neyman & 95.1  & -4.2  & 80.5  & -4.6  &  \\
          & Strat. Tree & 95.1  & -9.1  & 84.5  & -10.1 & \multicolumn{1}{c}{0.12} \\
          & CV Tree 2 & 95.2  & -9.1  & 84.6  & -10.3 & \multicolumn{1}{c}{0.14} \\
          & Optimal Tree & 94.6  & -9.8  & 85.3  & -9.4  &  \\
    \midrule
    \multirow{6}[2]{*}{2} & No Stratification & 95.2  & 0.0   & 77.7  & 0.0   &  \\
          & Ad-Hoc & 95.0  & -5.6  & 81.9  & -5.3  &  \\
          & Ad-Hoc Neyman & 94.7  & -6.9  & 82.0  & -6.2  &  \\
          & Strat. Tree & 94.6  & -14.1 & 88.1  & -13.4 & \multicolumn{1}{c}{0.27} \\
          & CV Tree 2 & 95.0  & -13.1 & 87.7  & -12.7 & \multicolumn{1}{c}{0.44} \\
          & Optimal Tree & 95.0  & -15.6 & 89.5  & -15.6 &  \\
    \midrule
    \multirow{6}[2]{*}{3} & No Stratification & 95.1  & 0.0   & 78.1  & 0.0   &  \\
          & Ad-Hoc & 94.8  & -6.9  & 83.6  & -6.2  &  \\
          & Ad-Hoc Neyman & 94.5  & -8.9  & 85.4  & -7.6  &  \\
          & Strat. Tree & 94.8  & -14.3 & 89.2  & -12.3 & \multicolumn{1}{c}{0.60} \\
          & CV Tree 2 & 94.3  & -14.1 & 87.9  & -11.3 & \multicolumn{1}{c}{0.62} \\
          & Optimal Tree & 94.5  & -17.6 & 91.1  & -15.5 &  \\
    \midrule
    \multirow{6}[2]{*}{4} & No Stratification & 95.0  & 0.0   & 77.2  & 0.0   &  \\
          & Ad-Hoc & 94.9  & -7.9  & 83.5  & -8.1  &  \\
          & Ad-Hoc Neyman & 94.7  & -10.5 & 85.5  & -10.7 &  \\
          & Strat. Tree & 95.2  & -13.1 & 87.6  & -13.6 & \multicolumn{1}{c}{0.96} \\
          & CV Tree 2 & 94.5  & -14.9 & 88.6  & -14.4 & \multicolumn{1}{c}{0.63} \\
          & Optimal Tree & 94.6  & -18.5 & 91.0  & -17.8 &  \\
    \midrule
    \multirow{6}[2]{*}{5} & No Stratification & 94.7  & 0.0   & 77.1  & 0.0   &  \\
          & Ad-Hoc & 95.3  & -8.6  & 84.2  & -10.2 &  \\
          & Ad-Hoc Neyman & 94.3  & -11.5 & 85.2  & -11.0 &  \\
          & Strat. Tree & 95.2  & -11.8 & 85.6  & -11.6 &  \\
          & CV Tree 2 & 94.9  & -15.6 & 88.7  & -16.5 &  \\
          & Optimal Tree & -     & -     &       & -     &  \\
    \bottomrule
    \end{tabular}%
  \caption{Simulation Results for Model 1}
 \label{tab:model1Var}
\end{table}%

\begin{table}[htbp]
  \centering
    \begin{tabular}{ccccccr}
    \toprule
     \multirow{2}[4]{*}{Depth} & \multirow{2}[4]{*}{Stratification Method} & \multicolumn{5}{c}{Criteria} \\
\cmidrule{3-7}          &       & Coverage & $\%\Delta$Length & Power & $\%\Delta$RMSE  & \multicolumn{1}{c}{Excess Variance} \\
    \midrule
    \multirow{6}[2]{*}{1} & No Stratification & 95.0  & 0.0   & 56.8  & 0.0   &  \\
          & Ad-Hoc & 95.0  & -0.8  & 57.4  & -1.2  &  \\
          & Ad-Hoc Neyman & 95.5  & -0.7  & 58.7  & -0.2  &  \\
          & Strat. Tree & 94.9  & -13.5 & 70.4  & -13.1 & \multicolumn{1}{c}{0.08} \\
          & CV Tree 2  & 94.8  & -13.4 & 68.8  & -12.5 & \multicolumn{1}{c}{0.08} \\
          & Optimal Tree & 94.6  & -14.0 & 69.3  & -13.3 &  \\
    \midrule
    \multirow{6}[2]{*}{2} & No Stratification & 94.7  & 0.0   & 56.9  & 0.0   &  \\
          & Ad-Hoc & 95.1  & -1.6  & 58.0  & -2.4  &  \\
          & Ad-Hoc Neyman & 94.9  & -1.3  & 57.8  & -0.6  &  \\
          & Strat. Tree & 95.1  & -14.6 & 71.2  & -15.6 & \multicolumn{1}{c}{0.29} \\
          & CV Tree 2 & 95.4  & -13.9 & 69.9  & -15.2 &  \multicolumn{1}{c}{0.38}\\
          & Optimal Tree & 95.2  & -16.8 & 73.0  & -17.6 &  \\
    \midrule
    \multirow{6}[2]{*}{3} & No Stratification & 94.7  & 0.0   & 56.4  & 0.0   &  \\
          & Ad-Hoc & 95.0  & -2.1  & 58.7  & -3.6  &  \\
          & Ad-Hoc Neyman & 94.9  & -1.4  & 57.0  & -2.4  &  \\
          & Strat. Tree & 94.6  & -12.8 & 67.4  & -12.2 & \multicolumn{1}{c}{0.61} \\
          & CV Tree 2  & 94.6  & -14.0 & 69.3  & -13.9 & \multicolumn{1}{c}{0.45} \\
          & Optimal Tree & 95.0  & -17.5 & 73.1  & -18.0 &  \\
    \midrule
    \multirow{6}[2]{*}{4} & No Stratification & 95.4  & 0.0   & 58.3  & 0.0   &  \\
          & Ad-Hoc & 95.1  & -2.8  & 60.6  & -2.3  &  \\
          & Ad-Hoc Neyman & 94.6  & -1.7  & 57.7  & -0.1  &  \\
          & Strat. Tree & 94.5  & -10.6 & 66.2  & -9.0  & \multicolumn{1}{c}{0.94} \\
          & CV Tree 2 & 94.6  & -14.2 & 69.3  & -13.0 & \multicolumn{1}{c}{0.45} \\
          & Optimal Tree & 95.3  & -17.8 & 73.2  & -17.5 &  \\
    \midrule
    \multirow{6}[2]{*}{5} & No Stratification & 95.0  & 0.0   & 56.5  & 0.0   &  \\
          & Ad-Hoc & 94.7  & -3.2  & 59.6  & -2.8  &  \\
          & Ad-Hoc Neyman & 95.2  & -1.6  & 57.3  & -2.1  &  \\
          & Strat. Tree & 94.8  & -8.7  & 64.8  & -8.8  &  \\
          & CV Tree 2  & 95.1  & -14.4 & 68.4  & -14.2 &  \\
          & Optimal Tree & -     & -     &       & -     &  \\
    \bottomrule
    \end{tabular}%
  \caption{Simulation Results for Model 2}
 \label{tab:model2Var}
\end{table}%

\begin{table}[htbp]
  \centering
    \begin{tabular}{ccccccr}
    \toprule
     \multirow{2}[4]{*}{Depth} & \multirow{2}[4]{*}{Stratification Method} & \multicolumn{5}{c}{Criteria} \\
\cmidrule{3-7}          &       & Coverage & $\%\Delta$Length & Power & $\%\Delta$RMSE  & \multicolumn{1}{c}{Excess Variance} \\
    \midrule
    \multirow{6}[2]{*}{1} & No Stratification & 94.9  & 0.0   & 31.1  & 0.0   &  \\
          & Ad-Hoc & 94.9  & -0.7  & 32.0  & 0.3   &  \\
          & Ad-Hoc Neyman & 95.3  & -0.6  & 32.1  & -0.3  &  \\
          & Strat. Tree & 95.0  & -1.5  & 31.6  & -1.1  & \multicolumn{1}{c}{0.36} \\
          & CV Tree 2 & 95.0  & -1.3  & 32.0  & -0.4  & \multicolumn{1}{c}{0.45} \\
          & Optimal Tree & 95.3  & -2.0  & 31.5  & -2.3  &  \\
    \midrule
    \multirow{6}[2]{*}{2} & No Stratification & 94.7  & 0.0   & 31.4  & 0.0   &  \\
          & Ad-Hoc & 94.9  & -1.4  & 31.0  & -1.9  &  \\
          & Ad-Hoc Neyman & 95.3  & -1.1  & 32.0  & -2.7  &  \\
          & Strat. Tree & 95.1  & -2.4  & 32.5  & -3.9  & \multicolumn{1}{c}{1.22} \\
          & CV Tree 2  & 95.0  & -2.2  & 32.3  & -4.5  & \multicolumn{1}{c}{1.37} \\
          & Optimal Tree & 95.2  & -4.2  & 33.3  & -5.1  &  \\
    \midrule
    \multirow{6}[2]{*}{3} & No Stratification & 95.0  & 0.0   & 30.5  & 0.0   &  \\
          & Ad-Hoc & 94.6  & -2.1  & 31.4  & -1.2  &  \\
          & Ad-Hoc Neyman & 95.3  & -1.4  & 31.2  & -2.1  &  \\
          & Strat. Tree & 94.8  & -2.2  & 31.3  & -2.2  & \multicolumn{1}{c}{3.12} \\
          & CV Tree 2  & 95.0  & -2.8  & 31.0  & -3.8  & \multicolumn{1}{c}{2.71} \\
          & Optimal Tree & 94.5  & -6.7  & 34.4  & -5.2  &  \\
    \midrule
    \multirow{6}[2]{*}{4} & No Stratification & 94.8  & 0.0   & 31.2  & 0.0   &  \\
          & Ad-Hoc & 94.9  & -2.6  & 32.2  & -3.0  &  \\
          & Ad-Hoc Neyman & 94.3  & -1.4  & 32.0  & -0.8  &  \\
          & Strat. Tree & 94.7  & 0.4   & 30.2  & 0.6   & \multicolumn{1}{c}{5.81} \\
          & CV Tree 2  & 95.3  & -3.3  & 32.4  & -4.5  & \multicolumn{1}{c}{3.15} \\
          & Optimal Tree & 95.3  & -7.9  & 35.6  & -8.6  &  \\
    \midrule
    \multirow{6}[2]{*}{5} & No Stratification & 94.7  & 0.0   & 31.0  & 0.0   &  \\
          & Ad-Hoc & 95.3  & -3.1  & 31.6  & -4.7  &  \\
          & Ad-Hoc Neyman & 94.9  & -1.7  & 30.6  & -1.5  &  \\
          & Strat. Tree & 94.7  & 2.6   & 29.2  & 3.6   &  \\
          & CV Tree 2 & 95.1  & -3.9  & 32.3  & -3.7  &  \\
          & Optimal Tree & -     & -     & -     & -     &  \\
    \bottomrule
    \end{tabular}%
  \caption{Simulation Results for Model 3}
 \label{tab:model3Var}
\end{table}%

\subsubsection{Application-Based Simulation}
To further assess the potential gains from stratification in the application, we repeat the simulation exercise of Section \ref{sec:simulations} with an application-based simulation design.  To generate the data, we draw observations from the entire dataset with replacement, and impute the missing potential outcome for each observation using nearest-neighbour matching on the Euclidean distance between the (scaled) covariates. We perform the simulations with a sample size of $30,000$, which corresponds approximately to the total number of observations in the dataset. To reproduce the empirical setting, we conduct the experiment in two waves, with sample sizes of $12,000$ and $18,000$ in each wave, respectively.  In all cases, when we stratify we consider a maximum of $4$ strata, which corresponds to the number of strata in Figure \ref{fig:app_kw}, and use SBR to perform assignment. We compare the following stratification methods using the same criteria as in Section \ref{sec:simulations}:

 \begin{itemize}[topsep = 1pt]
\item No Stratification: Here we assign treatment to half the sample, with no stratification.
\item Fixed Stratification: Here we use the stratification from Figure \ref{fig:app_kw},  and assign treatment to half the sample in each stratum. 
\item Fixed Neyman: Here we perform the experiment in two waves using the stratification from Figure \ref{fig:app_kw}. In the first wave, we assign individuals to treatment using the fixed stratification, and then use this data to estimate the Neyman allocation. In the second wave we use the estimated Neyman allocation to assign treatment.
\item Stratification Tree: Here we perform the experiment in two waves. In the first wave, we assign individuals to treatment using the Fixed stratification, and then use this data to estimate a stratification tree. In the second wave we use the estimated tree to assign treatment.
\item Cross-Validated Tree: Here we perform the experiment in two waves. In the first wave, we assign individuals to treatment using the Fixed stratification, and then use this data to estimate a stratification tree with depth selected via ($2$-fold) cross-validation. In the second wave we use the cross-validated tree to assign treatment.
\item Cross-Validated Subtree: Here we perform the experiment in two waves. In the first wave, we assign individuals to treatment using the Fixed stratification, and then use this data to estimate a \emph{constrained} stratification tree with depth selected via ($2$-fold) cross-validation, where we constrain the first split in the tree to partition large and small prior donors as in the tree from Figure \ref{fig:app_kw}. In the second wave we use the constrained cross-validated tree to assign treatment.
\item Infeasible Subtree: Here we estimate an infeasible ``optimal" tree by using a large auxiliary sample (see Appendix \ref{sec:compdet}). However, the first split in the tree is constrained as in the Cross-Validated Subtree. In the first wave, we assign individuals to treatment using the Fixed stratification. In the second wave, we assign individuals to treatment using the infeasible tree.
\item Infeasible Optimal Tree: Here we estimate an infeasible ``optimal" tree by using a large auxiliary sample (see Appendix \ref{sec:compdet}). In the first wave, we assign individuals to treatment using the Fixed stratification. In the second wave, we assign individuals to treatment using the infeasible tree.
\end{itemize}

  We perform $9000$ Monte Carlo iterations. Table \ref{tab:apptable} presents the simulation results.
\begin{table*}[h!]
  \centering
    \begin{tabular}{ccccc}
    \toprule
    \multirow{2}[4]{*}{Stratification Method} & \multicolumn{4}{c}{Criteria} \\
\cmidrule{2-5}         & Coverage & $\%\Delta$Length & Power & $\%\Delta$RMSE\\
    \midrule
    No Stratification & 93.5  & 0.0   & 52.1  & 0.0 \\
    Fixed & 93.9  & -0.3  & 53.5  & -0.5 \\
    Fixed Neyman & 93.9  & -2.1  & 53.9  & -2.4 \\
    Strat. Tree & 93.5  & 0.7   & 52.1  & 1.6 \\
    Strat. Tree CV & 94.0  & -1.4  & 52.8  & -2.5 \\
    Sub. Tree CV & 93.8  & -2.4  & 53.7  & -2.4 \\
    Sub. Tree Infeasible & 94.0  & -3.1  & 55.8  & -3.5 \\
    Infeasible Tree & 94.3  & -6.5  & 58.9  & -7.5 \\
    \bottomrule
    \end{tabular}%
 \caption{Simulation Results for Application-Based Simulation}
  \label{tab:apptable}%
\end{table*}%

We see in Table \ref{tab:apptable} that the overall gains from our procedure are small, which as we explained above may not be surprising given the nature of the experiment. The stratification tree performs slightly worse than no stratification, which agrees with the fact that the cross-validation procedure returned a tree of depth one in Figure \ref{fig:app_ft}.  As was the case in the simulations of Section \ref{sec:simulations}, the cross-validated stratification tree protects against overfitting, and seems to perform fairly well relative to the other feasible methods presented. Interestingly, the constrained CV tree and the fixed Neyman tree outperform the unconstrained CV tree. This may be because, as can be seen in Figure \ref{fig:strat_app}, splitting on large and small prior donors is optimal in this setting and hence imposing this as an explicit constraint seems to improve the performance of the estimation procedure.  Finally, we argue that even in settings where the gains from our procedure are modest, implementing our procedure can help discipline the choice of stratification when a clear choice is not apparent.

\section{Computational Details and Supplementary Simulation Details}\label{sec:compdet}
\subsection{Computational Details}
In this section we describe our strategy for computing stratification trees. We are interested in solving the following empirical minimization problem:
$$\widetilde{T}^{EM} \in \arg\min_{T \in \mathcal{T}_L} \widetilde{V}_m(T)~,$$
where $\widetilde{V}_m(\cdot)$ is defined in Section \ref{sec:mainres}.

Finding a globally optimal tree amounts to a discrete optimization problem in a large state space. Because of this, the most common approaches to fit decision trees in statistics and machine learning are greedy: they begin by searching for a single partitioning of the data which minimizes the objective, and once this is found, the processes is repeated recursively on each of the new partitions (\cite{breiman1984}, and \cite{friedman2001} provide a summary of these types of approaches). However, recent advances in optimization research provide techniques which make searching for globally optimal solutions feasible in our setting. 

A very promising method is proposed in \cite{bertsimas2017}, where they describe how to encode decision tree restrictions as mixed integer linear constraints. In the standard classification tree setting, the misclassification objective can be formulated to be linear as well, and hence computing an optimal classification tree can be computed as the solution to a Mixed Integer Linear Program (MILP), which modern solvers can handle very effectively (see \cite{florios2008}, \cite{chen2016},  \cite{mbakop2016},  \cite{kitagawa2017}, \cite{mogstad2017} for some other applications of MILPs in econometrics). Unfortunately, to our knowledge the objective function we consider cannot be formulated as a linear or quadratic objective, and so specialized solvers such as BARON would be required to solve the resulting program. Instead, we implement an evolutionary algorithm (EA) to perform a stochastic search for a global optimum. See \cite{barros2012} for a survey on the use of EAs to fit decision trees.

The algorithm we propose is based on the procedure described in the {\tt evtree} package description given in \cite{grubinger2011}. Algorithm \ref{alg:gen_tree} presents a pseudocode of the algorithm, with additional details of each step of the procedure provided in Section \ref{sec:alg_details}. In words, a ``population" of candidate trees is randomly generated, which we call the ``parents". Next, for each parent in the population, we select one of five \emph{variation operators} at random and apply it to the parent, which produces a new tree which we call its ``child". We then evaluate the objective function for the parent and child, and keep whichever of the two produces a smaller value for the objective. The resulting list of winners then becomes the new population of parents, and the entire procedure repeats iteratively until some termination conditions are met.  The best tree is then returned. If the population does not satisfy the termination conditions after a maximum number of iterations, then the best tree is returned. We describe each of these steps in more detail below.

\begin{algorithm}
\caption{EA for stratification trees}\label{alg:gen_tree}
\begin{algorithmic}[1]

\Procedure{EA}{${\tt pop\_size, \hspace{1mm} term\_cond, \hspace{1mm} max\_itera}$}  
    \State Initialize population of size ${\tt pop\_size}$.
    \For{\# iterations in {\tt max\_itera}}
    	\For{parent in population}
	  \State Compute optimal strata proportions for parent
	  \State Select variation operator ${\tt variation(\cdot)}$ at random
  	  \State Set child $\leftarrow {\tt variation(\text{parent})}$
	  \State Compute optimal strata proportions for child
  	  \State Set parent $\leftarrow {\tt selection(\text{child, parent})}$
  	  \EndFor
	  \If{population satisfies {\tt term\_cond}}
	       \State Return best parent
	  \Else
	       \State Continue
	  \EndIf
    \EndFor
    \State Return best parent
\EndProcedure

\end{algorithmic}
\end{algorithm}

Although we do note prove that this algorithm converges to a global minimum, it is shown in \cite{cerf1995} that similar algorithms will converge to a global minimum in probability, as the number of iterations goes to infinity. In practice, our algorithm converges to the global minimum in simple verified examples, and consistently achieves a lower minimum than a greedy search. Moreover, it reliably converges to the same minimum in repeated runs (that is, with different starting populations) for all of the examples we consider in the paper.

\subsubsection{Tuning Parameter Recommendations for Practitioners}
The proposed algorithm has several potential tuning parameters which affect its performance: 1. population size, 2. termination conditions, 3. variation operator probabilities. In this section we study the impact of each of these parameters on computing the optimal tree, and then provide recommendations for practitioners based on our results. We proceed by first generating a dataset of size 500 using all three models from Section \ref{sec:simulations} and computing the minimum value of our empirical objective in each case. We then vary the tuning parameters of the algorithm, minimize the objective function, and record the percentage gap relative to the ``true" optimum. We repeat this $100$ times for each choice of tuning parameter and report the average percentage gap below. We vary the tuning parameters as follows:

\begin{itemize}
\item {\bf Population Size:} We vary the population size.
\item {\bf Perturbation Percentage:} We vary the probability that the algorithm chooses a ``local" perturbation (split, prune) vs a ``global" perturbation (mutate, crossover). Here $30\%$ means ``$30\%$ chance of performing a global perturbation".
\item {\bf Termination Percentile:} We vary the percentage of the population that must be within a given tolerance of the best tree.
\item {\bf Termination Iterations:} We vary the number of iterations we wait once the best trees are within a given tolerance before terminating.
\end{itemize}

In all cases the default values for the parameters being held fixed are a population size of $500$, $60\%$ perturbation percentage, 50 termination iterations and termination percentile of $5$. Figure \ref{fig:tuning} presents the results. 

\begin{figure}[H]
\hspace{-1cm}
\begin{tikzpicture}
    \begin{axis}[
        xlabel=Population Size,
        ylabel=Average Percentage Gap,
        xmin=0, xmax=550,
        ymin=0, ymax=25,
        xtick={10,50,100,500},
        ytick={5,10,15,20,25,30,40,50,60,70,80,90,100}
        ]
    \addplot[smooth,mark=*,blue] plot coordinates {
        (10, 8)
        (50, 2.8)
        (100, 1.5)
        (500,0.4)
    };
    \addlegendentry{Model 1}

    \addplot[smooth,color=red,mark=x]
        plot coordinates {
            (10, 20.4)
            (50, 12.7)
            (100, 4.6)
            (500, 1.24)
        };
    \addlegendentry{Model 2}
    
        \addplot[smooth,color=black,mark=triangle]
        plot coordinates {
            (10, 16.1)
            (50, 10.4)
            (100, 2.1)
            (500, 0.7)
        };
    \addlegendentry{Model 3}
    
    \end{axis}
\hspace{9cm}

    \begin{axis}[
        xlabel=Perturbation Percentage,
        ylabel=Average Percentage Gap,
        xmin=0, xmax=100,
        ymin=0, ymax=25,
        xtick={10,30,60,90},
        ytick={5,10,15,20,25,30,40,50,60,70,80,90,100}
        ]
    \addplot[smooth,mark=*,blue] plot coordinates {
        (10, 0.43)
        (30, 0.34)
        (60, 0.4)
        (90, 1.3)
    };
    \addlegendentry{Model 1}

    \addplot[smooth,color=red,mark=x]
        plot coordinates {
            (10, 1.82)
            (30, 1.31)
            (60, 1.24)
            (90, 3.6)
        };
    \addlegendentry{Model 2}
    
        \addplot[smooth,color=black,mark=triangle]
        plot coordinates {
            (10, 1.84)
            (30, 0.68)
            (60, 0.7)
            (90, 0.88)
        };
    \addlegendentry{Model 3}
    
    \end{axis}
     
      \end{tikzpicture}
    
\hspace{-1cm}
\begin{tikzpicture}
    \begin{axis}[
        xlabel=Termination Percentile,
        ylabel=Average Percentage Gap,
        xmin=0, xmax=11,
        ymin=0, ymax=25,
        xtick={1,2,5,10},
        ytick={5,10,15,20,25,30,40,50,60,70,80,90,100}
        ]
    \addplot[smooth,mark=*,blue] plot coordinates {
        (1, 0.49)
        (2, 0.39)
        (5, 0.4)
        (10, 0.35)
    };
    \addlegendentry{Model 1}

    \addplot[smooth,color=red,mark=x]
        plot coordinates {
            (1, 1.78)
            (2, 1.33)
            (5, 1.24)
            (10, 0.93)
        };
    \addlegendentry{Model 2}
    
        \addplot[smooth,color=black,mark=triangle]
        plot coordinates {
            (1, 0.49)
            (2, 0.75)
            (5, 0.7)
            (10, 0.69)
        };
    \addlegendentry{Model 3}
    
    \end{axis}
\hspace{9cm}
\begin{axis}[
        xlabel=Termination Iterations,
        ylabel=Average Percentage Gap,
        xmin=0, xmax=55,
        ymin=0, ymax=25,
        xtick={1,10,20,50},
        ytick={5, 10,15,20,25,30,40,50,60,70,80,90,100}
        ]
    \addplot[smooth,mark=*,blue] plot coordinates {
        (1, 0.54)
        (10, 0.47)
        (20, 0.41)
        (50, 0.4)
    };
    \addlegendentry{Model 1}

    \addplot[smooth,color=red,mark=x]
        plot coordinates {
            (1, 1.55)
            (10, 1.48)
            (20, 1.2)
            (50, 1.24)
        };
    \addlegendentry{Model 2}
    
        \addplot[smooth,color=black,mark=triangle]
        plot coordinates {
            (1, 0.8)
            (10, 0.78)
            (20, 0.8)
            (50, 0.7)
        };
    \addlegendentry{Model 3}
    
    \end{axis}

    \end{tikzpicture}
    \caption{Impact of Tuning Parameters on Optimality Gap}\label{fig:tuning}
  \end{figure}  
  
The tuning parameter with the largest impact on performance is the population size parameter. Increasing the population size slows down computation but gives the algorithm more opportunities for exploration. The variation operator probabilities have a smaller impact on performance, but their effect on computation time is hard to predict ex-ante. The termination conditions play a very small role for these specific designs. For this reason, we recommend that practitioners leave the termination conditions and variation operator probabilities fixed at some default value (for example those described below) and consider tuning the population size parameter so that the output of the algorithm stabilizes as the population size is increased (in our experience a population size of $1000$ is often more than adequate).
    
\subsubsection{More Algorithm Details}\label{sec:alg_details}
\noindent {\bf Initialize Population:} We generate a user-defined number of depth 1 stratification trees (typically between 500 and 1000). For each tree, a covariate and a split point is selected at random, and then the optimal proportions are computed for the resulting strata.

 \noindent {\bf Optimal Strata Proportions}: Recall that for a given stratum, the optimal proportion is given by 
$$\pi^* = \frac{\sigma_1}{\sigma_0 + \sigma_1}~,$$
where $\sigma_0$ and $\sigma_1$ are the within-stratum standard deviations for treatments $0$ and $1$. In practice the user must also select an overlap parameter $\nu$ (as required in Assumption \ref{ass:prop/cell_size}). We select $\nu = 0.1$, so that if $\pi^* < 0.1$ then we assign a proportion of $0.1$, and if $\pi^* > 0.9$ then we assign a proportion of $0.9$.

\noindent {\bf Variation Operators:}
One of the following operators is chosen uniformly at random.
\begin{itemize}[topsep = 1pt]
\item \emph{Split}: Takes a tree and returns a new tree that has had one branch split into two new leaves. The operator begins by walking down the tree at random until it finds a leaf. If the leaf is at a depth smaller than $L$, then a random (valid) split occurs. 
\item \emph{Prune}: Takes a tree and returns a new tree that has had two leaves pruned into one leaf. The operator begins by walking down the tree at random until it finds a node whose children are leaves, and destroys those leaves. The optimal proportions are computed for the resulting strata.
\item \emph{Minor Tree Mutation}: Takes a tree and returns a new tree where the splitting value of some internal node is perturbed in such a way that the tree structure is not destroyed. To select the node, it walks down the tree a random number of steps, at random. The optimal proportions are computed for the resulting strata.
\item \emph{Major Tree Mutation}: Takes a tree and returns a new tree where the splitting value and covariate value of some internal node are randomly modified. To select the node, it walks down the tree a random number of steps, at random. This modification may result in a partition which no longer obeys a tree structure, in which case it destroys any subtrees that viloate the tree stucture. The optimal proportions are computed for the resulting strata.
\item \emph{Crossover}: Takes a tree and returns a new tree which is the result of a ``crossover". The new tree is produced by selecting a second tree from the population at random, and replacing a subtree of the original tree with a subtree from this randomly selected candidate. The subtrees are selected by walking down both trees at random. This may result in a partition which no longer obeys a tree structure, in which case it destroys any subtrees that violate the tree structure. The optimal proportions are computed for the resulting strata.
 \end{itemize}
 
 \noindent{\bf Selection:} For each parent-child pair (call these $T_p$ and $T_c$) we evaluate $\widetilde{V}(T_p)$ and $\widetilde{V}(T_c)$ and then keep whichever tree has the lower value. If it is the case that for a given $T$ any stratum has less than a user-defined number of observations per treatment (typically two), we set $\widetilde{V}(T) = \infty$ (this acts as a rough proxy for the minimum cell size parameter $\delta$, as specified in Assumption \ref{ass:prop/cell_size}).
 
\noindent {\bf Termination Conditions:} We terminate if the top 5\% of trees with respect to the objective are within a given tolerance of each other for at least 50 iterations. The best tree is then returned. If the algorithm does not terminate after 5000 iterations, then the best tree is returned.

\subsection{Supplementary Simulation Details}
In this section we provide additional details on our implementation of the simulation study. 

For each design we compute the ATE numerically. For Model 1 we find $ATE_1 = 0.1257$, for Model 2 we find $ATE_2 = 0.0862$ and for Model 3 we find $ATE_3 = 0.121$. To compute the optimal infeasible trees, we use an auxiliary sample of size $300,000$. The infeasible trees we compute are depicted in Figures \ref{fig:strat_m1}, \ref{fig:strat_m2} and \ref{fig:strat_m3} below.

\begin{figure}[H]\center
\begin{tikzpicture}
  [
    grow                    = down,
    level distance          = 5em,
    edge from parent/.style = {draw, -latex},
    sloped,
    level 1/.style={sibling distance=12em,font=\footnotesize},
    level 2/.style={sibling distance=6em,font=\footnotesize},
    level 3/.style={sibling distance=3em,font=\footnotesize},
  ]
  \node [root] {}
  child {node [root] {}
    child { node [root] {}
    	child {node [leaf]{$0.17$}
	  edge from parent node [above] {$x_1 \le 0.4$} }
	child {node [leaf]{$0.19$}
	  edge from parent node [above] {$x_1 > 0.4$} }
      edge from parent node [above] {$x_1 \le 0.48$} }
    child { node [root] {}
    	child {node[leaf]{$0.26$}
	  edge from parent node [above] {$x_1 \le 0.61$} }
	child {node[leaf]{$0.56$}
	  edge from parent node [above] {$x_1 > 0.61$} }
      edge from parent node [above] {$x_1 > 0.48$} }
      edge from parent node [above] {$x_2 \le 0.4$} }
       child {node [root] {}
    child { node [root] {}
    	child {node [leaf]{$0.21$}
	  edge from parent node [above] {$x_2 \le 0.6$} }
	child {node [leaf]{$0.41$}
	  edge from parent node [above] {$x_2 > 0.6$} }
      edge from parent node [above] {$x_1 \le 0.4$} }
    child { node [root] {}
    	child {node[leaf]{$0.36$}
	  edge from parent node [above] {$x_1 \le 0.59$} }
	child {node[leaf]{$0.56$}
	  edge from parent node [above] {$x_1 > 0.59$} }
      edge from parent node [above] {$x_1 > 0.4$} }
      edge from parent node [above] {$x_2 > 0.4$} }
;
      \end{tikzpicture}
\caption{Optimal Infeasible Tree for Model 1}\label{fig:strat_m1}
\end{figure}

\begin{figure}[H]\center
\begin{tikzpicture}
  [
    grow                    = down,
    level distance          = 5em,
    edge from parent/.style = {draw, -latex},
    sloped,
    level 1/.style={sibling distance=12em,font=\footnotesize},
    level 2/.style={sibling distance=6em,font=\footnotesize},
    level 3/.style={sibling distance=3em,font=\footnotesize},
  ]
  \node [root] {}
  child {node [root] {}
    child { node [root] {}
    	child {node [leaf]{$0.2$}
	  edge from parent node [above] {$x_2 \le 0.4$} }
	child {node [leaf]{$0.23$}
	  edge from parent node [above] {$x_2 > 0.4$} }
      edge from parent node [above] {$x_1 \le 0.4$} }
    child { node [root] {}
    	child {node[leaf]{$0.22$}
	  edge from parent node [above] {$x_1 \le 0.45$} }
	child {node[leaf]{$0.22$}
	  edge from parent node [above] {$x_1 > 0.45$} }
      edge from parent node [above] {$x_1 > 0.4$} }
      edge from parent node [above] {$x_1 \le 0.49$} }
       child {node [root] {}
    child { node [root] {}
    	child {node [leaf]{$0.23$}
	  edge from parent node [above] {$x_1 \le 0.54$} }
	child {node [leaf]{$0.23$}
	  edge from parent node [above] {$x_1 > 0.54$} }
      edge from parent node [above] {$x_1 \le 0.6$} }
    child { node [root] {}
    	child {node[leaf]{$0.57$}
	  edge from parent node [above] {$x_1 \le 0.72$} }
	child {node[leaf]{$0.66$}
	  edge from parent node [above] {$x_1 > 0.72$} }
      edge from parent node [above] {$x_1 > 0.6$} }
      edge from parent node [above] {$x_1 > 0.49$} }
;
      \end{tikzpicture}
\caption{Optimal Infeasible Tree for Model 2}\label{fig:strat_m2}
\end{figure}

\begin{figure}[H]\center
\begin{tikzpicture}
  [
    grow                    = down,
    level distance          = 5em,
    edge from parent/.style = {draw, -latex},
    sloped,
    level 1/.style={sibling distance=12em,font=\footnotesize},
    level 2/.style={sibling distance=6em,font=\footnotesize},
    level 3/.style={sibling distance=3em,font=\footnotesize},
  ]
  \node [root] {}
  child {node [root] {}
    child { node [root] {}
    	child {node [leaf]{$0.39$}
	  edge from parent node [above] {$x_3 \le 0.4$} }
	child {node [leaf]{$0.44$}
	  edge from parent node [above] {$x_3 > 0.4$} }
      edge from parent node [above] {$x_2 \le 0.4$} }
    child { node [root] {}
    	child {node[leaf]{$0.43$}
	  edge from parent node [above] {$x_3 \le 0.4$} }
	child {node[leaf]{$0.48$}
	  edge from parent node [above] {$x_3 > 0.4$} }
      edge from parent node [above] {$x_2 > 0.4$} }
      edge from parent node [above] {$x_1 \le 0.4$} }
       child {node [root] {}
    child { node [root] {}
    	child {node [leaf]{$0.43$}
	  edge from parent node [above] {$x_3 \le 0.4$} }
	child {node [leaf]{$0.47$}
	  edge from parent node [above] {$x_3 > 0.4$} }
      edge from parent node [above] {$x_2 \le 0.4$} }
    child { node [root] {}
    	child {node[leaf]{$0.48$}
	  edge from parent node [above] {$x_3 \le 0.4$} }
	child {node[leaf]{$0.49$}
	  edge from parent node [above] {$x_3 > 0.4$} }
      edge from parent node [above] {$x_2 > 0.4$} }
      edge from parent node [above] {$x_1 > 0.4$} }
;
      \end{tikzpicture}
\caption{Optimal Infeasible Tree for Model 3}\label{fig:strat_m3}
\end{figure}
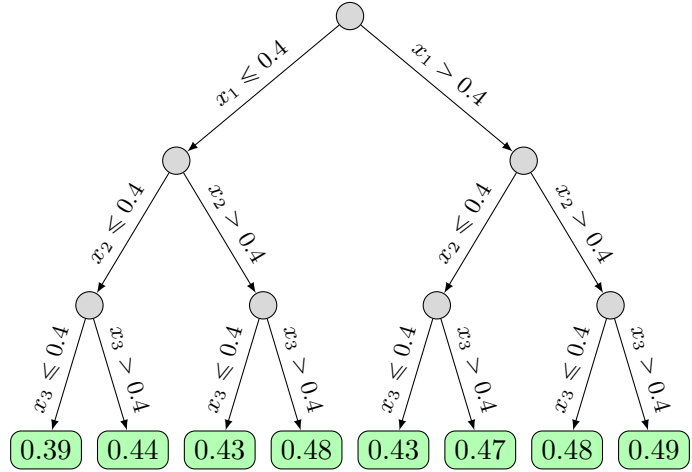

For the application-based design, the ATE is computed to be $0.61$. The infeasible tree we computed is depicted in Figure \ref{fig:strat_app}.

\begin{figure}[H]\center
\begin{tikzpicture}
  [
    grow                    = down,
    level distance          = 7em,
    edge from parent/.style = {draw, -latex},
    sloped,
    level 1/.style={sibling distance=12em,font=\footnotesize},
    level 2/.style={sibling distance=6em,font=\footnotesize},
  ]
  \node [root] {}
    child { node [root] {}
    	child {node [leaf]{$\pi(1) = 0.48$}
	  edge from parent node [above] {${\tt pre\hspace{1mm}gift \le 128}$} }
	child {node [leaf]{$\pi(2) = 0.86$}
	  edge from parent node [above] {${\tt pre\hspace{1mm}gift > 128}$} }
      edge from parent node [above] {${\tt pre\hspace{1mm}gift \le 154}$} }
    child { node [root] {}
    	child {node[leaf]{$\pi(3) = 0.43$}
	  edge from parent node [above] {${\tt total\hspace{1mm}gift \le 1980}$} }
	child {node[leaf]{$\pi(4) = 0.78$}
	  edge from parent node [above] {${\tt total\hspace{1mm}gift > 1980}$} }
      edge from parent node [above] {${\tt pre\hspace{1mm}gift > 154}$} };
      \end{tikzpicture}
\caption{Infeasible Optimal Tree for App.-based Simulation}\label{fig:strat_app}
\end{figure}
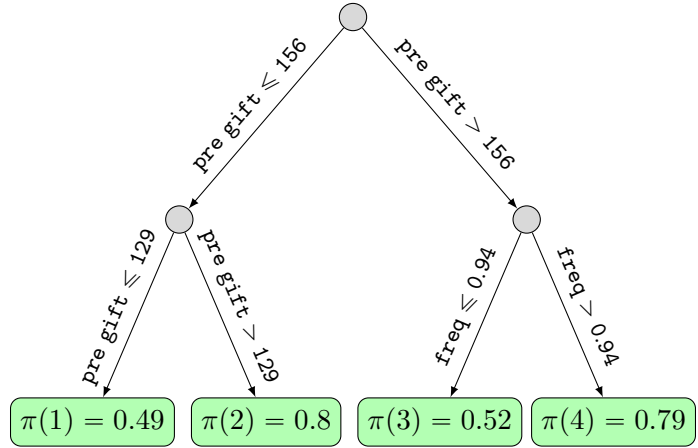

\clearpage	
\end{small}
\nocite{*}
\clearpage

\end{document}